\documentclass[journal]{IEEEtran}
\usepackage{textcomp}
\usepackage{blindtext, graphicx}
\usepackage{subcaption}
\usepackage{multirow}
\usepackage{supertabular}
\usepackage{amsfonts}
\usepackage{amsmath}
\usepackage{float}
\usepackage{placeins}
\usepackage{cite}
\usepackage{enumitem}
\usepackage{eurosym}
\usepackage{booktabs}
\usepackage{threeparttable} 

\usepackage{tikz}
\usepackage{textcomp}
\usetikzlibrary{positioning,shapes}
\usepackage{booktabs,tabularx,makecell}
\usepackage{siunitx}
\newcolumntype{L}[1]{>{\raggedright\arraybackslash}p{#1}}
\newcolumntype{C}[1]{>{\centering\arraybackslash}p{#1}}

\usepackage[skip=0.5pt]{caption} 

\usepackage[utf8]{inputenc}
\usepackage[english]{babel}
\usepackage{amssymb}
\usepackage{flushend}
\usepackage[linesnumbered,ruled,vlined]{algorithm2e}
\SetKw{KwGoal}{Goal}
\SetKwInOut{Input}{Inputs}
\SetKwInOut{Output}{Outputs}

\usepackage{amsmath}
\usepackage{hyperref}
\usepackage{url}

\usepackage[a4paper,margin=0.5in]{geometry} 

\usepackage{array}
\usepackage{comment}
\usepackage{dsfont}
\usepackage{mathtools}

\usepackage{mathrsfs}
\usepackage{adjustbox}
\renewcommand{\arraystretch}{2.0}

\newtheorem{definition}{Definition}[section]

\newtheorem{theorem}{Theorem}
\newtheorem{proof}{Proof}

\newtheorem{proposition}{Proposition}
\usepackage{xcolor}

\newtheorem{lemma}{Lemma}

\newtheorem{remark}{Remark}

\ifCLASSINFOpdf
\else
\fi

\begin{document}

\title{Steganography and Probabilistic Risk Analysis: A Game Theoretical Framework for Quantifying Adversary Advantage and Impact}

\author{Obinna Omego\textsuperscript{*}, Farzana Rahman, Jean-Christophe Nebel\textsuperscript{\textdagger}
\thanks{Obinna Omego, Farzana Rahman, and Jean-Christophe Nebel are with the School of Computer Science and Mathematics, Kingston University London, KT1 2EE, United Kingdom}
\thanks{Corresponding author: Obinna Omego e-mail: a.omego@Kingston.ac.uk}
\thanks{This version is a preprint uploaded to arXiv.}
\thanks{Original arXiv submission: December 2024. Current date of submission: $8^{th}$ August 2026.}}

\markboth{\parbox{\textwidth}{Author manuscript of the article published in PeerJ Computer Science 12:e4011 (2026). DOI: 10.7717/peerj-cs.4011.}}%
{}

  \maketitle

\begin{abstract}
In environments where adversaries engage in active surveillance and covert communication, defenders face the dual challenge of \emph{when} to deploy steganography and \emph{whether} it yields operational benefit. We present a novel game-theoretic model of steganographic operations that captures strategic interactions between a defender and an adversary through calibrated monetary primitives and nonlinear utility mappings. We derive mixed-strategy equilibria that drive conditional and unconditional success rates for hiding and detection, and introduce a time-varying \emph{adversarial advantage} metric that quantifies when an attacker’s incentives temporarily exceed the defender’s detection capacity. By linking this advantage to a new currency-unit risk measure, we extend the classical risk formula into a decision-aware, monetised framework. A Monte-Carlo simulation pipeline embeds payload shifts, detector learning, and scenario uncertainty to deliver distributions of success probabilities and expected losses rather than a static snapshot. Empirical calibration using breach-cost statistics, regulatory fine caps, and expert elicitations supports strategic prescriptions for steganographic deployment, detector investment, and governance trade-offs. Our results provide actionable insight into when steganography strengthens organisational resilience, and when it may yield marginal or negative value.
\end{abstract}


\IEEEpeerreviewmaketitle

\section{Introduction}
As adversaries continuously evolve their tactics to compromise confidential information and network infrastructure, it has become increasingly challenging to develop strategic decision-making processes for security \cite{AndersonSE3, Verizon2024}. Decision making under active adversaries is difficult to model with static checklists or ad-hoc heuristics: attackers adapt, incentives shift, and controls interact in ways that change the likelihood of compromise over time. Quantitative risk practice therefore prescribes separating \emph{likelihood} from \emph{impact} and sizing expected loss accordingly \cite{NIST80030,GordonLoeb2002,Romanosky2016}. At the same time, strategic interaction is central in steganography: defenders choose whether and how to embed, while a warden chooses how aggressively to search for embedded content. Game theory offers a natural framework for guiding these decisions and has already demonstrated that, under explicit modelling assumptions, content-adaptive embedding can surpass non-adaptive baselines. \cite{Schottle2016a,Schottle2013}.

This paper addresses a practical gap between steganographic game models and risk assessment used by decision makers. First, prior work often focuses solely on designing detection metrics and rarely incorporates impact into a risk measure that aligns with established framework \cite{NIST80030,GordonLoeb2002}. Second, while mixed strategies form a critical foundation in these games, the probabilities of different strategies reaching equilibrium are rarely translated into concrete metrics that represent the \emph{overall} success or failure rates for both the defender and the attacker. This data is essential for effectively assessing risk management strategies. Third, there is limited guidance on how to calibrate a deployable effectiveness parameter for the defender based on actual detectors operating on various types of media, including real images. This lack of guidance hinders the practical implementation of theoretical insights in real-world scenarios.

The research contributions of this paper are as follows:
\begin{itemize}
\item This paper proposes a novel two‐player non‐cooperative game‐theoretic model for steganographic operations in surveilled networks. In this model, the defender and the adversary (surveilling warden) both face monetary primitives (costs and benefits) and behavioural utilities that reflect risk attitudes and nonlinearities. The model supports explicit calibration of parameters such as costs of hiding, leak losses, operational harmony gains, and search-costs for the adversary.
\item Second, this study derives mixed‐strategy Nash equilibria for the game, yielding closed-form expressions (or tractable numerical solutions) for equilibrium attack‐search probability and defender mixing. These equilibrium mixes are then used to compute conditional and unconditional success rates for both players, offering a unified metric for steganographic success and detection failure.
\item Third, this study refines the classical quantitative-risk template by introducing a decision-conditioned construction. We introduce a novel metric of adversarial advantage, which captures the excess of the attacker’s incentive over the defender’s time-varying detection effectiveness. We then translate this advantage into a currency-unit risk measure, thereby converting equilibrium mixes into decision-ready, monetised risk quantities. This integration of strategic equilibrium behaviour enables practitioners to move directly from game-theoretic outcomes to actionable assessments of exposure and defence investment.
\item Finally, this paper shows how policy levers and intervention design (e.g., raising adversary search cost via internal controls or legal deterrents; reducing adversary benefit via data‐minimisation; improving detector/hider effectiveness) map seamlessly onto the curvature parameters of our utility‐space model, enabling comparative statics and strategic prescription for stakeholders rather than purely descriptive modelling.
\end{itemize}

The remainder of this paper is organised as follows. Section~\ref{sec:related} reviews background and related work. Section~\ref{Context_and_Background} clarifies the key terminology and context. Section~\ref{Adversary} presents the adversary model. Section~\ref{A_Steganographic_Game-Theoretic_Model} develops the steganographic game-theoretic model. Section~\ref{Game_Analysis} carries out the model analysis. Section~\ref{sec:sensitivity-mne} reports sensitivity analysis of mixed-Nash equilibrium solutions. Section~\ref{Risk_ana} offers the risk analysis framework. Section~\ref{sec:empirical} provides empirical calibration and validation. Section~\ref{sec:comparison-adaptive} compares with adaptive steganography schemes. Finally, Section~\ref{Conclusion} concludes the paper and outlines avenues for future work.

\section{Motivation and Scope}\label{Motivation_for_Security_Game Model}
In environments where targeted surveillance is feasible, organisations must decide \emph{when} to use steganography and \emph{what} security benefit it delivers relative to operational cost. Steganography is effective since this approach conceals the existence of communication rather than just encrypting its content, which reduces the risk of interception, censorship, or pre-emptive blocking \cite{@Fridrich2009, @RainerBohme2010}. While cryptography protects confidentiality after detection, steganography aims to avoid detection altogether—a significant advantage in situations where observable encryption might raise suspicion or prompt selective monitoring. Recent investigations into sophisticated mobile spyware and identity-exposure attacks against cellular subscribers demonstrate that capable adversaries can both identify and exploit valuable communications, underscoring the need for principled, quantitative decision-making rather than ad-hoc controls \cite{AmnestyPegasus2021,Chlosta2021}. In parallel, widely-adopted risk standards define operational risk as a function of the likelihood of adverse events and their consequences; any therefore, defensible method should (i) model strategic attacker–defender behaviour and (ii) connect those behaviours to \emph{likelihood} and \emph{impact} in a transparent way \cite{NIST80030}.

This paper focuses exclusively on \emph{spatial–domain} image steganography in uncompressed 8-bit grayscale images (See Section \ref{sec:empirical} for details). This focus is motivated by both its foundational role in modern steganographic research and its empirical tractability for benchmarking \cite{subramanian2021image}. Spatial-domain methods remain dominant in empirical studies because they offer transparent pixel-level cost modelling and are the basis for evaluating deep-learning-based detectors \cite{kaur2022systematic}. The four content–adaptive methods considered includes \textbf{WOW} (Wavelet Obtained Weights), \textbf{S-UNIWARD} (Spatial UNIversal WAvelet Relative Distortion), \textbf{HILL} (High-pass/Low-pass/Low-pass), and \textbf{MiPOD} (Minimising the Probability of Detection). These schemes constitute the most widely used and rigorously validated baselines in literature \cite{Holub2012WOW,holub2014universal,fridrich2011steganalysis,Sedighi2016MiPOD,li2014new}. On the detection side, we assume a capable warden equipped with two lightweight CNN-based detectors; these choices (GNCNN and Xu-Net–style CNNs) are standard operating points, serving as concrete anchors for calibration and steganalysis performance evaluation \cite{Kodovsky2012Ensemble,xu2016structural, qian2015deep}.

Throughout, we define the impact \(I\) as a monetary (or monetised) loss random variable associated with a confidentiality breach of the protected communication. It subsumes direct response costs, contractual and regulatory exposure, and secondary effects such as reputational harm, tailored to the specific context of the organisation. In our risk equations, \(I\) is multiplied by the \emph{adversary’s advantage} metric (rather than simply a success probability), which is conditioned on equilibrium behaviour (mixed strategies), success rate and the technical effectiveness of the chosen embedding method. This approach preserves the standard likelihood–impact template advocated by quantitative risk frameworks while making the likelihood term explicitly \emph{game-aware} \cite{NIST80030}.

Prior game-theoretic work on adaptive steganography has established that strategic considerations can significantly influence both embedding and detection decisions. However, what remains lacking is an explicit integration of a risk objective (likelihood–impact), a pipeline that turns mixed-strategy success rates into decision-relevant risk, and a deployable path to calibrate steganographic effectiveness using \(\beta_U\) from real detector operating points \cite{Schottle2016a,shi2020cnn}. Our contribution is to connect mixed-strategy equilibria to conditional success rates and subsequently to an explicit, non-negative risk measure. This is enabled by the design of a calibration procedure that maps empirical detector operating points (e.g., CNN at fixed FPR) into the model’s parameters. This creates a policy-legible bridge between steganographic design and quantitative risk management, while remaining faithful to the realities of contemporary steganography and steganalysis.

\section{Background and Related Work}\label{sec:related}
Game theory has provided meaningful insights into adaptive steganography. Early work formalised a two-player interaction in which the encoder adapts embedding to content while a warden optimises detection, revealing conditions under which linear decision rules emerge and adaptive embedding outperforms non-adaptive baselines. However, the model relied on assumptions, which may not be realistic in practical scenarios, such as independent cover symbols and perfect recovery adaptivity \cite{Schottle2016a}. Prior game-theoretic treatments established a more general framework demonstrating that adaptivity can remain secure against informed wardens under specific channel and payload conditions \cite{Schottle2013}. A closely related line of work considers a model under the \emph{independent embedding assumption} in which the warden’s best response converges to a linear aggregation decision rule, thereby reflecting the form of many practical steganalysis detectors, whereas the steganographer’s optimal counter-strategy remains significantly more difficult to characterise \cite{SchottlePascalandLaszkaAronandJohnsonBenjaminandGrossklagsJensandBohme2013}. That analysis sharpened the strategic lens introduced for adaptive embedding by demonstrating when simple detectors suffice against independently embedded payloads; however, its conclusions depend on strong simplifying assumptions such as symbol-independence and full feature-recoverability, and those assumptions have been examined in much greater depth by subsequent work \cite{ker2013moving}. 

Beyond classical content-adaptive encoders, several recent papers have refined how adaptivity is defined and operationalised in realistic channels. First, robust adaptive embedding has been extended to social-network pipelines: the Matching Robust Adaptive Steganography Scheme (MRAS) for JPEG Images over Social Networking Platforms is a scheme where JPEG learns to preserve host statistics through platform transforms (recompression, resizing) while matching perturbed  Discrete Cosine Transform (DCT) distributions, yielding improved extractability under post-processing without a large detectability penalty \cite{MRAS2025}. Second, immune-cover construction explicitly \emph{excludes} cover regions that are systematically exploitable by strong wardens, yielding a warden-aware notion of “holistic” spatial security that complements standard content adaptivity and reduces end-to-end detection risk at fixed payloads \cite{ImmuneCover2024}. User-tunable adaptivity has been explored through configurable controls (e.g., payload, invisibility, and robustness), allowing the instantiation of encoder settings that align with user-specified operating points across various image conditions. \cite{AISM2024}. In \cite{Wang2022ChannelErrors}, channel-aware robustness has been formalised by minimising expected error probability under stochastic post-processing; modelling the channel and embedding with that loss improves recovery after common distortions while maintaining competitive undetectability \cite{Wang2022ChannelErrors}. Moreover, fuzzy-logic-driven selection of embeddable spatial areas has been used to fuse texture/edge saliency into a soft decision surface, improving visual quality and lowering steganalysis success at low–moderate payloads \cite{FuzzyStego2025}.

In parallel, spatial-domain encoders such as WOW, S-UNIWARD, HILL, and MiPOD improved undetectability by minimising content-aware distortion, while detectors have progressed from rich-model ensembles to CNN-based validated on large public benchmarks \cite{Holub2012WOW,holub2014universal,li2014new,Sedighi2016MiPOD,fridrich2012rich,xu2016structural,cogranne2020alaska,luo2024comprehensive}. Together, these papers established a now-standard perspective for analysing steganographic encoders and detectors as best responses within a bimatrix game.

Yet, for practical risk assessment, three critical components remain absent. First, prior game-theoretic treatments on steganography largely do not go further than detection performance and do not integrate \emph{impact} into a quantitative risk metric consistent with established frameworks in information risk management \cite{NIST80030,GordonLoeb2002,Romanosky2016}. Second, while mixed strategies are central to these games, their \emph{equilibrium} mixing probabilities are rarely propagated into end-to-end risk; in practice, decision makers need overall success/failure rates induced by mixed Nash equilibria rather than only detector-level Receiver Operating Characteristic (ROC) summaries. Third, there is limited guidance on how to calibrate steganographic \emph{effectiveness} parameters from deployable detectors on real image data and then feed those calibrated quantities into a game-aware risk computation. This paper addresses these gaps by (i) deriving equilibrium success rates from a corrected mixed Nash equilibrium, (ii) embedding them in a nonnegative, impact-aware risk formulation compatible with quantitative risk practice, and (iii) outlining a deployable calibration path for steganographic effectiveness in the spatial domain (WOW, S-UNIWARD, HILL, MiPOD) with standard steganalysis back-ends.

\section{Terminology}
\label{Context_and_Background}
Game theory provides mathematical models and tools that can be used for investigating strategic decision-making. This section presents game-theoretical concepts and definitions that are pertinent to this study. 

Network security can be seen as a strategic game played between players. These players may be broadly classified as network administrators or authorised users defending the networks and malicious entities who wish to compromise the confidentiality, integrity and availability of systems and networks. The game is played on interconnected, simple and sophisticated systems, where vulnerabilities of assets are tried to be exploited by attacks, and defensive measures constitute its strategic move \cite{Erokhin2023Analysis0}. 

\begin{definition}
\label{definition}
A game $G \in \mathcal G$ is defined as a triple $(\mathcal {P, S, L})$, where $\mathcal P$ represents the set of players, $\mathcal S$ represents the set of strategy profiles, and $\mathcal L$ represents the set of payoff functions. Each payoff function in $\mathcal L$ determines the outcome for a player based on the strategy profile chosen by all players.
\end{definition}
In a complete information game with $n$ players, a strategy profile consists of an $n$-tuple, where each element corresponds to the strategy selected by each individual player. 
The players of a static (bi)matrix security game are denoted by $A$ and $D$, where $A$ denotes the attacker, and $D$ is the defender. The finite action spaces are the set of attacks:

\begin{equation}A := \{a_1,...,a_{NA} \}\end{equation}

whereas the set of approaches chosen by the defender is:

\begin{equation}D := \{d_1,...,d_{ND} \}.\end{equation}

The game's outcome is represented by the $N_{A} \times N_{A}$  game matrices $G^{A}$ and $G^{D}$ for the attacker and defender. The entries in the matrices represent the costs for players, which they minimise. In the instance of a zero-sum security game, i.e., when $G^{A} = - G^{D}$, the matrix:
\begin{equation}
    G: = - G^{D} = - G^{A}
\end{equation}

is said to be a game matrix. In this accord, $P^A$ maximises its payoff while $P^D$ minimises its cost based on the entries of the game matrix. 

An essential solution concept in game theory is the Nash Equilibrium, formally defined in the following definition \ref{definitionNE}.

\begin{definition} \label{definitionNE}
A Nash Equilibrium strategy $(s^*, y^*)$ satisfies:  $s^*Ay^* \geq sAy^* \quad \forall s$ and $s^*By^* \geq s^*By \quad \forall y$.
\end{definition}

As the strategies could be either mixed or pure, the associated Nash Equilibrium is referred to as pure or mixed. Besides, if all the inequalities in Definition \ref{definitionNE} are strict, then one has a strict Nash Equilibrium; otherwise, it is considered non-strict.

\begin{definition} \label{definitionNEBR}
A strategy $y^{\ast}$ is a Nash equilibrium best response to $s^{\ast}$ (denoted $y^{\ast}$ $\in \mathcal{BR}(s^{\ast})$ in the sequel), if it is a strategy satisfying $s^{\ast} By^{\ast} \geq s^{\ast}By \,\forall y$. Hence, a Nash equilibrium strategy is a strategy pair $(s^{\ast}, y^{\ast})$ of mutual best responses: $s^{\ast} \in \mathcal {BR}(y^{\ast})$ and $y^{\ast} \in \mathcal{BR}(s^{\ast})$.
\end{definition}

\begin{lemma}
{
\textit{If Player 1’s mixed strategy $s^\ast$ is the best response to the (mixed) strategy $y$ of the other player, then, for each pure strategy $s_i$ such that $s_i > 0$, it must be the case that $s_i$ is itself the best response. In particular, the payoff $s_{i}Ay$ must be the same for all such strategies.}
}
\end{lemma}

\section{Adversary Model}\label{Adversary}

To connect strategic behaviour to quantitative risk, this section presents a minimal adversary model that (i) is compatible with mixed strategies at equilibrium, (ii) exposes measurable success probabilities that can be calibrated on images, and (iii) feeds directly into standard likelihood\,$\times$\, impact risk computation \cite{NIST80030}. The construction is aligned with statistical detection theory for wardening tasks and with security games where attackers best-respond or mix \cite{LehmannRomano2022,@Manshaei2013}.

\begin{definition}[Adversary model]\label{def:adv-model}
Let $\mathcal{G}$ be the $2\times2$ steganographic game in Section~\ref{A_Steganographic_Game-Theoretic_Model} with mixed Nash equilibrium $(p^\star,q^\star)$ from Section~\ref{Game_Analysis}. The adversary model $\mathcal{M}_{\mathcal{A}}$ consists of:
\begin{enumerate}[label=(\roman*),leftmargin=1.2em]

\item \emph{Actions and mixing.}
Let the defender's action set be $S_{\mathcal U}=\{H,\bar H\}$ and the adversary's action set be $S_{\mathcal A}=\{L,\bar L\}$, the mixed strategies are
\[
\sigma_{\mathcal U}=(q,1-q)\in\Delta(S_{\mathcal U}), 
\qquad 
\sigma_{\mathcal A}=(p,1-p)\in\Delta(S_{\mathcal A})\,,
\]
and the mixed–strategy Nash equilibrium (MSNE) from Section~\ref{Game_Analysis} is $(p^\star,q^\star)$, where $p^\star$ is the equilibrium probability of $L$ and $q^\star$ is the equilibrium probability of $H$. 

\item \emph{Effectiveness parameters.} When $\mathcal{U}$ hides, $\beta_t^{\mathcal{U}}\in[0,1]$ denotes, at epoch $t$, the probability that hiding prevents compromise; when $\mathcal{U}$ does not hide, $\bar\beta^{\mathcal{U}}\in[0,1]$ captures baseline (non-stego) protection. These parameters are estimated on spatial-domain  PGM (Portable Graymap) imagery with standard steganalysis back-ends in our experiments.

\item \emph{Success event and rates.} Let $\mathrm{Succ}$ be the event ``$\mathcal{A}$ compromises at epoch $t$''. The \emph{conditional} success probabilities $\Pr_{\mathcal{A}}(\mathrm{Succ}\mid \mathrm{Hide}_{\mathcal{U}};t)$ and $\Pr_{\mathcal{A}}(\mathrm{Succ}\mid \lnot\mathrm{Hide}_{\mathcal{U}};t)$ and their derivations are given in \S\ref{subsec:cond-success}; the \emph{unconditional} (overall) success $\Pr_{\mathcal{A}}(\mathrm{Succ};t)$ that averages over $q^\star$ appears in \S\ref{subsec:uncond-success}.

\item \emph{Adversary advantage.} 
The advantage of the warden/adversary $\mathcal{A}$ is
\begin{equation}
\mathsf{Adv}^{(+)}_{\mathcal{G},\mathcal{A}}(t):=\;\Pr_{\mathcal{A}}(\mathrm{Succ}\mid \lnot\mathrm{Hide}_{\mathcal{U}})\;-\;\Pr_{\mathcal{A}}(\mathrm{Succ}\mid \mathrm{Hide}_{\mathcal{U}};t),
\label{eq:signed-adv}
\end{equation}
which is \emph{derived from} the conditional rates in \S\ref{subsec:cond-success}. This model is employed in the risk analysis of \S\ref{Risk_ana}; an absolute-risk formulation is detailed in Appendix~\ref{app:risk-variants}.
\end{enumerate}
\end{definition}

\noindent In Definition~\ref{def:adv-model}, (iii) links the game to observable operating characteristics, while (iv) identifies the sole lift required to compare $s_{\mathcal{U}}$ in risk terms. In \S\ref{Risk_ana}, we compute operational risk $R_t=\mathbb{E}[I]\cdot\Pr_{\mathcal{A}}(\mathrm{Succ};t)$---with impact $I$ as scoped in \S\ref{Motivation_for_Security_Game Model}---using the unconditional rate from \S\ref{subsec:uncond-success}.

\noindent\textbf{Why this model is realistic.}
(i) Wardens operate as detectors with well-defined conditional error/success probabilities; modelling success conditional on the defender’s action is standard in detection theory \cite{LehmannRomano2022}. (ii) Strategic mixing is canonical in security games \cite{@Manshaei2013}; our $(p^\star,q^\star)$ enter only through the referenced success-rate sections. (iii) \emph{Calibratability.} The parameters $\beta_t^{\mathcal{U}}$ and $\bar\beta_t^{\mathcal{U}}$ are estimated on held-out grayscale images stored in the PGM format—e.g., standard spatial-domain benchmarks such as BOSSBase~1.01—using deployable steganalysis pipelines \footnote{BOSSBase~1.01 contains 10{,}000 grayscale images in \texttt{.pgm} format (512\,$\times$\,512), as noted in an academic repository: \url{https://rgu-repository.worktribe.com/output/2508022/population-based-methods-in-image-steganalysis}. Deployable steganalysis code (rich-model feature extractors and the Ensemble Classifier) is provided by the Digital Data Embedding Laboratory: \url{https://www3.cs.stonybrook.edu/~cvl/content/presentations/slides2014barni.pdf} (slide with ‘download/ensemble’ and ‘download/feature\_extractors’ links).}
 (See Section \ref{Risk_ana} for details), making $\Pr_{\mathcal{A}}(\cdot)$ and $\mathsf{Adv}^{(+)}_{\mathcal{G},\mathcal{A}}(\cdot)$ directly usable in risk computations consistent with NIST guidance\cite{NIST80030}.

\section{A Steganographic Game-Theoretic Model} 
\label{A_Steganographic_Game-Theoretic_Model}
Before we present the formal game, we introduce the notion of operational harmony for the defender. In this context, “harmony” refers to the normal state of system operations when no attack is detected and the organisation can carry out its business activities with minimal disruption. The associated monetary gain is denoted by $B^U_{\text{harmony}}$. We then formalise a two-player, static, non-cooperative game tailored to steganographic operations in surveilled networks. The defender (company) is denoted by \(\mathcal{U}\), and the adversary (surveilling warden) by \(\mathcal{A}\). The model is purpose-built to (i) separate \emph{primitive monetary quantities} (costs/benefits one could elicit from business or operational data data) from (ii) \emph{behavioural utilities} that capture risk attitudes and realistic nonlinearities. This separation lets us map empirical or policy inputs into a game whose payoffs better reflect real decision making under risk \cite{GordonLoeb2002, Romanosky2016}.

\paragraph{Players and actions.}
Let $(H,\bar H)\equiv(Hide_{\mathcal U},\lnot Hide_{\mathcal U})$ and $(L,\bar L)\equiv(Look_{\mathcal A},\lnot Look_{\mathcal A})$. The defender chooses whether to embed messages steganographically: \(s_{\mathcal{U}} \in \{ Hide_{\mathcal{U}},\, \lnot Hide_{\mathcal{U}} \}\).
The adversary chooses whether to actively search for stego content: \(s_{\mathcal{A}} \in \{ Look_{\mathcal{A}},\, \lnot Look_{\mathcal{A}} \}\).
The joint action \((s_{\mathcal{U}}, s_{\mathcal{A}})\) determines monetary consequences for both parties.

\paragraph{Primitive monetary components.}
We denote the defender’s \emph{monetary} cost and benefit components by
\[
C^U_{\text{hide}} \ge 0,\quad C^U_{\text{leak}} \ge 0,\quad B^U_{\text{hide}} \ge 0,\quad B^U_{\text{harmony}} \ge 0,
\]
corresponding to (respectively) implementation/operation of steganography, losses in the event of a successful compromise, gains from effective hiding, and the baseline operational benefit when the organisation is not attacked and remains in a harmonious state of operations.
For the adversary, we use
\[
C^A_{\text{look}} \ge 0,\quad B^A_{\text{leak}} \ge 0,
\]
capturing resource/time/legal exposure costs of surveillance and the gains from extracting valuable information.

\paragraph{Nonlinear utility transformations.}
Behavioural realism requires that the \emph{utility} of each monetary component reflects:
(i) \emph{diminishing returns} for benefits (e.g., the first increments of security yield larger marginal utility than later ones),
(ii) \emph{superlinear escalation} for certain losses (e.g., reputational/regulatory cascades), and
(iii) flexible curvature for deployment costs.
We therefore transform the primitives with monotone, twice-differentiable mappings \cite{Pratt1964}:

\begin{multline}
    \tilde B^U_{\text{hide}} = \beta_U \log\!\bigl(1 + B^U_{\text{hide}}/b_U\bigr), \\
\tilde B^U_{\text{harmony}} = \delta_U \log\!\bigl(1 + B^U_{\text{harmony}}/h_U\bigr),
\label{eq:def-concave-benefits}
\end{multline}
\begin{multline}
    \tilde C^U_{\text{leak}} = \alpha_{\ell}\,\bigl(C^U_{\text{leak}}\bigr)^{\gamma_{\ell}}, \gamma_{\ell} > 1, \\
\tilde C^U_{\text{hide}} = \eta_0 + \eta_1 \bigl(C^U_{\text{hide}}\bigr)^{\gamma_h}, \quad \gamma_h \ge 1,
\label{eq:def-convex-defender-costs}
\end{multline}
\begin{multline}
    \tilde B^A_{\text{leak}} = \beta_A \log\!\bigl(1 + B^A_{\text{leak}}/b_A\bigr),
\\ \tilde C^A_{\text{look}} = \alpha_{a}\, \bigl(C^A_{\text{look}}\bigr)^{\gamma_{a}}, \quad \gamma_{a}\ge 1.
\label{eq:def-adversary}
\end{multline}

\noindent
Here \(b_U,h_U,b_A>0\) are scale constants that tune the onset of saturation; 
\(\beta_U,\delta_U,\beta_A>0\) set units/weights for utilities;
\(\alpha_{\ell},\eta_0,\eta_1,\alpha_a>0\) scale costs, and exponents \(\gamma_{\ell}>1\) model superlinear loss escalation (e.g., regulatory penalties and reputational spillovers), while \(\gamma_h,\gamma_a\) allow linear or convex scaling of deployment/search costs.
The choices in \eqref{eq:def-concave-benefits}–\eqref{eq:def-adversary} are standard: 
logarithmic forms implement Arrow–Pratt risk aversion with decreasing absolute risk aversion (DARA) and diminishing marginal utility; power functions model convex loss regions where marginal disutility grows with incident size \cite{Pratt1964,EDPBGuidelines2023Fines,Romanosky2016}.\footnote{Alternative concave families (e.g., CRRA \(x^{1-\rho}/(1-\rho)\) with \(\rho\in(0,1)\)) are compatible with our analysis and can replace the logs if specific elicitation suggests.}
Empirically, security economics exhibits such curvature: breach consequences often magnify through legal/regulatory actions and reputational contagion, whereas protection benefits saturate once the \emph{most} salient risks are mitigated \cite{Romanosky2016,GordonLoeb2002}.

\paragraph{Utility payoff mapping.}
Given \((s_{\mathcal{U}}, s_{\mathcal{A}})\), we compute \emph{utilities} (not raw money) by adding transformed benefits and subtracting transformed costs. 
Let \(U_{\mathcal{U}}(\cdot,\cdot)\) and \(U_{\mathcal{A}}(\cdot,\cdot)\) denote utility payoffs.
The off-diagonal outcomes encode successful exploitation or needless search; the diagonal encodes defended/undefended calm.
We keep the model agnostic to the specific steganographic methods, though content-adaptive schemes (See Section \ref{sec:empirical} for empirical experiment details using \textbf{WOW}, \textbf{S-UNIWARD}, \textbf{HILL}, and \textbf{MiPOD}) and game-theoretic treatments of adaptivity (Sch{\"o}ttle--B{\"o}hme) provide natural instantiations for \(B^U_{\text{hide}}\) and adversary payoffs. See table \ref{tab:nonlinear-game} for utility payoff matrix

\begin{table*}[t]
\centering
\caption{Utility payoff matrix for the steganographic game \(\mathcal{G}\) (defender~\(\mathcal{U}\) rows; adversary~\(\mathcal{A}\) columns). Each cell lists \(\bigl(U_{\mathcal{U}},\,U_{\mathcal{A}}\bigr)\) using the nonlinear utilities in \eqref{eq:def-concave-benefits}–\eqref{eq:def-adversary}.}
\label{tab:nonlinear-game}
\vspace{1ex}
\begin{tabular}{c|c|c}
\(\mathcal{U}\downarrow\;\;\mathcal{A}\rightarrow\) 
& \(Look_{\mathcal{A}}\) 
& \(\lnot Look_{\mathcal{A}}\) \\ \hline\hline
\(Hide_{\mathcal{U}}\) &
\(\bigl(\;\tilde B^U_{\text{hide}} - \tilde C^U_{\text{hide}},\; -\,\tilde C^A_{\text{look}}\;\bigr)\) &
\(\bigl(\;\tilde B^U_{\text{hide}} - \tilde C^U_{\text{hide}},\; 0\;\bigr)\) \\
\hline
\(\lnot Hide_{\mathcal{U}}\) &
\(\bigl(\;-\tilde C^U_{\text{leak}},\; \tilde B^A_{\text{leak}} - \tilde C^A_{\text{look}}\;\bigr)\) &
\(\bigl(\;\tilde B^U_{\text{harmony}},\; 0\;\bigr)\) \\
\end{tabular}
\end{table*}

\paragraph{Interpretation.}
\begin{itemize}
\item \textbf{Top-left} (\(Hide_{\mathcal{U}}, Look_{\mathcal{A}}\)): the defender accrues protected-operations utility \(\tilde B^U_{\text{hide}}\) but pays convex deployment/maintenance \(\tilde C^U_{\text{hide}}\); the adversary pays search \(\tilde C^A_{\text{look}}\) with no leak benefit.
\item \textbf{Top-right} (\(Hide_{\mathcal{U}}, \lnot Look_{\mathcal{A}}\)): identical defender utility as above; the adversary abstains (utility \(0\)).
\item \textbf{Bottom-left} (\(\lnot Hide_{\mathcal{U}}, Look_{\mathcal{A}}\)): potential compromise yields superlinear defender disutility \(\tilde C^U_{\text{leak}}\) (escalation channels); the adversary gets saturated leak benefit \(\tilde B^A_{\text{leak}}\) net of convex search costs.
\item \textbf{Bottom-right} (\(\lnot Hide_{\mathcal{U}}, \lnot Look_{\mathcal{A}}\)): the defender enjoys baseline “harmony” utility \(\tilde B^U_{\text{harmony}}\); the adversary neither spends nor gains.
\end{itemize}

\noindent 
\begin{remark}(on realism and identifiability). The mapping \(\{\text{money}\}\to\{\text{utility}\}\) cleanly isolates curvature and keeps \emph{units} interpretable for elicitation.
Parameters \(\theta=\{b_U,h_U,b_A,\beta_U,\delta_U,\beta_A,\alpha_{\ell},\eta_0,\eta_1,\alpha_a,\gamma_{\ell},\gamma_h,\gamma_a\}\) can be fitted or bounded from:
(i) empirical breach/incident datasets,
(ii) regulatory penalty regimes, and 
(iii) engineering cost curves
(\emph{e.g.}, convex run-time/latency overheads, analyst-hours).
This preserves your downstream sensitivity/equilibrium/risk analysis while making payoffs reflect documented diminishing returns and superlinear liabilities \cite{Romanosky2016,EDPBGuidelines2023Fines}.
\end{remark}

\subsection{Relaxing Rationality: Bounded and Noisy Attackers}
\label{subsec:bounded-rationality}

Classical equilibrium models assume perfectly rational players with common knowledge of pay-offs. In real deployments, however, attackers may behave in ways that deviate from this ideal: they may be myopic \cite{Kardes2011Robust}, focusing only on the immediate payoff rather than anticipating future detection costs (for example, a hacker launching a quick exploit without considering that successive attacks reduce stealth); they may be noisy \cite{McKelveyPalfrey1995,Camerer2011BGT}, meaning they make probabilistic or error‐prone decisions rather than consistently choosing the exact best response (for example, due to imperfect information, fatigue, or randomised tactics); or they may be heterogeneous \cite{Camerer2011BGT}, meaning that the population of attackers comprises different types with varying capabilities, objectives or decision rules (for example, novice insiders, skilled criminals and fully automated bots operating at different levels of search effort). To reflect these empirical realities, we adopt three standard relaxations from behavioural and robust game theory, and propagate them through our success‐rate and risk calculations.

\paragraph{(R1) Quantal–Response (logit) attackers.}
Instead of best responses, the adversary selects actions with probability proportional to their expected utility:
\[
\pi_{\mathcal A}(a \mid \lambda)
=\frac{\exp\{\lambda\,U_{\mathcal A}(a)\}}{\sum_{a'}\exp\{\lambda\,U_{\mathcal A}(a')\}},\quad \lambda\ge 0,
\]
where $\lambda$ controls noise ( $\lambda=0$ uniform random; $\lambda\!\to\!\infty$ recovers best response). A \emph{quantal–response equilibrium} (QRE) solves fixed points of such probabilistic responses for both players; it preserves the comparative statics of the game while allowing systematic mistakes \cite{McKelveyPalfrey1995,Camerer2011BGT}. In this analysis, we replace the adversary’s mix $p^\star$ by the QRE search probability $\tilde p(\lambda)$, and likewise the defender’s $q^\star$ by $\tilde q(\lambda)$, then reuse the conditional and unconditional success formulas defined in Sections~\ref{subsec:cond-success}–\ref{subsec:uncond-success} and the risk mapping expressed in Section~\ref{Risk_ana}.

\paragraph{(R2) Level-$k$ / type mixtures.}
We admit a population mixture of attacker types $\{\tau_k\}$ with weights $\rho_k$; e.g., level–0 random, level–1 best response to level–0, etc., as in behavioural game theory \cite{Camerer2011BGT}. The effective search probability becomes a convex combination $\tilde p=\sum_k \rho_k\,p_k$, producing success rates by linearity of probability (Sections~\ref{subsec:cond-success}–\ref{subsec:uncond-success}) and an expected‐loss risk via the same substitution in Section~\ref{Risk_ana}. This captures heterogeneous or learning-lag attackers without assuming a single perfectly rational type.

\paragraph{(R3) Robust (worst–case) attackers.}
For settings where unpredictability is paramount, we compute \emph{robust} risk by taking the supremum over a set of plausible attacker mixes $\mathcal P$ (or payoff perturbations):
\[
\widehat R \;=\; \sup_{p\in \mathcal P}\; \mathbb E[I]\cdot \Pr_{\mathcal A}(\mathrm{Succ};\,p, q),
\]
which yields conservative, distribution-free bounds under payoff uncertainty or bounded rationality \cite{Kardes2011Robust}. We adopt the same substitution principle: use $p\!\in\!\mathcal P$ in Sections~\ref{subsec:cond-success}–\ref{subsec:uncond-success}, then apply Section~\ref{Risk_ana} with the resulting envelope.

\paragraph{Discussion.}
These relaxations are standard in security games and behavioural game theory: QRE models stochastic choice around better responses; level-$k$ mixtures encode heterogeneous sophistication; robust formulations protect against misspecification and adversarial unpredictability. All three integrate seamlessly with our framework by replacing $(p^\star,q^\star)$ with $(\tilde p,\tilde q)$ from (R1)–(R3) and reusing the downstream success‐rate and risk equations. This approach addresses the concern that “attackers may act unpredictably or irrationally”.

\subsection{Game Model Assumptions: Payoff and Utility Dynamics}
\label{sec:assumptions}

We work with the nonlinear utility transformation introduced in Section~\ref{A_Steganographic_Game-Theoretic_Model}. Let the \emph{primitive} monetary components $\{C^U_{\text{hide}}, C^U_{\text{leak}}, B^U_{\text{hide}}, B^U_{\text{harmony}}, C^A_{\text{look}}, B^A_{\text{leak}}\}$ be strictly nonnegative and finite.\footnote{All parameters are assumed measurable and elicitable from either internal accounting, breach postmortems, or regulatory/market studies; see, e.g., \cite{Romanosky2016,GordonLoeb2002,IoannidisPymWilliams2009}.}

\paragraph{(A1) Monotone, smooth utility curvature.}
For each monetary component $x\!\ge\!0$, the utility mapping $\tilde x \equiv u(x)$ is $\mathcal{C}^2$ and strictly increasing. Benefits use concave forms (e.g., $\log$ or CRRA) to capture diminishing marginal utility; losses and some operational costs use power functions with exponent $>1$ to capture superlinear disutility/overheads.

This is consistent with Arrow–Pratt risk aversion for benefits and convex escalation for large incidents/costs \cite{levy2002arrow,Pratt1964, deck2024simple}. Regulatory fine schedules and reputational cascades also justify convex loss utilities \cite{EDPBGuidelines2023Fines, Romanosky2016}.

\paragraph{(A2) ``Adequate protection'' ordering (defender).}
At the \emph{monetary} level,
\begin{equation}
C^U_{\text{hide}} < C^U_{\text{leak}},
\label{eq:adequate-protection}
\end{equation}
reflecting the classical security design maxim that prevention should cost less than failure (adequate protection). Because $u(\cdot)$ is increasing, \eqref{eq:adequate-protection} preserves the ordering at the utility level ($\tilde C^U_{\text{hide}}<\tilde C^U_{\text{leak}}$) \cite{pfleeger2015security,maghrabimaeva}.

\paragraph{(A3) ``Easiest penetration'' ordering (adversary).}
At the \emph{monetary} level,
\begin{equation}
C^A_{\text{look}} < B^A_{\text{leak}},
\label{eq:easiest-penetration}
\end{equation}
encoding that rational adversaries attack when expected gains exceed costs (``easiest penetration''). Monotonicity gives $\tilde C^A_{\text{look}}<\tilde B^A_{\text{leak}}$ \cite{pfleeger2015security,maghrabimaeva}. \vspace{0.2em}

\paragraph{(A2$^\prime$) Positive Security Incentive (optional, utility space).}
In scenarios where regulatory, reputational, and operational stability dominate marginal gains from a single control, we calibrate the utility ranking
\begin{equation}
\label{eq:positive-incentive-utility}
\tilde B^{U}_{\text{harmony}}
\;>\;
\tilde B^{U}_{\text{hide}} - \tilde C^{U}_{\text{hide}}\,.
\end{equation}
This utility–space expression replaces the monetary inequality to respect the curvature introduced by our nonlinear mappings. It aligns with risk‐management guidance that prioritises mission integrity and compliance harms \cite{NIST80030}, and it echoes well-established security-economics findings: namely, that optimal investment in protection is often substantially less than the full expected loss \cite{GordonLoeb2002}. It also reflects the super‐linear nature of data-protection penalties that can make “harmony” (steady, non-disrupted operations) more valuable than the net benefit of an incremental control \cite{EDPBGuidelines2023Fines}. Full strategic implications are analysed in Section~\ref{Game_Analysis}. 

\paragraph{(A3$^\prime$) Rational Attack Motivation (optional, utility space).}
For environments with economically motivated adversaries, we assume the attacker’s leak benefit is strictly positive in utility space,
\begin{equation}
\label{eq:rational-attack-utility}
\tilde B^{A}_{\text{leak}} \;>\; 0,
\end{equation}
and (where appropriate for a given scenario) the expected \emph{net} utility of attacking is positive:
\begin{equation}
\label{eq:rational-attack-expected}
EU_{\mathcal A}(Look_{\mathcal A}\,|\,q)
=(1-q)\,\tilde B^{A}_{\text{leak}}-\tilde C^{A}_{\text{look}} \;>\; 0\,.
\end{equation}
These statements reflect standard rational-choice/cost–benefit models of offending and attacker behaviour \cite{Becker1968,AndersonSE3}. They are optional scenario calibrations; the stronger monetary cost–benefit ordering introduced in assumption (A2) earlier in the paper implies condition \eqref{eq:rational-attack-utility}. Formal consequences for best responses/mixing are developed in Section~\ref{Game_Analysis}.

\paragraph{(A4) Baselines and normalization.}
The adversary’s abstention payoff is normalised to $0$; the defender’s calm-operation utility satisfies $\tilde B^U_{\text{harmony}}>0$. Normalisation by positive affine transforms of utilities leaves best responses and equilibria invariant \cite{osborne2004introduction}. \vspace{0.2em}

\paragraph{(A5) Expected-utility choice.}
Players maximise \emph{expected utility} (vNM) of the bimatrix entries in Table~\ref{tab:nonlinear-game}; mixed strategies are evaluated as \emph{linear} expectations of \emph{utilities}, not raw money \cite{osborne2004introduction}. Let the defender mix with $q\in[0,1]$ over $\{Hide_{\mathcal U},\lnot Hide_{\mathcal U}\}$ and the adversary mix with $p\in[0,1]$ over $\{Look_{\mathcal A},\lnot Look_{\mathcal A}\}$. With utility entries from Table~\ref{tab:nonlinear-game}, our \emph{primary} formulas are the action-wise expected utilities:
\begin{multline}
      EU_{\mathcal U}(Hide_{\mathcal U}\,|\,p)=p\,u_{\mathcal U}(H,L)+ \\ (1-p)\,u_{\mathcal U}(H,\bar L)
=\tilde B^U_{\text{hide}}-\tilde C^U_{\text{hide}},  
\end{multline}
\begin{multline}
     EU_{\mathcal U}(\lnot Hide_{\mathcal U}\,|\,p)
=p\,u_{\mathcal U}(\bar H,L)+(1-p)\,u_{\mathcal U}(\bar H,\bar L)
\\ =\tilde B^U_{\text{harmony}}-p\!\left(\tilde B^U_{\text{harmony}}+\tilde C^U_{\text{leak}}\right),  
\end{multline}
\begin{multline}
 EU_{\mathcal A}(Look_{\mathcal A}\,|\,q)
=q\,u_{\mathcal A}(H,L)+\\(1-q)\,u_{\mathcal A}(\bar H,L)
=(1-q)\tilde B^A_{\text{leak}}-\tilde C^A_{\text{look}},     
\end{multline}
\begin{align}
 EU_{\mathcal A}(\lnot Look_{\mathcal A}\,|\,q)
&=q\,u_{\mathcal A}(H,\bar L)+(1-q)\,u_{\mathcal A}(\bar H,\bar L)=0.   
\end{align}

Action-wise expected utilities are preferred here because (i) they operate directly on the \emph{nonlinear, risk-aware} utilities $(\tilde B,\tilde C)$ we specified, preserving curvature in all subsequent decisions; (ii) they map one-to-one into best-response comparisons and mixed-strategy indifference in finite games without expanding the full joint expectation, thereby keeping the link between primitives and behaviour transparent; and (iii) they are the standard device for characterizing optimal responses under expected utility and for deriving mixed-strategy conditions in normal-form games\cite{osborne2004introduction,Tadelis2013}. Full equilibrium characterisation (best responses, indifference, and existence) is deferred to Section~\ref{Game_Analysis}.

\paragraph{(A7) Identifiability and elicitation.}
Curvature parameters $\theta=\{b_U,h_U,b_A,\beta_U,\delta_U,\beta_A,\alpha_\ell,\gamma_\ell,\eta_0,\eta_1,\alpha_a,\gamma_h,\gamma_a\}$ can be constrained or fitted from (i) observed post-incident loss tails (e.g., superlinear penalties/regulatory fines), (ii) internal cost curves for deployment/monitoring (often convex in workload), and (iii) macro estimates for breach impacts. Canonical security-economics work (e.g., Gordon--Loeb) also supports diminishing returns to investment, which aligns with concave benefit utilities on the defender side \cite{GordonLoeb2002,IoannidisPymWilliams2009}. \vspace{0.2em}

\paragraph{(A8) Non–zero-sum and asymmetry.}
We make no zero-sum assumption. The defender’s and adversary’s utilities need not be negatives of each other; this better reflects compliance, reputation, and process frictions that do not symmetrically transfer utility between players \cite{osborne2004introduction}. \vspace{0.2em}

\paragraph{(A9) Mapping monetary scenarios to utilities.}
Because $u(\cdot)$ is strictly increasing, any policy/empirical inequality posited at the monetary level (e.g., \eqref{eq:adequate-protection}–\eqref{eq:easiest-penetration}) induces the same ordering at the utility level. This permits scenario design and sensitivity analysis to be specified in familiar monetary terms while guaranteeing consistent utility comparisons in equilibrium analysis \cite{osborne2004introduction}.

\section{Game Model Analysis}
\label{Game_Analysis}
This section analyses the Nash Equilibrium strategies of the steganographic game $\mathcal{G}$ presented in section \ref{A_Steganographic_Game-Theoretic_Model}. First, it is demonstrated that, whereas the game $\mathcal{G}$ does not admit a pure Nash Equilibrium solution, a Mixed Nash Equilibrium exists. Second, the game's strategies and payoffs are discussed in detail, which includes the use of numerical values that aim at representing those found in some real scenarios. For bounded-rational variants (QRE, level-
k, robust), see Section~\ref{subsec:bounded-rationality}.

\subsection{Pure Nash Equilibrium Analysis}\label{PureNashAnalysis}
All possible states are considered to solve the game $\mathcal G$ and find potential Pure Nash Equilibria. By examining deviation incentives, the analysis in this subsection demonstrates that no combination of pure strategies exists for players $\mathcal{U}$ and $\mathcal{A}$ to satisfy a Nash equilibrium solution. The evaluation of deviation incentives involves assessing potential gains or losses for a player if they deviate unilaterally while the other player holds fixed \cite{osborne2004introduction,Tadelis2013}. 

\begin{theorem}
The steganographic security game $\mathcal G$ with utility payoffs given in Table~\ref{tab:nonlinear-game} admits no pure Nash equilibrium.
\end{theorem}

Throughout, we use the assumptions in Section~\ref{sec:assumptions}: (i) $\tilde C^A_{\text{look}}>0$ and $\tilde B^A_{\text{leak}}>\tilde C^A_{\text{look}}$; 
(ii) $\tilde C^U_{\text{leak}}>\tilde C^U_{\text{hide}}$; and 
(iii) \emph{harmony-dominance} $\tilde B^U_{\text{harmony}}>\tilde B^U_{\text{hide}}-\tilde C^U_{\text{hide}}$.
Note these are \emph{utility}-level inequalities; the monotone curvature mappings ensure signs and orderings are preserved relative to the underlying monetary primitives.

\medskip
\textit{(a) State 1} $(Hide_{\mathcal U}, Look_{\mathcal A})$.
\begin{proposition}
Under $Hide_{\mathcal U}$, the adversary strictly prefers $\lnot Look_{\mathcal A}$ to $Look_{\mathcal A}$; hence $(Hide_{\mathcal U}, Look_{\mathcal A})$ cannot be a pure Nash equilibrium.
\end{proposition}
\begin{proof}
From Table~\ref{tab:nonlinear-game},
\[
u_{\mathcal A}(H,L) \;=\; -\,\tilde C^A_{\text{look}} \;<\; 0
\quad\text{and}\quad
u_{\mathcal A}(H,\bar L) \;=\; 0.
\]
Intuitively, when the defender hides, the warden’s search generates cost but no direct leak benefit; with no leak payoff to offset $\tilde C^A_{\text{look}}$, searching is strictly worse than abstaining. Formally, since $\tilde C^A_{\text{look}}>0$ and utility is additively separable in this cell, the unilateral deviation $Look_{\mathcal A}\to \lnot Look_{\mathcal A}$ increases the adversary’s payoff from $-\,\tilde C^A_{\text{look}}$ to $0$. Therefore $(Hide_{\mathcal U}, Look_{\mathcal A})$ violates mutual best response and cannot be a PSNE.
\end{proof}

\medskip
\textit{(b) State 2} $(Hide_{\mathcal U}, \lnot Look_{\mathcal A})$.
\begin{proposition}
If $\tilde B^U_{\text{harmony}}>\tilde B^U_{\text{hide}}-\tilde C^U_{\text{hide}}$, then $(Hide_{\mathcal U}, \lnot Look_{\mathcal A})$ is not a pure Nash equilibrium.
\end{proposition}
\begin{proof}
Holding $\lnot Look_{\mathcal A}$ fixed, the adversary’s payoff $u_{\mathcal A}(H,\bar L)=0$ already weakly dominates $u_{\mathcal A}(H,L)=-\tilde C^A_{\text{look}}$; thus $\lnot Look_{\mathcal A}$ is a best response to $Hide_{\mathcal U}$. The defender’s choice is pivotal: with $Hide_{\mathcal U}$ the defender gets
\[
u_{\mathcal U}(H,\bar L) \;=\; \tilde B^U_{\text{hide}}-\tilde C^U_{\text{hide}},
\]
while deviating to $\lnot Hide_{\mathcal U}$ yields
\[
u_{\mathcal U}(\bar H,\bar L) \;=\; \tilde B^U_{\text{harmony}}.
\]
By the stated inequality $\tilde B^U_{\text{harmony}}>\tilde B^U_{\text{hide}}-\tilde C^U_{\text{hide}}$, the deviation $\;Hide_{\mathcal U}\to\lnot Hide_{\mathcal U}\;$ strictly increases the defender’s utility in the calm (no-search) column. Economically, if the warden is not searching, continuing to pay the convex running cost of the stego control is strictly dominated by enjoying the harmony state. Hence $(Hide_{\mathcal U}, \lnot Look_{\mathcal A})$ fails to be a PSNE.
\end{proof}

\medskip
\textit{(c) State 3} $(\lnot Hide_{\mathcal U}, Look_{\mathcal A})$.
\begin{proposition}
If $\tilde C^U_{\text{leak}}>\tilde C^U_{\text{hide}}$ and $\tilde B^U_{\text{hide}}\ge 0$, then $(\lnot Hide_{\mathcal U}, Look_{\mathcal A})$ is not a pure Nash equilibrium.
\end{proposition}
\begin{proof}
Given $Look_{\mathcal A}$, the defender compares the \emph{leak} cell with the \emph{hide while attacked} cell. Under no hiding,
\[
u_{\mathcal U}(\bar H,L) \;=\; -\,\tilde C^U_{\text{leak}}.
\]
If the defender deviates to hiding,
\[
u_{\mathcal U}(H,L) \;=\; \tilde B^U_{\text{hide}} - \tilde C^U_{\text{hide}}
\;\;\ge\;\; -\,\tilde C^U_{\text{hide}}
\quad\text{(since $\tilde B^U_{\text{hide}}\ge 0$).}
\]
Consider the deviation gain:
\begin{multline*}
\Delta_{\mathcal U}
=
u_{\mathcal U}(H,L)-u_{\mathcal U}(\bar H,L)=\bigl(\tilde B^U_{\text{hide}}-\tilde C^U_{\text{hide}}\bigr)
- \bigl(-\,\tilde C^U_{\text{leak}}\bigr)\\=
\tilde B^U_{\text{hide}}+\tilde C^U_{\text{leak}}-\tilde C^U_{\text{hide}}.   
\end{multline*}

With $\tilde C^U_{\text{leak}}>\tilde C^U_{\text{hide}}$ and $\tilde B^U_{\text{hide}}\ge 0$, we have $\Delta_{\mathcal U}>0$. Hence, the defender strictly prefers to deviate to $Hide_{\mathcal U}$ against an active searcher. (For completeness: the adversary is indeed best responding in this row because $u_{\mathcal A}(\bar H,L)=\tilde B^A_{\text{leak}}-\tilde C^A_{\text{look}}>0 = u_{\mathcal A}(\bar H,\bar L)$.) Therefore $(\lnot Hide_{\mathcal U}, Look_{\mathcal A})$ is not a PSNE.
\end{proof}

\medskip
\textit{(d) State 4} $(\lnot Hide_{\mathcal U}, \lnot Look_{\mathcal A})$.
\begin{proposition}
If $\tilde B^A_{\text{leak}}>\tilde C^A_{\text{look}}$, then $(\lnot Hide_{\mathcal U}, \lnot Look_{\mathcal A})$ is not a pure Nash equilibrium.
\end{proposition}
\begin{proof}
Holding $\lnot Hide_{\mathcal U}$ fixed, the adversary compares
\[
u_{\mathcal A}(\bar H,\bar L)=0
\quad\text{vs.}\quad
u_{\mathcal A}(\bar H,L)=\tilde B^A_{\text{leak}}-\tilde C^A_{\text{look}}.
\]
Under the profitability condition $\tilde B^A_{\text{leak}}>\tilde C^A_{\text{look}}$, the deviation $\;\lnot Look_{\mathcal A}\to Look_{\mathcal A}\;$ yields strictly positive utility (leak benefit net of convex search cost). Intuitively, if the defender leaves traffic unprotected, a rational warden exploits it. Thus $(\lnot Hide_{\mathcal U}, \lnot Look_{\mathcal A})$ fails to be a PSNE.
\end{proof}

\medskip
\noindent Each of the four pure profiles admits a profitable unilateral deviation by at least one player under standard security–economics inequalities on the utility-transformed payoffs. Hence, no pure Nash equilibrium exists.

Having demonstrated that the game $G$ does not admit any Pure Nash Equilibrium solution, the following subsection \ref{MixNE} presents the Mixed Nash Equilibrium solution.

\subsection{Mixed Nash Equilibrium Analysis}\label{MixNE}
A mixed–strategy Nash equilibrium (MSNE) allows players to randomise across pure actions so that each player is indifferent across the actions used with positive probability; in finite normal–form games, this is characterised by \emph{linearity of expected utility} and the \emph{indifference principle} \cite{osborne2004introduction,Tadelis2013}. We now derive the MSNE for $\mathcal G$ using the \emph{action-wise expected utilities} defined in (A5) and the utility entries in Table~\ref{tab:nonlinear-game}.

\begin{theorem}\label{theorem_p}
Consider the steganographic game $\mathcal G$ with utility payoffs in Table~\ref{tab:nonlinear-game}. Suppose
\begin{multline*}
0<\tilde C^A_{\text{look}}<\tilde B^A_{\text{leak}},\qquad
0<\tilde B^U_{\text{harmony}}-\tilde B^U_{\text{hide}}+\tilde C^U_{\text{hide}} \\
<\tilde B^U_{\text{harmony}}+\tilde C^U_{\text{leak}}.   
\end{multline*}

Then there exists an \emph{interior} mixed–strategy equilibrium $(p^\star,q^\star)\in(0,1)^2$, where
\begin{equation}\label{equation_p}
    p^{\star}
    =\frac{\tilde B^U_{\text{harmony}}-\tilde B^U_{\text{hide}}+\tilde C^U_{\text{hide}}}
           {\tilde B^U_{\text{harmony}}+\tilde C^U_{\text{leak}}},
\end{equation}
\begin{equation}\label{equation_q}
    q^{\star}
    = 1-\frac{\tilde C^A_{\text{look}}}{\tilde B^A_{\text{leak}}}.
\end{equation}
Here $p^\star$ is the adversary’s probability of \emph{searching} $(Look_{\mathcal A})$ and $q^\star$ is the defender’s probability of \emph{hiding} $(Hide_{\mathcal U})$. 
\end{theorem}

\paragraph{Derivation via action-wise indifference.}
Let the adversary mix with $p\in[0,1]$ over $\{L,\bar L\}$ and the defender mix with $q\in[0,1]$ over $\{H,\bar H\}$, where $H\equiv Hide_{\mathcal U}$, $L\equiv Look_{\mathcal A}$, and bars denote “do not”.

\smallskip
\noindent\textbf{Defender’s indifference (solves for $p^\star$).}  
Using (A5) with Table~\ref{tab:nonlinear-game},
\begin{align*}
EU_{\mathcal U}(H\,|\,p) \;=\; \tilde B^U_{\text{hide}}-\tilde C^U_{\text{hide}},\qquad \\
EU_{\mathcal U}(\bar H\,|\,p)\;=\; \tilde B^U_{\text{harmony}}-p\,\bigl(\tilde B^U_{\text{harmony}}+\tilde C^U_{\text{leak}}\bigr).   
\end{align*}

In any interior MSNE the defender is indifferent: $EU_{\mathcal U}(H\,|\,p^\star)=EU_{\mathcal U}(\bar H\,|\,p^\star)$. Hence
\[
\tilde B^U_{\text{hide}}-\tilde C^U_{\text{hide}}
=\tilde B^U_{\text{harmony}}-p^\star\!\left(\tilde B^U_{\text{harmony}}+\tilde C^U_{\text{leak}}\right),
\]
which yields \eqref{equation_p}:
\[
p^\star
=\frac{\tilde B^U_{\text{harmony}}-\tilde B^U_{\text{hide}}+\tilde C^U_{\text{hide}}}
       {\tilde B^U_{\text{harmony}}+\tilde C^U_{\text{leak}}}.
\]

\smallskip
\noindent\textbf{Adversary’s indifference (solves for $q^\star$).}
Again from (A5) and Table~\ref{tab:nonlinear-game},
\[
EU_{\mathcal A}(L\,|\,q)=(1-q)\,\tilde B^A_{\text{leak}}-\tilde C^A_{\text{look}},\qquad
EU_{\mathcal A}(\bar L\,|\,q)=0.
\]
Indifference at an interior MSNE requires $EU_{\mathcal A}(L\,|\,q^\star)=EU_{\mathcal A}(\bar L\,|\,q^\star)$, hence
\[
(1-q^\star)\,\tilde B^A_{\text{leak}}-\tilde C^A_{\text{look}}=0
\quad\Rightarrow\quad
q^\star=1-\frac{\tilde C^A_{\text{look}}}{\tilde B^A_{\text{leak}}},
\]
which is \eqref{equation_q}.

The stated inequalities ensure $0<p^\star<1$ and $0<q^\star<1$:the denominator of \eqref{equation_p} is strictly positive by construction, and the numerator lies strictly between $0$ and the denominator; similarly, $0<\tilde C^A_{\text{look}}<\tilde B^A_{\text{leak}}$ ensures $q^\star\in(0,1)$. By Nash’s theorem, an equilibrium exists in every finite game; given the pure–strategy nonexistence proved in Section~\ref{PureNashAnalysis}, the equilibrium characterised above is (under the stated regularity) the unique interior MSNE solving the mutual indifference conditions \cite{osborne2004introduction}.
\subsection{Conditional Success Rates under MixNE}
\label{subsec:cond-success}

Building on the mixed–strategy solution in \S\ref{MixNE}, let \(p^\star\in[0,1]\) and \(q^\star\in[0,1]\) denote, respectively, the adversary’s and defender’s equilibrium randomization probabilities. The steganographic \emph{effectiveness} of the defender is modeled by the time–varying parameter \(\beta_t^{\mathcal U}\in[0,1]\), interpreted as the conditional probability that steganographic embedding prevents a successful compromise at decision epoch~\(t\). 

In contrast, the parameter \(\bar\beta^{\mathcal U}\in[0,1]\) represents a baseline or non-steganographic control state, capturing the limited residual protection offered by ordinary hardening measures such as intrusion prevention, network segmentation, or anomaly filtering. 
A large value of \(\bar\beta^{\mathcal U}\) (e.g., \(0.9\)–\(0.95\)) encodes ineffective baseline controls when \(\lnot\mathrm{Hide}_{\mathcal U}\) is chosen, indicating that conventional defences provide minimal resistance to compromise and motivating the adoption of steganographic concealment to enhance security. All formulas remain valid if \(p^{\star},q^{\star}\) are replaced by the bounded-rational variants \(\tilde p,\tilde q\) from Section~\ref{subsec:bounded-rationality}.
\paragraph{Adversary (conditional) success.}
Given the defender’s action, the adversary’s conditional success probabilities are defined as
\begin{align}
\Pr_{\mathcal A}\!\bigl(Succ \mid Hide_{\mathcal U};\,t\bigr)
    &= p^\star \bigl(1-\beta_t^{\mathcal U}\bigr), 
    \label{eq:adv-hide-cond}\\[2mm]
\Pr_{\mathcal A}\!\bigl(Succ \mid \lnot Hide_{\mathcal U};\,t\bigr)
    &= p^\star \bigl(1-\bar\beta^{\mathcal U}\bigr),
    \label{eq:adv-nohide-cond}
\end{align}  
The monotonicity is immediate: 
$\partial \Pr_{\mathcal A}(Succ\mid\cdot)/\partial p^\star \ge 0$ and 
$\partial \Pr_{\mathcal A}(Succ\mid\cdot)/\partial \beta \le 0$.

\paragraph{Defender (conditional) success.}
By the complementarity of detection and miss events, the defender’s conditional success probabilities follow as
\begin{align}
\Pr_{\mathcal U}\!\bigl(Succ \mid Hide_{\mathcal U};\,t\bigr)
    &= 1 - p^\star \bigl(1-\beta_t^{\mathcal U}\bigr)
     \;=\; (1-p^\star) + p^\star \beta_t^{\mathcal U}, 
     \label{eq:def-hide-cond}\\[2mm]
\Pr_{\mathcal U}\!\bigl(Succ \mid \lnot Hide_{\mathcal U};\,t\bigr)
    &= 1 - p^\star \bigl(1-\bar\beta^{\mathcal U}\bigr)
     \;=\; (1-p^\star) + p^\star \bar\beta^{\mathcal U}.
     \label{eq:def-nohide-cond}
\end{align}
Thus, $\beta_t^{\mathcal U}$ captures the steganographic contribution to defence effectiveness,
whereas $\bar\beta^{\mathcal U}$ quantifies the limited success rate of baseline controls.

\paragraph{Extremal and intermediate regimes.}
For \eqref{eq:adv-hide-cond} (the same logic applies to \eqref{eq:adv-nohide-cond}):
\begin{multline}
   \text{Max adversary success:}\quad 
\Pr_{\mathcal A}(Succ \mid Hide_{\mathcal U};\,t)=1
\\ \text{iff } p^\star=1 \text{ and } \beta_t^{\mathcal U}=0; \label{eq:adv-max} 
\end{multline}
\begin{multline}
\text{Min adversary success:}\quad 
\Pr_{\mathcal A}(Succ \mid Hide_{\mathcal U};\,t)=0
\\ \text{if } p^\star=0 \ \text{or}\ \beta_t^{\mathcal U}=1; \label{eq:adv-min}   
\end{multline}
\begin{multline}
   \text{Intermediate regime:}\quad 
\Pr_{\mathcal A}(Succ \mid Hide_{\mathcal U};\,t)\in(0,1)
\\ \text{for } p^\star\in(0,1),\ \beta_t^{\mathcal U}\in(0,1).\label{eq:adv-int} 
\end{multline}

By \eqref{eq:def-hide-cond}, the defender’s extrema invert: 
$\Pr_{\mathcal U}(Succ \mid Hide_{\mathcal U};\,t)=1$ when $p^\star=0$ or $\beta_t^{\mathcal U}=1$, 
and $\Pr_{\mathcal U}(Succ \mid Hide_{\mathcal U};\,t)=0$ when $p^\star=1$ and $\beta_t^{\mathcal U}=0$. 
These bounds make explicit how equilibrium aggressiveness $p^\star$ 
and technical effectiveness $(\beta_t^{\mathcal U},\bar\beta^{\mathcal U})$ jointly govern operational outcomes.

\paragraph{Dynamic \(\beta_t^{\mathcal U}\) (learning/adaptation).}
To reflect evolving steganalysis and defensive tuning, the defender’s effectiveness is allowed to vary over time:
\begin{align}
\beta_{t+1}^{\mathcal U} \;=\; \beta_t^{\mathcal U} \;-\; \eta\,\mathsf{Adv}^{(+)}_{\mathcal{G},\mathcal{A}}(t),
\qquad \eta>0,\ \ \mathsf{Adv}^{(+)}_{\mathcal{G},\mathcal{A}}(t) \ge 0,
\label{eq:beta-grad}
\end{align}
where $\mathsf{Adv}^{(+)}_{\mathcal{G},\mathcal{A}}(t)$ summarises adversarial learning or detector improvement at epoch~$t$ 
(e.g., enhanced feature extraction or classifier precision reduces $\beta_t^{\mathcal U}$).
This adaptive formulation captures the dynamic tension between evolving steganalysis and defensive concealment strategies.
\subsection{Unconditional Success Rates under MixNE}
\label{subsec:uncond-success}
Building on the conditional analysis and the equilibrium randomisations \(p^\star,q^\star\) from \S\ref{MixNE}, we now aggregate over the defender’s hide/no–hide cases using the \emph{law of total probability}. This yields the \emph{unconditional} (overall) success probabilities that weight each conditional branch by its respective mixing probability. Formally, the events \(\{Hide_{\mathcal U},\lnot Hide_{\mathcal U}\}\) form a partition of the sample space for the defender’s action, hence
$\Pr(A)=\Pr(A\mid Hide_{\mathcal U})\,\Pr(Hide_{\mathcal U})+\Pr(A\mid \lnot Hide_{\mathcal U})\,\Pr(\lnot Hide_{\mathcal U}),$
as in \cite{ross2020first}.  
We retain the time–varying effectiveness parameters 
\(\beta_t^{\mathcal U}\in[0,1]\) (steganographic concealment) 
and \(\bar\beta^{\mathcal U}\in[0,1]\) (non–steganographic baseline control).

\paragraph{Adversary (unconditional) success.}
Using \eqref{eq:adv-hide-cond}–\eqref{eq:adv-nohide-cond} and \(\Pr(Hide_{\mathcal U})=q^\star\), we obtain
\begin{align}
\Pr_{\mathcal A}\!\bigl(Succ;\,t\bigr)
& = \Pr_{\mathcal A}\!\bigl(Succ \mid Hide_{\mathcal U};\,t\bigr)\,q^\star
   + \Pr_{\mathcal A}\!\bigl(Succ \mid \lnot Hide_{\mathcal U};\,t\bigr)\, \nonumber\\ & \quad(1-q^\star) \nonumber\\
&= p^\star\Big[q^\star\big(1-\beta_t^{\mathcal U}\big) + (1-q^\star)\big(1-\bar\beta^{\mathcal U}\big)\Big].
\label{eq:adv-uncond}
\end{align}
Expression~\eqref{eq:adv-uncond} is affine in \(p^\star\) and \(q^\star\) and nonincreasing in each effectiveness parameter. In particular,
\begin{multline}
\frac{\partial}{\partial p^\star}\Pr_{\mathcal A}(Succ;t)
= q^\star(1-\beta_t^{\mathcal U})+(1-q^\star)(1-\bar\beta^{\mathcal U}) \;\ge 0,
\\
\frac{\partial}{\partial \beta_t^{\mathcal U}}\Pr_{\mathcal A}(Succ;t)
= -\,p^\star q^\star \;\le 0,
\end{multline}
and the analogous sign holds for \(\bar\beta^{\mathcal U}\).  
Thus, adversary success grows with aggressiveness \(p^\star\) but decreases with either form of defensive effectiveness.

\paragraph{Defender (unconditional) success.}
From \eqref{eq:def-hide-cond}–\eqref{eq:def-nohide-cond} and the same probabilistic partition,
\begin{align}
\Pr_{\mathcal U}\!\bigl(Succ;\,t\bigr)
&= \Pr_{\mathcal U}\!\bigl(Succ \mid Hide_{\mathcal U};\,t\bigr)\,q^\star
 + \Pr_{\mathcal U}\!\bigl(Succ \mid \lnot Hide_{\mathcal U};\,t\bigr) \nonumber \\ &  \quad(1-q^\star) \nonumber\\
&= q^\star\!\big[(1-p^\star)+p^\star\beta_t^{\mathcal U}\big]
  + (1-q^\star)\!\big[(1-p^\star)+p^\star\bar\beta^{\mathcal U}\big] \nonumber\\
&= 1 - p^\star\Big[q^\star\big(1-\beta_t^{\mathcal U}\big)
     + (1-q^\star)\big(1-\bar\beta^{\mathcal U}\big)\Big]
\nonumber \\ &= 1 - \Pr_{\mathcal A}\!\bigl(Succ;\,t\bigr),
\label{eq:def-uncond}
\end{align}
where the final equality follows from the complementarity of compromise and preservation events in each detection branch \cite{KayDetTheory}.  
This dual formulation highlights that unconditional performance is jointly shaped by both the equilibrium mixing probabilities \((p^\star,q^\star)\) and the relative strengths \((\beta_t^{\mathcal U},\bar\beta^{\mathcal U})\), linking operational outcomes directly to the underlying game–theoretic structure.
\paragraph{Extremal and intermediate regimes.}
Because \(\Pr_{\mathcal A}(Succ;t)\) is affine in \(p^\star\in[0,1]\), its extrema with respect to \(p^\star\) occur at the boundary points:
\begin{align}
\nonumber\text{Max adversary success (given \(q^\star,\beta_t^{\mathcal U},\bar\beta^{\mathcal U}\)):} \Pr_{\mathcal A}(Succ;t)=\\ \nonumber
q^\star(1-\beta_t^{\mathcal U}) +(1-q^\star)(1-\bar\beta^{\mathcal U}) \\ 
 \quad \text{at }p^\star=1, \label{eq:adv-uncond-max}
\end{align}
\begin{align}
\text{Min adversary success:}\quad
&\Pr_{\mathcal A}(Succ;t)=0
&\text{at }p^\star=0.
\label{eq:adv-uncond-min}
\end{align}

By \eqref{eq:def-uncond}, the defender’s extrema invert accordingly:
\begin{multline*}
   \Pr_{\mathcal U}(Succ;t)=1
   \quad\text{when }p^\star=0,\\[1mm]
   \Pr_{\mathcal U}(Succ;t)
   = q^\star\beta_t^{\mathcal U}
     + (1-q^\star)\bar\beta^{\mathcal U}
   \quad\text{when }p^\star=1.
\end{multline*}

For interior values \(p^\star\in(0,1)\) and effectiveness parameters 
\(\beta_t^{\mathcal U},\bar\beta^{\mathcal U}\in(0,1)\), both players’ unconditional success probabilities lie strictly between these limits. The convex combination in \eqref{eq:adv-uncond} therefore makes explicit that, under equilibrium mixing, the overall adversary success rate is the weighted average of the branch-specific miss-detection probabilities, fully consistent with the law of total probability \cite{ross2020first}.

\subsection{Interpretation and Implications for the Nash Equilibrium Solutions}
This subsection interprets the Mixed Nash Equilibrium (MixNE) obtained in \S\ref{MixNE} and draws implications for security economics under the nonlinear utility mappings of \S\ref{A_Steganographic_Game-Theoretic_Model}. In a $2\times2$ game, the MixNE arises by making each player indifferent across their two actions; algebraically, this yields the closed-form probabilities in \eqref{equation_p}–\eqref{equation_q}. The logic is standard: in any finite game, a player mixes only over actions that yield equal expected \emph{utility}, so that unilateral deviations do not improve payoffs. Our use of transformed (concave/convex) utilities does not alter this structure—indifference is taken in utility space—so the derivation mirrors the textbook treatment of mixed strategies \cite{osborne2004introduction,Tadelis2013}.

\paragraph{Reading \(p^\star\).} 
Equation \eqref{equation_p} shows that the adversary’s equilibrium aggressiveness \(p^\star\) depends only on the \emph{defender’s} transformed parameters,
\[
p^\star=\frac{\tilde B^U_{\text{harmony}}-\tilde B^U_{\text{hide}}+\tilde C^U_{\text{hide}}}{\tilde B^U_{\text{harmony}}+\tilde C^U_{\text{leak}}}\in[0,1].
\]
Intuitively, when leak losses escalate superlinearly (larger \(\tilde C^U_{\text{leak}}\) via higher \(\gamma_\ell\) or scale), the denominator grows, pushing \(p^\star\) \emph{down}. To keep the defender indifferent between hide and no–hide, the warden must \emph{attack less often}; otherwise the defender would strictly prefer to hide. Conversely, if deploying steganography is convex (larger \(\tilde C^U_{\text{hide}}\)), the numerator rises and \(p^\star\) increases; the more costly defence lets the warden mix more aggressively while maintaining the defender’s indifference. Because \(\tilde B^U_{\text{hide}}\) and \(\tilde B^U_{\text{harmony}}\) are concave (diminishing returns), their marginal influence on \(p^\star\) attenuates as the corresponding monetary primitives grow, matching risk-averse behaviour (Arrow–Pratt) \cite{osborne2004introduction, levy2002arrow,Pratt1964}.

\paragraph{Reading \(q^\star\).}
Equation \eqref{equation_q} depends only on the \emph{adversary’s} transformed parameters,
\[
q^\star=1-\frac{\tilde C^A_{\text{look}}}{\tilde B^A_{\text{leak}}}\in[0,1].
\]
Raising search/surveillance cost \(\tilde C^A_{\text{look}}\) (e.g., steganalysis compute, analyst time, legal exposure) increases \(q^\star\), i.e., the defender hides more often at equilibrium; decreasing the warden’s net benefit \(\tilde B^A_{\text{leak}}\) has the same effect. Because \(\tilde B^A_{\text{leak}}\) is concave, expansions of the warden’s monetary benefit face diminishing marginal utility, which dampens the reduction in \(q^\star\) when leak payoffs grow. These comparative statics are the utility-space analogue of classic best-response thresholds for mixed strategies \cite{Tadelis2013}.

\paragraph{Feasibility and corners.}
The bounds \(0\le q^\star\le1\) require \(\tilde C^A_{\text{look}}\le \tilde B^A_{\text{leak}}\); otherwise the adversary’s “look” is dominated and the equilibrium collapses to a corner with \(q^\star=1\) and \(p^\star=0\). Similarly, \(0\le p^\star\le1\) imposes \(\tilde B^U_{\text{harmony}}+\tilde C^U_{\text{leak}}>0\) (true by construction) and \(\tilde B^U_{\text{harmony}}-\tilde B^U_{\text{hide}}+\tilde C^U_{\text{hide}}\in[0,\tilde B^U_{\text{harmony}}+\tilde C^U_{\text{leak}}]\). When feasibility fails, the MixNE devolves to a pure best response (e.g., \(\lnot Look_{\mathcal A}\) if surveillance is too costly), exactly as standard equilibrium existence/selection arguments predict \cite{osborne2004introduction}.

\paragraph{Link to success rates and risk.}
The equilibrium mixes \((p^\star,q^\star)\) are \emph{inputs} to operational metrics. In \S\ref{subsec:cond-success}–\S\ref{subsec:uncond-success} we translated them into (i) conditional success rates under hide/no–hide branches and (ii) unconditional overall success via total probability, and then into expected loss \(\mathbb E[L]\) by multiplying by impact. This is exactly the probability–impact decomposition used in quantitative risk analysis; if success and impact are dependent, a Bayesian treatment is natural \cite{FentonNeilBNs}.

\paragraph{Policy levers and calibration.}
Because curvature is explicit, parameters can be elicited from (i) observed breach cost escalation (fixing \(\gamma_\ell\), \(\alpha_\ell\)), (ii) measured deployment overheads (\(\gamma_h,\eta_1\)), and (iii) steganalysis pipeline costs/benefits for the warden (\(\gamma_a,\alpha_a,\beta_A,b_A\)). With these in hand, comparative statics map interventions (e.g., raising \(\tilde C^A_{\text{look}}\) through stronger internal controls or legal deterrents; lowering \(\tilde B^A_{\text{leak}}\) by data minimisation) to predictable shifts in \((p^\star,q^\star)\), and hence to conditional/unconditional success and expected loss. This aligns with security–economics practice and mirrors prior game-theoretic models of steganography where warden/encoder incentives are explicitly equilibrated \cite{Schottle2013,Schottle2016a}.

\subsection{Quantitative Evaluation}\label{sec:quantitative-eval}
This subsection provides a coherent quantitative study of the steganographic game that connects monetary primitives to observable strategic behaviour. Building on the nonlinear utility maps in \eqref{eq:def-concave-benefits}--\eqref{eq:def-adversary} and the equilibrium characterisation in Theorem~\ref{theorem_p}, we show how curvature in benefits and costs reshapes best responses and mixing. Figures~\ref{fig:F1_curvature}--\ref{fig:F5_distribution} supply (i) an intuitive understanding of the chosen mappings and (ii) empirical evidence regarding equilibrium probabilities, their curvature sensitivity, and their dispersion under parameter uncertainty.

\begin{figure*}[ht]
  \centering
  \includegraphics[width=0.9\linewidth]{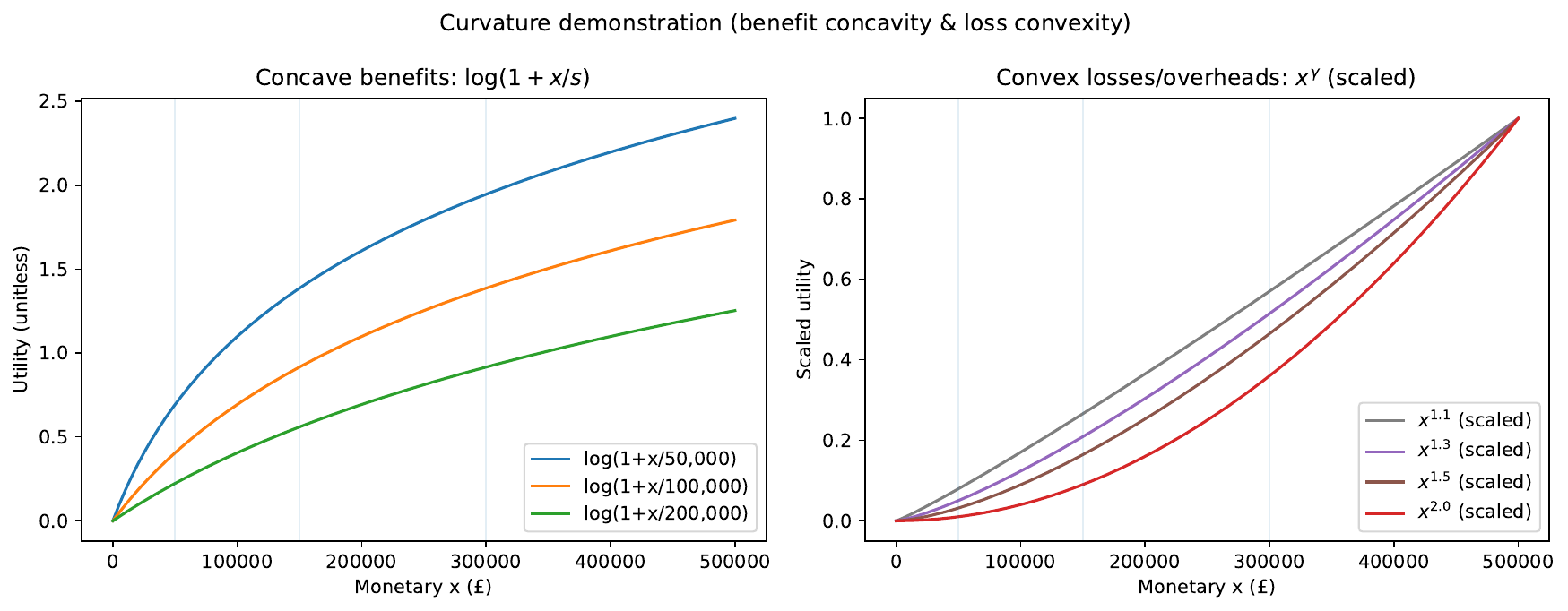}
  \caption{Utility transformations. \textbf{Left:} Concave benefit curves $x\mapsto\log(1+x/s)$ for three scales $s$ show diminishing marginal returns. \textbf{Right:} Convex cost curves $x\mapsto x^\gamma$ (normalised at £500 k) for four exponents illustrate accelerating disutility for large losses.}
  \label{fig:F1_curvature}
\end{figure*}

\begin{figure*}[ht]
  \centering
  \includegraphics[width=\linewidth]{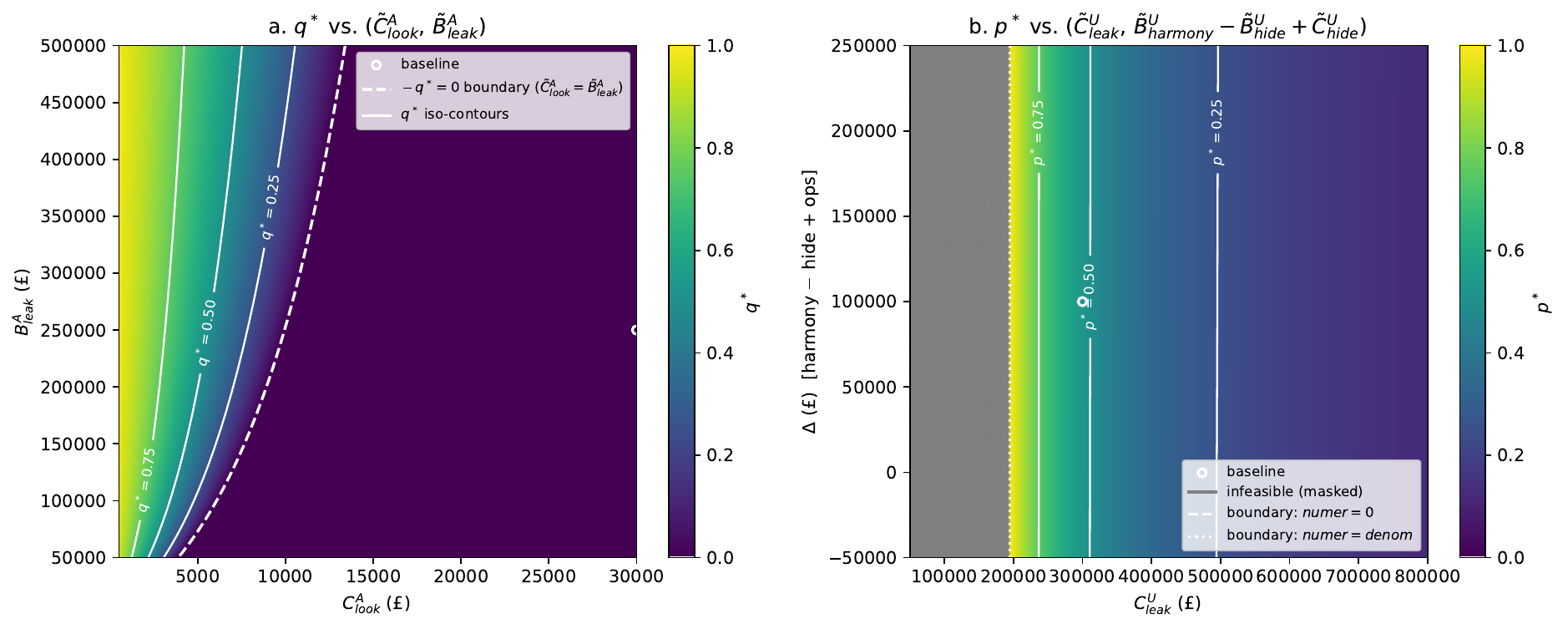}
  \caption{Equilibrium mixing probabilities under nonlinear utilities.
    \textbf{(a)} Defender’s hide–probability $q^\star$ vs.\ $(C^A_{\text{look}},\,B^A_{\text{leak}})$. White iso-contours show $q^\star\in\{0.25,0.50,0.75\}$; the white dashed curve marks the $q^\star=0$ boundary (i.e., where the attacker’s cost equals its benefit). 
    \textbf{(b)} Adversary’s search–probability $p^\star$ vs.\ $(C^U_{\text{leak}},\,\Delta)$ where $\Delta=\tilde B^U_{\text{harmony}}-\tilde B^U_{\text{hide}}+\tilde C^U_{\text{hide}}$. White iso-contours show $p^\star\in\{0.25,0.50,0.75\}$. The dashed white curve marks the boundary where the attacker’s net payoff numerator equals zero (no incentive to search), and the dotted curve marks the boundary where the attacker’s numerator equals the denominator (search probability would hit unity). Masked regions indicate non-interior regimes (i.e., pure-strategy corners).}
  \label{fig:F2_surfaces}
\end{figure*}

\begin{figure*}[ht]
  \centering
  \includegraphics[width=\linewidth]{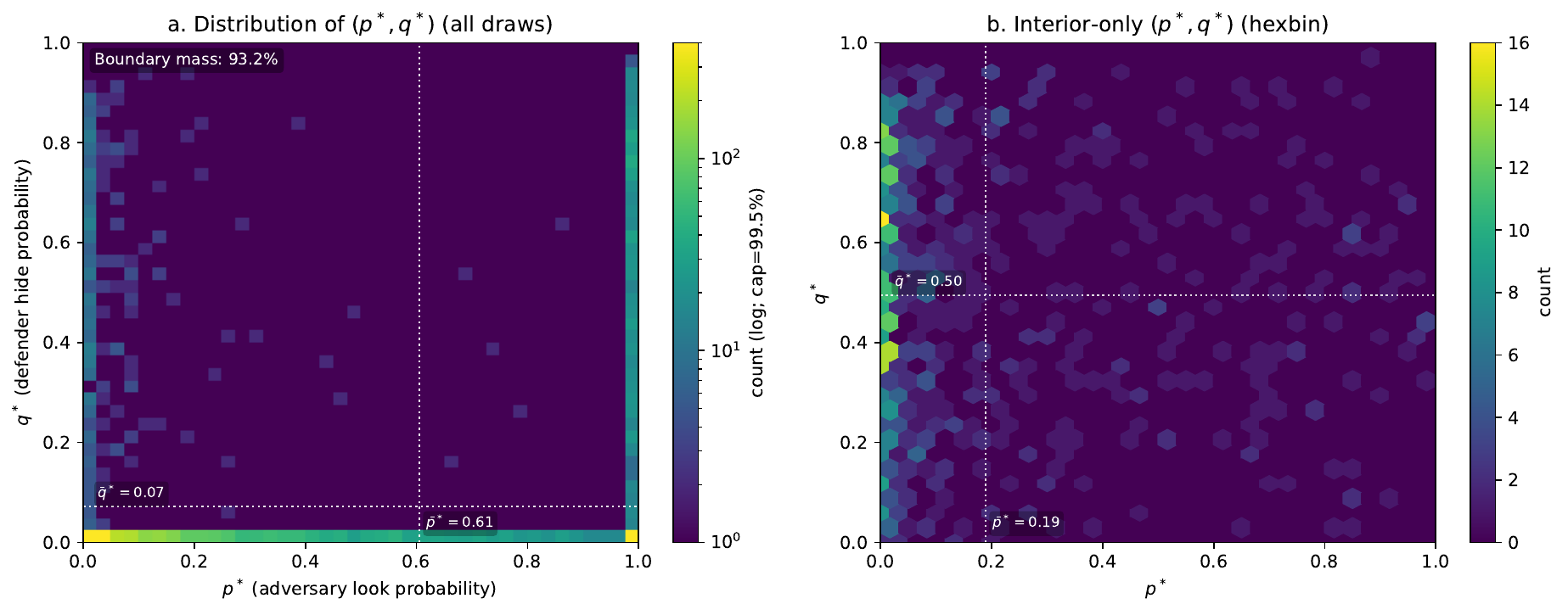}
  \caption{Monte Carlo distribution of equilibrium mixes under parameter uncertainty. \textbf{Left (a):} 2D histogram of $(p^\star,q^\star)$ over all draws; white dashed crosshairs show sample means $(\bar p^*,\bar q^*)$. \textbf{Right (b):} Hexbin for interior equilibria ($0<p^\star,q^\star<1$); white dashed crosshairs show interior means $(\bar p^*,\bar q^*)$.}
  \label{fig:F5_distribution}
\end{figure*}

\paragraph{Pipeline.}
The six practitioner-elicited \emph{monetary} inputs (annualised), i.e., $\{B^U_{\text{hide}},\,B^U_{\text{harmony}},\,C^U_{\text{hide}},\,C^U_{\text{leak}},\,B^A_{\text{leak}},\,C^A_{\text{look}}\}\!\in\!\mathbb{R}_{\ge 0}^6$, are transformed by concave (log) and convex (power) utilities governed by the curvature vector
\[
\theta=\{b_U,h_U,b_A,\beta_U,\delta_U,\beta_A,\alpha_\ell,\eta_0,\eta_1,\alpha_a,\gamma_\ell,\gamma_h,\gamma_a\}.
\]
For each $\theta$, the interior conditions of Theorem~\ref{theorem_p} are verified and the mixed-strategy solution from \eqref{equation_p}–\eqref{equation_q} are computed; if an interior condition fails, best responses lie on the boundary and probabilities are clamped to $[0,1]$ the run is tagged as a \emph{boundary-regime}
case (i.e., a non-interior equilibrium) and stored for later sensitivity review.

\paragraph{Calibration modes.}
To isolate the effect of risk attitudes from monetary magnitudes, two complementary baselines are considered:
\begin{enumerate}
  \item \textbf{Linear baseline:} Identity mappings (effectively $\gamma\!=\!1$ with very large scales) reproduce the textbook linear-payoff game for a sanity check.
  \item \textbf{Curvature-aware baseline:} Finite scales and exponents calibrated to practice,
   \[
    \gamma_\ell=1.5,\quad \gamma_h=1.2,\quad \gamma_a=1.2,
  \]
  with scale parameters $b_U=h_U=b_A=10^{5},\quad \alpha_\ell=10^{-6},\quad \eta_1=2\times 10^{-4},\quad \alpha_a=2\times 10^{-5}$. These values were selected to reflect diminishing marginal benefits and superlinear penalties, while keeping transformed utilities on comparable scales for equilibrium analysis.
\end{enumerate}
\paragraph{Parameter rationale.}
The adopted curvature and scale parameters reflect empirically grounded patterns in security economics and regulatory risk. The convex leak–loss exponent $\gamma_\ell\!=\!1.5$ models the well-documented superlinear growth of breach costs with incident size and regulatory exposure \cite{Romanosky2016,IBM2025CDBR,EDPBGuidelines2023Fines}. The moderate convexity values for operational and adversarial overheads ($\gamma_h\!=\!1.2$, $\gamma_a\!=\!1.2$) capture increasing marginal burdens on deployment and attacker effort, consistent with engineering cost curves and incident-response workload escalation \cite{GordonLoeb2002,Verizon2024}. Scale parameters such as $b_U\!=\!h_U\!=\!b_A\!=\!10^{5}$ and the small-rate coefficients $\alpha_\ell$, $\eta_1$, and $\alpha_a$ ensure that transformed utilities remain on comparable magnitudes, preserving interior mixed-strategy conditions while embodying diminishing returns and realistic penalty growth. These calibrated values therefore operationalise established economic and regulatory regularities within the equilibrium analysis.

Accordingly, Figures~\ref{fig:F1_curvature}--\ref{fig:F5_distribution} visualise (i) utility shapes, (ii) equilibrium mixing for $p^\star,q^\star$, (iii) curvature sensitivity, and (iv) Monte Carlo dispersion \cite{ross2022simulation} of equilibrium mixes under parameter uncertainty. By decoupling monetary inputs from behavioural curvature, this pipeline is consistent with NIST SP~800-30 risk practice and standard simulation methodology \cite{NIST80030,ross2022simulation}. 

\paragraph{Calibrated monetary primitives (scenario setup).}\label{scenario setup}
Consistent with quantitative risk practice, we parameterise defender and adversary \emph{monetary primitives} by empirically anchored \emph{ranges}, not single points. For total loss from a successful leak, we center an impact prior on IBM’s most recent \emph{Cost of a Data Breach} statistics (2025 global average \$4.44M), converting values into GBP using the contemporaneous 2025 exchange rate (USD\,1~$\approx$~GBP\,0.79).\footnote{Exchange rate sourced from the Bank of England (2025 average USD$\rightarrow$ GBP spot rate\url{https://www.bankofengland.co.uk/boeapps/database/Rates.asp?Travel=NIxAZx&into=USD&utm}).\label{fn:boe_exchange_rate}} The highest possible values at the maximum fine permitted are aligned with the GDPR administrative-fine cap (up to 4\% of worldwide annual turnover; Article~83) to capture rare but very large regulatory penalties often omitted in average-only models \cite{IBM2025CDBR,gdpr_article83_legislation}. Likelihood baselines and incident modes are cross-checked against the Verizon DBIR 2024 to keep base rates realistic for sector and actor mixes \cite{Verizon2024}. Where organisation-specific data are unavailable (e.g., confidentiality or GDPR constraints), we elicit expert distributions for response and hardening costs using structured expert judgment (Cooke’s classical model) with proven calibration and aggregation models \cite{ColsonCooke2018}. These distributions are propagated through the mixed-NE success rates (Sections~\ref{subsec:cond-success}–\ref{subsec:uncond-success}), and expected loss is evaluated in \S\ref{Risk_ana} \cite{NIST80030}.

\medskip
\noindent\emph{Concrete instantiation.}
For a medium-sized services organisation, we report \textbf{all monetary quantities in GBP (\pounds)} (with USD,1~$\approx$~GBP,0.79, based on the official exchange rate data cited in Footnote~\ref{fn:boe_exchange_rate}), and the following ranges for the \emph{expected} values of the random primitives:
\begin{align*}
    \mathbb{E}\!\left[B^{\mathcal U}_{\text{harmony}}\right] \in [1.5,\,2.5]\times10^{5}\,\pounds, 
    \mathbb{E}\!\left[C^{\mathcal U}_{\text{hide}}\right] \in [3.5,\,7.5]\times10^{4}\,\pounds, \\
    \mathbb{E}\!\left[B^{\mathcal U}_{\text{hide}}\right]  \in [1.0,\,2.0]\times10^{5}\,\pounds, 
    \mathbb{E}\!\left[C^{\mathcal U}_{\text{leak}}\right]   \in [2.0,\,6.0]\times10^{5}\,\pounds.
\end{align*}
Here, “$\mathbb{E}[X]\in[a,b]\times10^{k}\,\pounds$’’ means the mean lies between $a\times10^{k}$ and $b\times10^{k}$ \emph{pounds sterling}. Then, limits on extreme cost values are placed where necessary by applying Article 83 GDPR bound (i.e., administrative fines of up to 4 \% of a company’s worldwide annual turnover, whichever is higher) to reflect rare but very large regulatory penalties. These ranges satisfy (A2)–(A2$^\prime$) and preserve the best-response thresholds in \S\ref{Game_Analysis}. For the adversary, $C^{\mathcal A}_{\text{look}}$ and $B^{\mathcal A}_{\text{leak}}$ are parameterized analogously, respecting (A3)–(A3$^\prime$). Specifically, we set
\[
B^{\mathcal A}_{\text{leak}}=\rho\,C^{\mathcal U}_{\text{leak}},
\]
with \(\rho=0.83\).\footnote{Here \(\rho\) denotes the \emph{capturable fraction} of defender impact that a rational adversary can monetize (resale value, extortion yield, competitive gain), which empirical analyses of underground markets suggest is substantial yet typically below full impact \cite{Ablon2014HackersBazaar}}, so that with $C^{\mathcal U}_{\text{leak}}=\pounds 300{,}000$ we obtain $B^{\mathcal A}_{\text{leak}}\approx\pounds 250{,}000$ $(B^{\mathcal A}_{\text{leak}} \;=\; 0.83 \times \pounds 300{,}000 \;=\; \pounds 249{,}000 \approx \pounds 250{,}000)$; we set $C^{\mathcal A}_{\text{look}}=\pounds 30{,}000$, preserving the “easiest penetration’’ ordering (A3) and its utility-space counterpart (A3$^\prime$). 

\paragraph{Calibration evidence (brief).}
Anchors for breach economics use IBM’s 2025 hub (global average \$4.44M; 2024 \$4.88M) \cite{IBM2025CDBR} and UK trade-press synthesis indicating low seven-figure UK averages \cite{computerweekly_2025_uk_costs}. Regulatory tails follow Article~83 \cite{gdpr_article83_legislation}. Base-rate realism is aligned with Verizon DBIR 2024 \cite{Verizon2024}. Expert-judgment calibration and spending benchmarks follow \cite{ColsonCooke2018,GordonLoeb2002}.

Figure~\ref{fig:F1_curvature} provides an explicit visual account of the two core utility transforms introduced in Section~\ref{A_Steganographic_Game-Theoretic_Model}. The left panel plots three instances of the concave benefit map $x\mapsto \log(1+x/s)$ where the horizontal axis is the pre-transform monetary amount (in $\pounds$) and the vertical axis the resulting, dimensionless utility. Each line corresponds to a different scale parameter $s$: the blue line ($s=\text{\pounds}50\,$k) has the largest initial slope (approximate linear tangent near zero) and therefore the greatest marginal utility at small $x$, but it saturates fastest as $x$ grows; the orange ($s=\text{\pounds}100\,$k) and green ($s=\text{\pounds}200\,$k) lines display progressively gentler initial slopes and slower saturation. This family makes tangible the Arrow–Pratt intuition \cite{Pratt1964}: concavity implies diminishing absolute risk appetite and decreasing marginal returns to further increments of the same nominal gain. Operationally, a smaller $s$ models a decision-maker for whom modest monetary gains are relatively more valuable (steeper early marginal utility), whereas a larger $s$ models an actor for whom only sizeable gains materially move utility.

The right panel isolates loss and operational-overhead curvature by plotting $x\mapsto x^{\gamma}$ across a range of exponents $\gamma\in\{1.1,1.3,1.5,2.0\}$. For comparison, we normalise each power curve to equal one at the top of the monetary axis (£500k), so colour differences emphasise shape rather than scale. Near zero, all curves rise slowly; however, as $x$ increases, the larger exponents accelerate growth markedly. The red curve ($\gamma=2.0$) displays the steepest curvature: relatively moderate changes in monetary exposure at high magnitudes produce much larger increments in transformed disutility than they would under near-linear (low-$\gamma$) specifications. Intuitively, this captures regulatory and reputational regimes where large incidents produce disproportionately large costs (fines, cascading damage, loss of business), and therefore the transformed loss space is heavy-tailed and strongly nonlinear.

Together, these panels illustrate how curvature impacts equilibrium: concave benefit maps diminish marginal gains for added protection, while convex loss maps increase breach costs. The choice of scale $s$ and exponent $\gamma$ significantly affects the transformed payoffs that influence mixed-strategy formulas, impacting $p^\star$ and $q^\star$ as shown in the comparative statics.

Building on these primitives, Figure~\ref{fig:F2_surfaces} maps the defender’s hide‐mixing probability $q^\star$ and the adversary’s search‐mixing probability $p^\star$ (Theorem~\ref{theorem_p}) over key transformed inputs.  In panel (a), the horizontal axis is the adversary’s search cost $C^A_{\text{look}}$ (monetary), the vertical axis its leak benefit $B^A_{\text{leak}}$, both shown prior to transformation.  The colour scale encodes $q^\star=1-\tilde C^A_{\text{look}}/\tilde B^A_{\text{leak}}$, with darker hues indicating less frequent hiding.  The solid white contours at $q^\star\in\{0.25,0.50,0.75\}$ and the dashed boundary $\tilde C^A_{\text{look}}=\tilde B^A_{\text{leak}}$ clearly demarcate the regimes where the defender finds concealment increasingly unattractive.  Panel (b) plots $p^\star=(\Delta)/(\Delta+\tilde C^U_{\text{leak}})$ against the defender’s leak cost $C^U_{\text{leak}}$ and the net term $\Delta=\tilde B^U_{\text{harmony}}-\tilde B^U_{\text{hide}}+\tilde C^U_{\text{hide}}$.  The light‐grey masked region flags parameter combinations where $p^\star$ clips to the boundary, while white contours at $p^\star\in\{0.25,0.50,0.75\}$ and vertical dashed/dotted lines trace the indifference locus, confirming that higher leak costs or lower harmony benefits compel the adversary to search more aggressively.

To assess robustness under realistic uncertainty, Figure~\ref{fig:F5_distribution} presents a Monte‐Carlo ensemble of $n=10\,000$ draws over monetary and curvature hyperparameters (Appendix~\ref{Quantitative Evaluation setup}).  The left (panel F3a) is a two-dimensional histogram \emph{of all draws}, - that is, of all $10{,}000$ sampled parameter sets without excluding boundary cases. The high \emph{frequency} of outcomes very close to the boundaries $p^\star\in\{0,1\}$ or $q^\star\in\{0,1\}$ indicates that, for many plausible parameterisations, at least one player’s best response is a \emph{boundary solution} (e.g., “almost never” or “almost always”), rather than a genuinely mixed response.  Dotted lines denote the sample means $\bar p^\star$ and $\bar q^\star$, quantifying central tendency.  On the right (panel F3b), we restrict to interior equilibria ($0<p^\star,q^\star<1$) and display a hexbin plot that uncovers the dispersion of genuinely mixed strategies.  The skew toward low $p^\star$ and low $q^\star$ accords with the convex‐cost effects observed in Figures~\ref{fig:F2_surfaces}, thereby reinforcing that risk‐averse and superlinear cost structures naturally drive players to boundary or near‐boundary behaviour under parameter uncertainty.

Curvature-driven sensitivity of mixing probabilities is analyzed in §\ref{sec:sensitivity-mne}, where we sweep $(\gamma_\ell,\gamma_h,\gamma_a)$ and quantify local and global effects (Figure \ref{fig:F3_sensitivity}, Figure \ref{fig:S3_interior_mass}).

These results confirm that our nonlinear utility framework produces interpretable equilibrium patterns: cost convexity deters aggressive actions, concave benefit saturation moderates strategy adoption returns, and realistic parameter variability can shift equilibria toward pure strategies. These insights validate our model's relevance for risk management and policy design in surveilled networks.

\section{Sensitivity Analysis of Mixed Nash Equilibrium Solution}
\label{sec:sensitivity-mne}

This section develops comparative statics for the mixed–strategy Nash equilibrium of the steganographic game under the nonlinear utility mappings introduced in \S\ref{sec:assumptions} and analysed in \S\ref{Game_Analysis}. Our goal is to identify which \emph{transformed} payoffs move the equilibrium probabilities and in what direction, while preserving the policy-relevant orderings (A2)–(A3) and the utility–space refinements (A2$^\prime$)–(A3$^\prime$). We derive sensitivities in expected-utility space (where best responses are ordinally invariant), then map to monetary primitives via the chain rule to expose curvature effects, so that curvature is explicitly reflected in the sensitivity coefficients. This approach follows security-game and risk-analysis guidance on incentives under uncertainty \cite{@Manshaei2013}, convex fines and reputational effects justify convex loss utilities \cite{EDPBGuidelines2023Fines}, and diminishing returns reflect Arrow–Pratt risk aversion \cite{Pratt1964}.

Based on the equilibrium probabilities $(p^{*}, q^{*})$ expressed in Equations \ref{equation_p} and \ref{equation_q} from Theorem~\ref{theorem_p}, the following shorthand is introduced
\[
\Delta \;\equiv\; \tilde B^U_{\text{harmony}}-\tilde B^U_{\text{hide}}+\tilde C^U_{\text{hide}},
\qquad
D \;\equiv\; \tilde B^U_{\text{harmony}}+\tilde C^U_{\text{leak}},
\]
so $p^\star=\Delta/D$ and $q^\star=1-\tilde C^{A}_{\text{look}}/\tilde B^{A}_{\text{leak}}$, where the notation $\Delta$ denotes $\mathcal{U}$'s net baseline advantage of not hiding, and $D$ is $\mathcal{U}$'s exposure scale when not hiding.  
We report derivatives for the \emph{interior} case $0<p^\star,q^\star<1$; on the boundaries ($p^\star\in\{0,1\}$ or $q^\star\in\{0,1\}$) the usual envelope arguments imply zero or one–sided responses and clipping (cf.\ Figure~\ref{fig:F2_surfaces}).\footnote{Affine (positive) rescalings of utilities leave best responses and equilibria invariant; our use of expected utility and positive monotone transforms therefore preserves equilibrium structure \cite{osborne2004introduction}.}

Differentiating $p^\star=\Delta/D$ and $q^\star=1-\tilde C^{A}_{\text{look}}/\tilde B^{A}_{\text{leak}}$ yields
\begin{align}
\frac{\partial p^\star}{\partial \tilde B^U_{\text{harmony}}}
&=\frac{D-\Delta}{D^2}
=\frac{\tilde C^U_{\text{leak}}+\tilde B^U_{\text{hide}}-\tilde C^U_{\text{hide}}}{\bigl(\tilde B^U_{\text{harmony}}+\tilde C^U_{\text{leak}}\bigr)^2},
\label{eq:psens-Bh}\\[-0.2em]
\frac{\partial p^\star}{\partial \tilde B^U_{\text{hide}}}
&=-\frac{1}{D},
\,
\frac{\partial p^\star}{\partial \tilde C^U_{\text{hide}}}
=\frac{1}{D},
\,
\frac{\partial p^\star}{\partial \tilde C^U_{\text{leak}}}
=-\frac{\Delta}{D^2}
=-\frac{p^\star}{D},
\label{eq:psens-others}\\[0.3em]
\frac{\partial q^\star}{\partial \tilde B^{A}_{\text{leak}}}
&=\frac{\tilde C^{A}_{\text{look}}}{\bigl(\tilde B^{A}_{\text{leak}}\bigr)^2}
=\frac{1-q^\star}{\tilde B^{A}_{\text{leak}}},
\qquad
\frac{\partial q^\star}{\partial \tilde C^{A}_{\text{look}}}
=-\,\frac{1}{\tilde B^{A}_{\text{leak}}}.
\label{eq:qsens-utility}
\end{align}
Thus, ceteris paribus, higher \emph{transformed} leak disutility depresses the adversary’s search mix $p^\star$ (via $-p^\star/D<0$), while steeper transformed deployment overheads $\tilde C^U_{\text{hide}}$ raise $p^\star$ by enlarging the indifference numerator $\Delta$. For the defender, $q^\star$ is increasing in the attacker’s transformed benefit and decreasing in transformed search cost, exactly as the indifference condition for $q^\star$ implies. These directions are consistent with the orderings in (A2)–(A3) and the utility–space calibrations (A2$^\prime$)–(A3$^\prime$), which encode adequate protection for the defender and economically motivated attack for the adversary \cite{@Manshaei2013,Becker1968}. 

The Jacobian entries of the utilities transforms in \eqref{eq:def-concave-benefits}–\eqref{eq:def-adversary} are
\[
\begin{aligned}
&\frac{\partial \tilde B^U_{\text{harmony}}}{\partial B^U_{\text{harmony}}}=\frac{\delta_U}{h_U+B^U_{\text{harmony}}},
\quad
\frac{\partial \tilde B^U_{\text{hide}}}{\partial B^U_{\text{hide}}}=\frac{\beta_U}{b_U+B^U_{\text{hide}}},
\\
&\frac{\partial \tilde C^U_{\text{hide}}}{\partial C^U_{\text{hide}}}
=\eta_1\,\gamma_h\,(C^U_{\text{hide}})^{\gamma_h-1},
\quad
\frac{\partial \tilde C^U_{\text{leak}}}{\partial C^U_{\text{leak}}}
=\alpha_\ell\,\gamma_\ell\,(C^U_{\text{leak}})^{\gamma_\ell-1},
\\
&\frac{\partial \tilde B^{A}_{\text{leak}}}{\partial B^{A}_{\text{leak}}}
=\frac{\beta_A}{b_A+B^{A}_{\text{leak}}},
\quad
\frac{\partial \tilde C^{A}_{\text{look}}}{\partial C^{A}_{\text{look}}}
=\alpha_a\,\gamma_a\,(C^{A}_{\text{look}})^{\gamma_a-1}.
\end{aligned}
\]
Hence, the monetary–space sensitivities follow by the chain rule:
\begin{multline}
\frac{\partial p^\star}{\partial B^U_{\text{harmony}}}
=\frac{\partial p^\star}{\partial \tilde B^U_{\text{harmony}}}\cdot \frac{\delta_U}{h_U+B^U_{\text{harmony}}},\quad
\frac{\partial p^\star}{\partial B^U_{\text{hide}}}
\\=\frac{\partial p^\star}{\partial \tilde B^U_{\text{hide}}}\cdot \frac{\beta_U}{b_U+B^U_{\text{hide}}}.
\label{eq:p-monetary-1}
\end{multline}
\begin{multline}
\frac{\partial p^\star}{\partial C^U_{\text{hide}}}
=\frac{\partial p^\star}{\partial \tilde C^U_{\text{hide}}}\cdot \eta_1\,\gamma_h\,(C^U_{\text{hide}})^{\gamma_h-1},\quad
\frac{\partial p^\star}{\partial C^U_{\text{leak}}}
\\=\frac{\partial p^\star}{\partial \tilde C^U_{\text{leak}}}\cdot \alpha_\ell\,\gamma_\ell\,(C^U_{\text{leak}})^{\gamma_\ell-1}.
\label{eq:p-monetary-2}
\end{multline}

\begin{multline}
\frac{\partial q^\star}{\partial B^{A}_{\text{leak}}}
=\frac{\partial q^\star}{\partial \tilde B^{A}_{\text{leak}}}\cdot \frac{\beta_A}{b_A+B^{A}_{\text{leak}}}
=\frac{\tilde C^{A}_{\text{look}}}{\bigl(\tilde B^{A}_{\text{leak}}\bigr)^2}\cdot \frac{\beta_A}{b_A+B^{A}_{\text{leak}}}.
\label{eq:q-monetary-1}
\end{multline}

\begin{multline}
\frac{\partial q^\star}{\partial C^{A}_{\text{look}}}
=\frac{\partial q^\star}{\partial \tilde C^{A}_{\text{look}}}\cdot \alpha_a\,\gamma_a\,(C^{A}_{\text{look}})^{\gamma_a-1}
=-\frac{\alpha_a\,\gamma_a\,(C^{A}_{\text{look}})^{\gamma_a-1}}{\tilde B^{A}_{\text{leak}}}.
\label{eq:q-monetary-2}
\end{multline}
Equations \eqref{eq:p-monetary-1}–\eqref{eq:q-monetary-2} make transparent how the \emph{shape} parameters $(\gamma,\text{scales})$ attenuate or amplify monetary shocks: logarithmic benefits dampen large positive changes through the $1/(b+\cdot)$ factor, whereas power losses magnify large exposures through $(\cdot)^{\gamma-1}$ with $\gamma>1$ \cite{Pratt1964,EDPBGuidelines2023Fines}.

Considering practical interpretations and policy levers, the following are discussed:

(i) \emph{Defender levers (moving $p^\star$).} Increasing transformed leak consequences (e.g., via higher expected regulatory/contractual exposure or better quantification of reputational harm feeding into $\tilde C^U_{\text{leak}}$) reduces $p^\star$ by \eqref{eq:psens-others}; simultaneously, curbing operational convexity (lower $\gamma_h$ or $\eta_1$ in $\tilde C^U_{\text{hide}}$) pulls down $\Delta$ and thus $p^\star$. This is consistent with (A2) and with institutional guidance that stresses mission/business harms in risk treatment \cite{NIST80030,EDPBGuidelines2023Fines}.  

(ii) \emph{Adversary–facing levers (moving $q^\star$).} From \eqref{eq:qsens-utility}, higher transformed search costs (legal risk, detection risk, analyst time) lower $q^\star$ only insofar as they reduce the attacker’s incentive to look (since $q^\star$ is a response to $(\tilde C^A_{\text{look}},\tilde B^A_{\text{leak}})$). Conversely, larger transformed leak payoffs push $q^\star$ up, inducing the defender to hide more often. Economic models of offending predict such cost–benefit responses \cite{Becker1968}. 

(iii) \emph{Boundary and regime shifts.} The interior formulas apply on the feasible wedge $\Delta\in(0,D)$ for $p^\star$ and $\tilde C^A_{\text{look}}<\tilde B^A_{\text{leak}}$ for $q^\star$. Crossing these thresholds clips probabilities to $\{0,1\}$ and collapses sensitivities (Figure~\ref{fig:F2_surfaces}, Figure~\ref{fig:F3_sensitivity}). In practice, this means incentive programs or regulatory changes can trigger qualitative regime changes (e.g., pushing the attacker to \emph{never look} or the defender to \emph{always hide}). 

The derivative formulas \eqref{eq:psens-Bh}–\eqref{eq:q-monetary-2} are directly usable for \emph{local} sensitivity dashboards and priority setting: estimate transformed quantities (and their elasticities) from breach/incident, compliance, and cost data, and read off how marginal changes move $(p^\star,q^\star)$. They also rationalise the heatmaps in Figure~\ref{fig:F2_surfaces} and the dispersion patterns in Figure~\ref{fig:F5_distribution}, where interior vs.~boundary regimes emerge endogenously under parameter uncertainty \cite{@Manshaei2013}. 

\subsection{Interior comparative statics around the baseline}
\label{subsec:local-CS}
\begin{figure}[h!]
  \centering
  \includegraphics[width=0.95\linewidth]{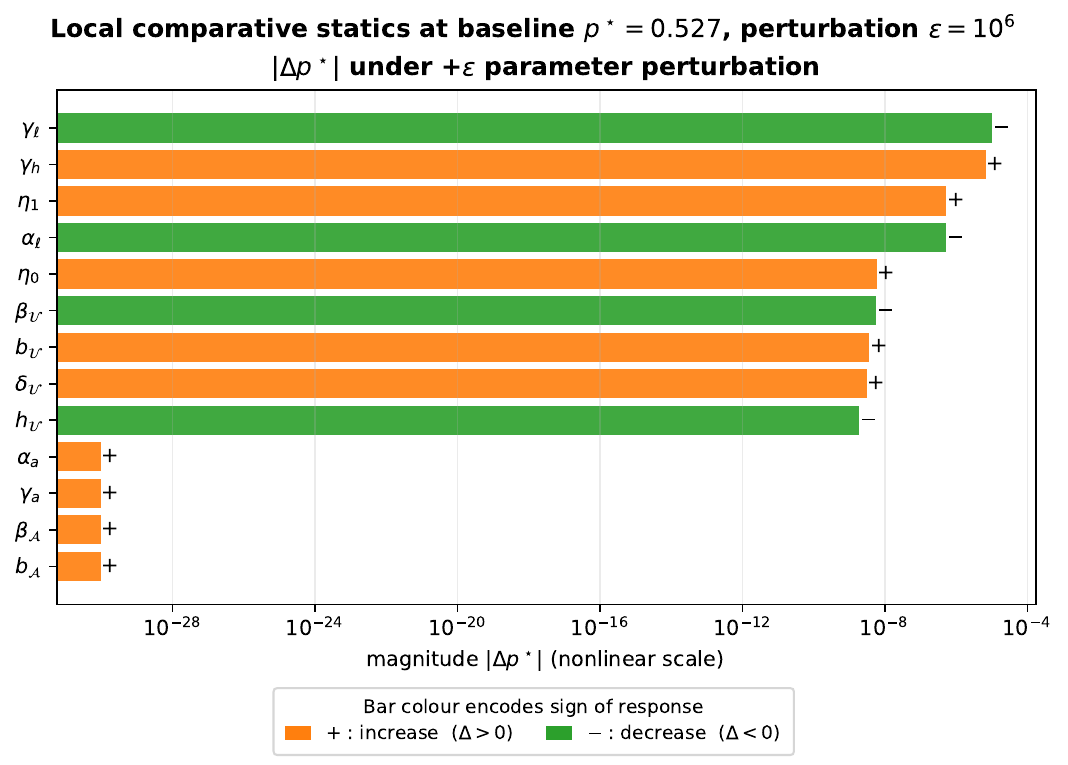}
  \caption{\textbf{Local comparative statics of the adversary’s equilibrium search probability $p^\star$ at the baseline.} Bars report the magnitude $|\Delta p^\star|$ induced by a small positive perturbation $\varepsilon$ to each transformed utility parameter, plotted on a logarithmic scale. Bar colours (orange and green) indicate the sign of the response ($+$ increase, $-$ decrease).}
  \label{fig:F6_cs_bars}
\end{figure}

We study how small perturbations to the transformed utilities shift the mixed–strategy equilibrium $(p^\star,q^\star)$ at the baseline. 
Figure~\ref{fig:F6_cs_bars} quantifies local comparative statics of the mixed-strategy equilibrium derived in \S\ref{MixNE} under small perturbations of the transformed utility primitives. Consistent with the response-surface structure in Figure~\ref{fig:F2_surfaces}, sensitivity of $p^\star$ is dominated by defender-side curvature and cost parameters ($\gamma_\ell$, $\gamma_h$, $\eta_1$), while attacker-side benefit and search-cost parameters induce changes several orders of magnitude smaller.

Importantly, the baseline equilibrium lies in a boundary regime as characterised in \S\ref{PureNashAnalysis} and \S\ref{Game_Analysis}: the defender’s optimal response satisfies $q^\star=0$. As a result, $q^\star$ is locally pinned to the boundary and admits no first-order response to parameter perturbations, whereas $p^\star$ remains interior and varies smoothly. The result in Figure~\ref{fig:F6_cs_bars} therefore, complements the interior comparative statics by explicitly illustrating how, in boundary regimes, equilibrium sensitivity concentrates on the unconstrained player while the boundary strategy remains locally insensitive.

\subsection{Quantitative Sensitivity Evaluation (Sensitivity Focus)}
\label{sec:quant_eval_sensitivity}
\begin{figure*}[h!]
  \centering
  \includegraphics[width=\linewidth]{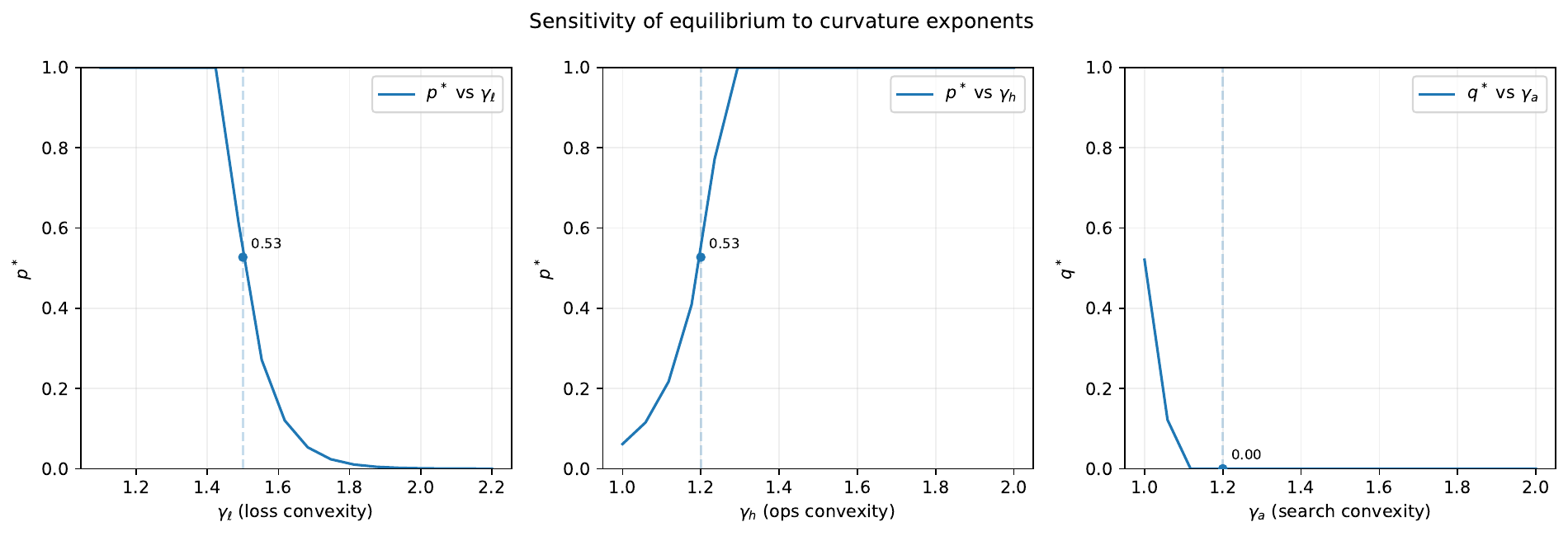}
  \caption{Curvature sensitivity of equilibrium mixing. \textbf{Left:} $p^\star$ vs.\ leak-loss convexity $\gamma_\ell$.
  \textbf{Middle:} $p^\star$ vs.\ deployment-overhead convexity $\gamma_h$.
  \textbf{Right:} $q^\star$ vs.\ search-cost convexity $\gamma_a$.
  Vertical dotted lines mark the baseline exponents used in our calibrated plots.}
  \label{fig:F3_sensitivity}
\end{figure*}

\begin{figure}[h!]
  \centering
  \includegraphics[width=0.9\linewidth]{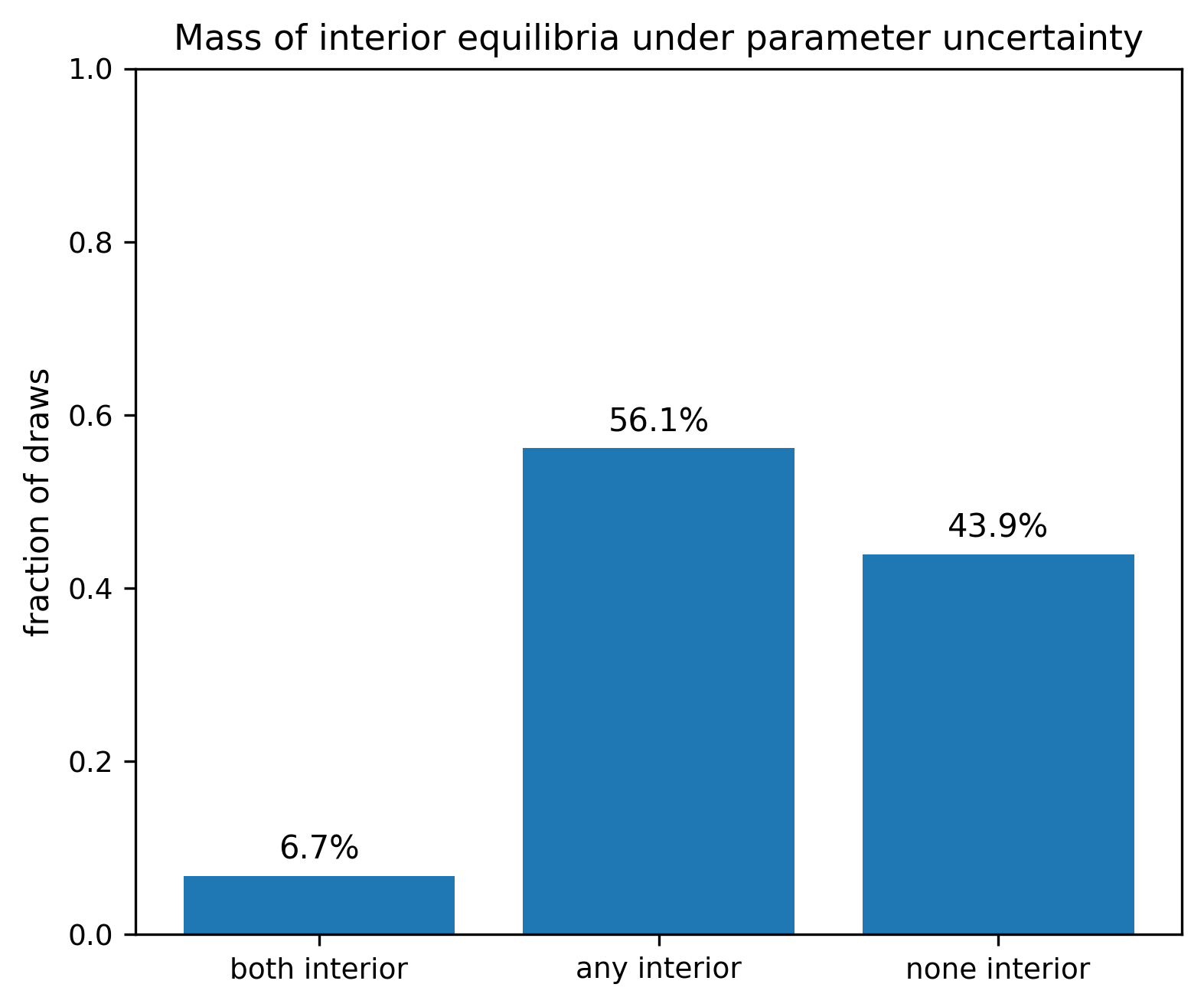}
    \caption{Prevalence of interior equilibria under parameter uncertainty. Bars show the proportion of Monte Carlo draws yielding: (i) \emph{both} players mixing ($0<p^\star,q^\star<1$); (ii) \emph{any} interior outcome (at least one player mixing); or (iii) \emph{none} (both players at boundary choices, i.e., pure responses). Percent labels above bars report sample percentages across all draws.}
  \label{fig:S3_interior_mass}
\end{figure}

\begin{figure*}[ht]
  \centering
  \includegraphics[width=\linewidth]{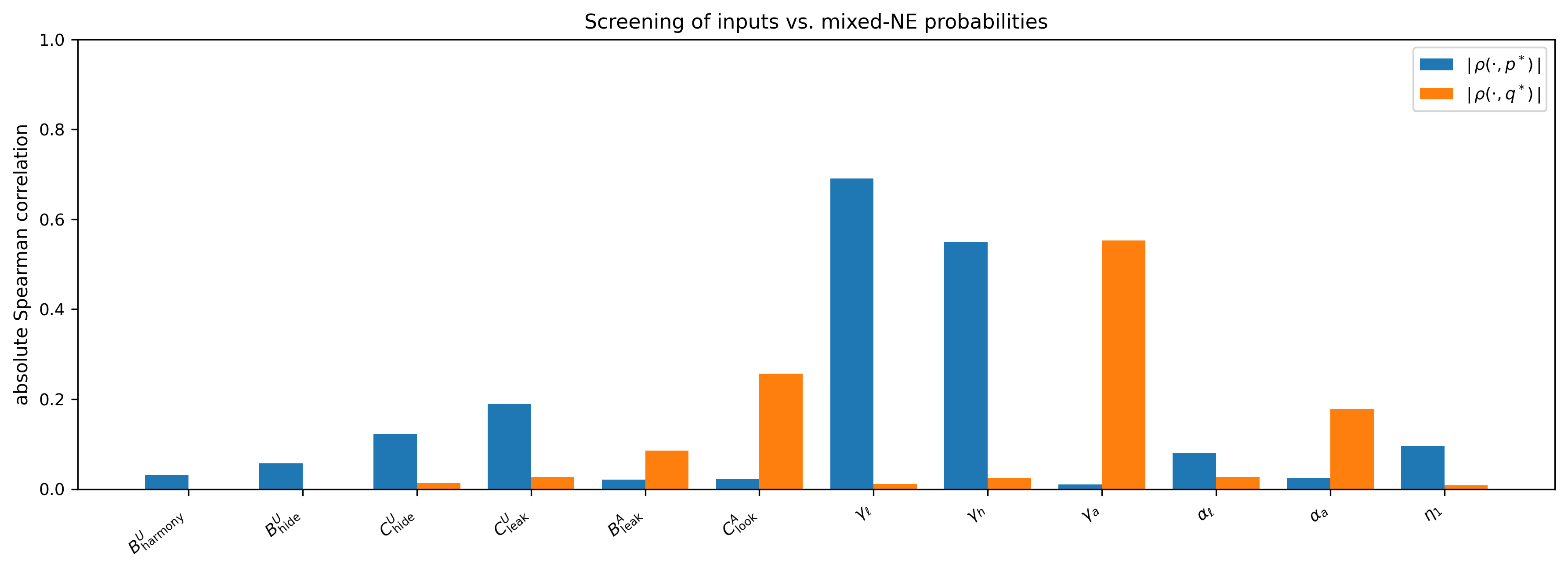}
  \caption{Screening of inputs against mixed-NE probabilities using absolute Spearman correlations. Bars compare $|\rho(\cdot,p^\star)|$ and $|\rho(\cdot,q^\star)|$ for each monetary primitive and curvature/scale parameter. Consistent with the theory, $p^\star$ is most sensitive to the defender-side convexities ($\gamma_\ell,\gamma_h$) and to $C^U_{\text{leak}}$, whereas $q^\star$ loads primarily on the attacker search convexity $\gamma_a$ and on $C^A_{\text{look}}$; the remaining inputs show weaker monotone influence.}
  \label{fig:S4_screening}
\end{figure*}
Building on the interior comparative statics derived above, we now examine how small changes in primitives and risk–attitude parameters shift the mixed–strategy equilibrium. Sensitivities are evaluated in expected–utility space, where positive affine rescalings preserve best responses and equilibria, so that the directions and relative magnitudes reported here are invariant to utility scaling (cf.\ Assumption~(A4).\footnote{Equilibrium predictions are invariant under positive affine re‐scalings of utility in standard game theory. See, e.g., Osborne’s textbook discussion of cardinal utility and ordinal invariance.} in Section~\ref{sec:assumptions}).

First, we vary the curvature exponents while holding the calibrated monetary primitives unchanged. Figure~\ref{fig:F3_sensitivity} shows that \(p^\star\) falls monotonically with the breach–loss convexity \(\gamma_\ell\) and rises with the deployment/overhead convexity \(\gamma_h\), while \(q^\star\) falls with the search–cost convexity \(\gamma_a\). These monotonicities agree with the interior derivatives implied by \(p^\star=\Delta/D\) and \(q^\star=1-\tilde C^A_{\mathrm{look}}/\tilde B^A_{\mathrm{leak}}\), and the vertical dotted lines mark the baseline exponents used elsewhere in our calibrated plots.

Second, we quantify the \emph{prevalence} of genuine mixing once joint uncertainty over both monetary and curvature/scale parameters is acknowledged. Using a Latin–hypercube design, Figure~\ref{fig:S3_interior_mass} reports the shares of draws that yield (i) both players mixing, (ii) any interior behaviour (at least one player mixing), and (iii) none (pure responses by both). In our calibration, interior regimes are common in aggregate but simultaneous interiority for both players is comparatively rare, indicating that boundary solutions (and hence clipped sensitivities) should be expected over large regions of the plausible parameter space.

Screening inputs against \((p^\star,q^\star)\) using absolute Spearman rank correlations, we capture \emph{global} monotone influence. Figure~\ref{fig:S4_screening} indicates that \(p^\star\) is dominated by the defender–side convexities \((\gamma_\ell,\gamma_h)\) and the monetary leak cost \(C^U_{\mathrm{leak}}\), whereas \(q^\star\) is driven primarily by the attacker’s search convexity \(\gamma_a\) and monetary search cost \(C^A_{\mathrm{look}}\). This aligns with the transformed best–response structure: convexification of losses magnifies \(\tilde C^U_{\mathrm{leak}}\), pushing \(p^\star\) down, while convex search costs magnify \(\tilde C^A_{\mathrm{look}}\), pushing \(q^\star\) down.

Taken together, the evidence suggests a practical hierarchy of levers: shaping \emph{how} costs grow (curvature) often moves equilibrium behaviour more than shifting \emph{how large} the costs are (levels). In particular, mechanisms that steepen breach-loss curvature depress \(p^\star\), while mechanisms that steepen search curvature depress \(q^\star\). This perspective informs the risk quantification in Section~\ref{Risk_ana}, where we propagate these equilibrium responses.

\section{Risk Analysis}
\label{Risk_ana}

This section presents a quantitative risk analysis of the steganographic game model introduced in Section~\ref{A_Steganographic_Game-Theoretic_Model}. 

Conventional deterministic risk formulations inadequately capture the inherent uncertainty in adversarial and networked environments. Indeed, real–world security risks are stochastic, context–dependent, and often nonstationary, making fixed‐probability estimates unreliable \cite{boyd2022assumptions, Cox2009}. To address this limitation, conditional–probability models have become central to modern cyber–risk assessment, providing a formal mechanism to express uncertainty about evolving threats and control effectiveness \cite{fenton2024risk}.

Recent advances incorporate probabilistic inference and machine learning into risk quantification. Bayesian frameworks have been widely applied to model uncertainty and update beliefs about attack likelihoods as new evidence arises \cite{Wei2022Application, uflaz2024quantifying, behbehani2023cloud, d2023including}, while neural approaches have been used to estimate dynamically security posture from real‐time telemetry and attack traces \cite{basit2023dynamic}. 

Building on these foundations, this present work introduces a game–theoretic risk formulation that improves the conventional product form $R=p\!\times\!I$ with a decision–conditioned metric. This formulation captures the time–varying interplay between attacker aggressiveness, defender concealment effectiveness, and the resulting operational loss potential. The computational procedure implementing this metric is summarised in Algorithm~\ref{alg:adv-risk}. This theoretical formulation underpins the empirical studies in Section~\ref{sec:risk_interpretation}, analysing the temporal evolution of various metrics to characterise the dynamic risk behaviour.

\begin{algorithm}[ht]
\caption{Advantage–based risk at epoch $t$}
\label{alg:adv-risk}
\DontPrintSemicolon
\footnotesize
\KwGoal{Compute the adversarial risk using
$R_t = I\cdot \max\{\mathsf{Adv}^{(+)}_{\mathcal{G},\mathcal{A}}(t),0\}$,
with the positive part advantage
$\mathsf{Adv}^{(+)}_{\mathcal{G},\mathcal{A}}(t)=\max\{\mathsf{Adv}^{(+)}_{\mathcal{G},\mathcal{A}}(t),0\}$ and
$\mathsf{Adv}^{(+)}_{\mathcal{G},\mathcal{A}}(t) := p^\star\!\left(\bar\beta_t^{\mathcal U}-\beta^{\mathcal U}\right)$.}

\Input{From the equilibrium aggressiveness $p^\star\in[0,1]$ (from \S\ref{MixNE});
effectiveness parameters $\beta_t^{\mathcal U},\ \bar\beta^{\mathcal U}\in[0,1]$ (with/without hiding) at epoch $t$; impact $I\ge 0$ (currency units) for the loss of a successful compromise.}

\Output{Risk $R_t\ge 0$, signed differential $\mathsf{Adv}^{(+)}_{\mathcal{G},\mathcal{A}}(t)$, and (optional) security benefit $B_{\text{sec},t}$.}

\BlankLine
\textbf{(Equilibrium step)}\;
\lIf{$p^\star$ is not given}{
$p^\star \gets \text{clip}_{[0,1]}\!\left(
\dfrac{\tilde B^{\mathcal U}_{\mathrm{harmony}}-\tilde B^{\mathcal U}_{\mathrm{hide}}+\tilde C^{\mathcal U}_{\mathrm{hide}}}
     {\tilde B^{\mathcal U}_{\mathrm{harmony}}+\tilde C^{\mathcal U}_{\mathrm{leak}}}
\right)$
}
\Indp \(\text{clip}_{[0,1]}\) truncates to the interval \([0,1]\). (cf.\ \S\ref{MixNE})\; \Indm

\textbf{(Sanity bounds)}\;
$\beta_t^{\mathcal U}\gets\text{clip}_{[0,1]}(\beta_t^{\mathcal U})$;\quad
$\bar\beta^{\mathcal U}\gets\text{clip}_{[0,1]}(\bar\beta^{\mathcal U})$.\;

\textbf{(Differential)}\;
$\mathsf{Adv}^{(+)}_{\mathcal{G},\mathcal{A}}(t) \gets p^\star\!\left(\beta_t^{\mathcal U}-\bar\beta^{\mathcal U}\right)$.\;

\textbf{(Positive–part advantage)}\;
$\mathsf{Adv}^{(+)}_{\mathcal{G},\mathcal{A}}(t)\gets \max\{\mathsf{Adv}^{(+)}_{\mathcal{G},\mathcal{A}}(t),0\}$.\;

\textbf{(Risk)}\;
$R_t \gets I\cdot \mathsf{Adv}^{(+)}_{\mathcal{G},\mathcal{A}}(t)$.\;

\textbf{(Optional: security benefit)}\;
\lIf{$\mathsf{Adv}^{(+)}_{\mathcal{G},\mathcal{A}}(t)<0$}{ $B_{\text{sec},t}\gets I\cdot \max\{-\mathsf{Adv}^{(+)}_{\mathcal{G},\mathcal{A}}(t),0\}$ }

\textbf{Return} $(R_t,\mathsf{Adv}^{(+)}_{\mathcal{G},\mathcal{A}}(t),B_{\text{sec},t})$.\;

\BlankLine
\textbf{Notes and diagnostics.}
\begin{enumerate}[label=(\alph*), leftmargin=2.2em]
  \item \emph{Monotonicity:} For $\mathsf{Adv}^{(+)}_{\mathcal{G},\mathcal{A}}(t)>0$, $R_t$ increases linearly in $I$ and $\beta_t^{\mathcal U}$, and decreases linearly in $\bar\beta^{\mathcal U}$, with slopes
  $\partial R_t/\partial \beta_t^{\mathcal U}=I\,p^\star$ and
  $\partial R_t/\partial \bar\beta^{\mathcal U}=-I\,p^\star$.
  \item \emph{Zero–risk regime:} If $\beta_t^{\mathcal U}\le \bar\beta^{\mathcal U}$ (no advantage for the adversary), then $\mathsf{Adv}^{(+)}_{\mathcal{G},\mathcal{A}}(t)\le 0$ and $R_t=0$ by construction.
  \item \emph{Numerical safeguards:} Clip probabilities to $[0,1]$. When computing $p^\star$ from utilities, guard denominators away from $0$ (e.g., add $10^{-12}$).
\end{enumerate}
\end{algorithm}

\subsection{Strategic Risk Analysis in Adversarial Contexts}\label{Adversarisk}

\paragraph{Purpose and alignment.}
We refine the classical quantitative-risk template $R= p \times I$ using a decision-conditioned construction \(R_t = I \times \mathrm{Adv}^{(+)}(t)\),
where the ``likelihood'' is the \emph{adversarial advantage} induced by the equilibrium attack propensity and the relative effectiveness of controls. Specifically, under the mixed strategy \(p^\star\) from \S\ref{MixNE}, and effectiveness parameters \(\beta_t^{\mathcal U},\bar\beta^{\mathcal U}\in[0,1]\) (hiding vs.\ baseline), we set
\(\mathrm{Adv}^{(+)}(t)=p^\star(\bar\beta^{\mathcal U}-\beta_t^{\mathcal U})\).
This retains the standard ``likelihood $\times$ impact'' form endorsed by NIST~SP~800-30, ISO/IEC~27000, and Open FAIR \cite{NIST80030,ISO27000,OpenFAIR}, but makes the likelihood \emph{explicitly game- and control-aware}, enabling principled comparisons across equilibrium mixes and defensive postures.

\paragraph{Risk metric choice (signed advantage, nonnegative risk).}
\begin{definition}[Adversarial Advantage]
\label{def:advantage}
Given the equilibrium aggressiveness of $\mathcal{A}$ (\S\ref{MixNE}) $p^\star\in[0,1]$, the \emph{time–varying adversarial advantage} quantifies the incremental benefit gained by the adversary at epoch~$t$ relative to the baseline, and is defined as
\begin{align}\label{eq:AdvPlus-def}
\mathsf{Adv}^{(+)}_{\mathcal{G},\mathcal{A}}(t)
&:= \Pr_{\mathcal A}(\mathrm{Succ}\mid \lnot \mathrm{Hide}_{\mathcal U})
      - \Pr_{\mathcal A}(\mathrm{Succ}\mid \mathrm{Hide}_{\mathcal U};t) \nonumber\\
&= p^\star\bigl(\bar\beta^{\mathcal U}-\beta_t^{\mathcal U}\bigr), \qquad \quad t=1,2,\dots,T.
\end{align}
\end{definition}

where \(\beta_t^{\mathcal U}\in[0,1]\) denotes the time-varying steganographic effectiveness (probability that hiding \emph{prevents} compromise at epoch \(t\)), and \(\bar\beta^{\mathcal U}\in[0,1]\) denotes the baseline (no-stego) effectiveness of other controls (cf.\ \S\ref{subsec:cond-success}). From \eqref{eq:adv-hide-cond}–\eqref{eq:adv-nohide-cond},
\begin{align*}
\Pr_{\mathcal A}\!\bigl((\mathrm{Succ} \mid \lnot (\mathrm{Hide}_{\mathcal U};\bigr)
    &= p^\star \bigl(1-\bar\beta^{\mathcal U}\bigr),\\
\Pr_{\mathcal A}\!\bigl((\mathrm{Succ} \mid (\mathrm{Hide}_{\mathcal U};\,t\bigr)
    &= p^\star \bigl(1-\beta_t^{\mathcal U}\bigr).
\end{align*}  
Multiplication by the equilibrium attack probability $p^\star$
scales this difference by the likelihood that the adversary chooses to attack. Thus \(\mathsf{Adv}^{(+)}(t)>0\) when hiding underperforms baseline (i.e., \(\beta_t^{\mathcal U}<\bar\beta^{\mathcal U}\)),
\(\mathsf{Adv}^{(+)}(t)=0\) when both are equally effective, and \(\mathsf{Adv}^{(+)}(t)<0\) when hiding dominates baseline.
Note that \(\mathsf{Adv}^{(+)}(t)\) depends on \(p^\star\) and the \emph{relative} gap
\(\bar\beta^{\mathcal U}-\beta_t^{\mathcal U}\) but not on the defender’s mixing \(q^\star\), because it compares the two branches directly.

Let \(I\ge 0\) denotes the (currency) impact. We report a nonnegative \emph{operational risk}:
\begin{equation}\label{eq:Risk-nonneg}
R_t \;=\; I\cdot \max\!\left\{\mathsf{Adv}^{(+)}_{\mathcal{G},\mathcal{A}}(t),\,0\right\},
\end{equation}
where $\max\{\mathsf{Adv}^{(+)}_{\mathcal{G},\mathcal{A}}(t),0\}$ denotes “take the larger of the two numbers $\mathsf{Adv}^{(+)}_{\mathcal{G},\mathcal{A}}(t)$ and $0$.”  Concretely,
\[
\max\{\mathsf{Adv}^{(+)}_{\mathcal{G},\mathcal{A}}(t),0\} = 
\begin{cases}
\mathsf{Adv}^{(+)}_{\mathcal{G},\mathcal{A}}(t), & \text{if }\mathsf{Adv}^{(+)}_{\mathcal{G},\mathcal{A}}(t)>0,\\
0,        & \text{if }\mathsf{Adv}^{(+)}_{\mathcal{G},\mathcal{A}}(t)\le0.
\end{cases}
\]
Thus, $\max\{\mathsf{Adv}^{(+)}_{\mathcal{G},\mathcal{A}}(t),0\}$ extracts the \emph{positive part} of $\mathsf{Adv}^{(+)}_{\mathcal{G},\mathcal{A}}(t)$ ensuring the \emph{risk} $R_t$ is never negative.  By contrast, $\min\{\mathsf{Adv}^{(+)}_{\mathcal{G},\mathcal{A}}(t),0\}$ would extract the \emph{negative part} of $\mathsf{Adv}^{(+)}_{\mathcal{G},\mathcal{A}}(t)$ (and be nonpositive), which does not align with the conventional interpretation of risk as a nonnegative expectation of loss.  Instead, the negative part $\min\{\mathsf{Adv}^{(+)}_{\mathcal{G},\mathcal{A}}(t),0\}$ is captured separately as the \emph{security benefit}
\begin{equation}\label{Benefit}
   B_{\mathrm{sec},t}
\;=\;
I\cdot\max\{-\mathsf{Adv}^{(+)}_{\mathcal{G},\mathcal{A}}(t),0\}
\;=\;
I\cdot\bigl(-\min\{\mathsf{Adv}^{(+)}_{\mathcal{G},\mathcal{A}}(t),0\}\bigr), 
\end{equation}
This approach preserves the standards-aligned nonnegativity of risk \cite{NIST80030,ENISA_RM_Standards} while retaining the sign information via a separate upside measure.

Moreover, if baseline controls are ineffective, \(\bar\beta^{\mathcal U}\) can be taken close to \(1\).
Then for any epoch with \(\beta_t^{\mathcal U}<\bar\beta^{\mathcal U}\), we have
\(\mathsf{Adv}^{(+)}_{\mathcal G,\mathcal A}(t)>0\).
Moreover, for two epochs \(t_1,t_2\) with \(\beta_{t_2}^{\mathcal U}<\beta_{t_1}^{\mathcal U}\),
\[
\mathsf{Adv}^{(+)}_{\mathcal G,\mathcal A}(t_2)-\mathsf{Adv}^{(+)}_{\mathcal G,\mathcal A}(t_1)
= p^\star\bigl(\beta_{t_1}^{\mathcal U}-\beta_{t_2}^{\mathcal U}\bigr) \;>\; 0,
\]
so the curve rises as detectors improve.

For completeness, we document in Appendix~\ref{app:risk-variants} the alternative \emph{absolute} variant --- %
$\mathsf{Adv}^{\mathsf{attack}}_{\mathcal{G},\mathcal{A}}(t)
:= \mathsf{Adv}^{(\mathrm{abs})}_{\mathcal{G},\mathcal{A}}(t)
=\bigl|\Delta_{t}\bigl|$ with $R^{(\mathrm{abs})}_t \;=\; I \cdot \mathsf{Adv}^{\mathsf{attack}}_{\mathcal{G},\mathcal{A}}(t) \;\ge 0$
--- and show how it collapses the sign information (useful for ranking scenarios by \emph{magnitude} only). In this paper, the main text, plots, and tables use the signed construction above.

\paragraph{Unconditional (overall) success and expected loss.}
Building on the conditional success definitions in \S\ref{subsec:cond-success}, we aggregate over the defender’s mixed action with \(\Pr(\mathrm{Hide}_{\mathcal U})=q^\star\) and \(\Pr(\lnot \mathrm{Hide}_{\mathcal U})=1-q^\star\) yields the \emph{unconditional} (overall) adversary success probability in \eqref{eq:adv-uncond}.
The bracketed term is a convex combination of the miss–detection rates in the two branches and is therefore in \([0,1]\). Monotonicities follow immediately: \(\partial \Pr_{\mathcal A}(\mathrm{Succ};t)/\partial p^\star\ge 0\), \(\partial \Pr_{\mathcal A}(\mathrm{Succ};t)/\partial \beta_t^{\mathcal U}\le 0\), and \(\partial \Pr_{\mathcal A}(\mathrm{Succ};t)/\partial \bar\beta^{\mathcal U}\le 0\); moreover \(\partial \Pr_{\mathcal A}(\mathrm{Succ};t)/\partial q^\star = p^\star(\bar\beta^{\mathcal U}-\beta_t^{\mathcal U})\), which has the same sign as the effectiveness gap between baseline and hiding. 

Let $I_t \ge 0$ be the monetary loss if a compromise occurs at time $t$. We make the standard assumption that, once $t$ is fixed, the \emph{size} of this loss does not depend on whether the attack succeeds (i.e., $I_t$ is independent of the success event conditional on $t$). Under this assumption, the expected loss separates into “likelihood $\times$ impact”:
\begin{multline}
\mathbb{E}[L_t]
\;=\;\mathbb{E}[I_t]\cdot \Pr_{\mathcal A}\!\bigl(\mathrm{Succ};t\bigr)\\
= \mathbb{E}[I_t]\cdot p^\star\Big[q^\star(1-\beta_t^{\mathcal U}) + (1-q^\star)(1-\bar\beta^{\mathcal U})\Big].
\label{eq:ELt-main}
\end{multline}
This approach matches the quantitative risk convention form advocated in operational guidance (e.g., NIST SP~800-30) \cite{NIST80030}.

When independence is too strong, — e.g., if the impact distribution depends on whether the attack succeeded under hiding vs.\ no hiding, or on the attack type jointly influenced by \(q^\star\)— one can replace \eqref{eq:ELt-main} by a conditional form or by a Bayesian network that encodes the dependencies among \(\{\mathrm{Hide}_{\mathcal U},\,\text{Succ},\,I_t\}\). A simple refinement is
\begin{multline}
 \mathbb{E}[L_t]
\;=\;
\mathbb{E}\!\left[I_t \,\middle|\, \mathrm{Succ}\right]\Pr_{\mathcal A}(\mathrm{Succ};t),
\quad\\ \text{or}\\ \quad
\mathbb{E}[L_t]
\;=\;
q^\star\,\mathbb{E}\!\left[I_t \,\middle|\, \mathrm{Succ},\mathrm{Hide}\right]\Pr_{\mathcal A}(\mathrm{Succ}\mid \mathrm{Hide};t) \\
+
(1-q^\star)\,\mathbb{E}\!\left[I_t \,\middle|\, \mathrm{Succ},\lnot\mathrm{Hide}\right]\Pr_{\mathcal A}(\mathrm{Succ}\mid \lnot \mathrm{Hide};t),
\label{eq:ELt-conditional}   
\end{multline}
or, more generally, one can propagate expectations through a Bayesian network that links technical effectiveness \((\beta_t^{\mathcal U},\bar\beta^{\mathcal U})\), strategy mix \((p^\star,q^\star)\), attack type/severity, and loss \(I_t\) (cf.\ standard total–probability decompositions  \cite{hubbard2023measure,RFProbLoss,FentonNeilBNs}). This retains the clarity of the total-probability aggregation in \eqref{eq:adv-uncond}, while allowing impact to vary with pathway and adversary response. The computation of epoch-level risk and corresponding security benefit follows the advantage-based formulation detailed in Algorithm \ref{alg:adv-risk}.

\paragraph{Why refine $R=\text{probability}\times\text{impact}$ in adversarial settings.}
Classical guidance frames cyber risk as a combination of the likelihood of a loss event and its impact, and prioritises the reduction of expected loss through controls with the best marginal benefit per unit cost. In our game, the “likelihood” piece is not a single scalar: it \emph{changes with the defender’s action}. The baseline “no hide” branch and the “hide” branch lead to different adversary success rates, which we already model via the conditional probabilities in \eqref{eq:adv-hide-cond}–\eqref{eq:adv-nohide-cond}. Compressing these two branches into a single $p$ obscures the very policy question at stake—\emph{should we hide now?}—and can distort prioritisation. The differential in Equation \eqref{eq:AdvPlus-def} exposes that policy lever directly: it is the causal lift (or drag) that hiding provides against compromise at time $t$. The mapping \eqref{eq:Risk-nonneg} and \eqref{Benefit} reports \emph{nonnegative risk of forgoing hiding}, while $B_{\text{sec},t}$ quantifies the \emph{security benefit} when hiding strictly helps. This keeps “risk” aligned with established standards (nonnegative likelihood/consequence construction) yet preserves signed information for decision-making.

\paragraph{Governance and decision thresholds.}
The pair $(R_t,B_{\text{sec},t})$ supports clear governance rules. Because $\mathsf{Adv}^{(+)}_{\mathcal{G},\mathcal{A}}(t)=p^\star(\beta_t^{\mathcal U}-\bar\beta^{\mathcal U})$, the boundary $\mathsf{Adv}^{(+)}_{\mathcal{G},\mathcal{A}}(t)=0$ is a transparent technical threshold: hiding and not hiding are equivalent \emph{on likelihood}. Above the threshold ($\mathsf{Adv}^{(+)}_{\mathcal{G},\mathcal{A}}(t)>0$), the control increases adversary success and the organisation faces a positive \emph{risk of not hiding} equal to $R_t$; below it ($\mathsf{Adv}^{(+)}_{\mathcal{G},\mathcal{A}}(t)<0$), the organisation accrues a \emph{security benefit} of $B_{\text{sec},t}$. Practically, a risk committee can set a minimal acceptable risk-reduction per cost and approve deployment only when $R_t/\text{Cost}(\mathrm{Hide})$ exceeds that threshold—fully in line with enterprise risk management practices that tie investment to demonstrable risk reduction.

\paragraph{Prioritisation and portfolio selection.}
By separating the branches and using $\mathsf{Adv}^{(+)}_{\mathcal{G},\mathcal{A}}(t)$, one can compare heterogeneous controls on a common “risk-reduction per unit spend” axis. Controls that primarily improve baseline hardening (raising $\bar\beta^{\mathcal U}$) will shrink $\mathsf{Adv}^{(+)}_{\mathcal{G},\mathcal{A}}(t)$ from the right; controls that specifically improve stego efficacy (raising $\beta_t^{\mathcal U}$) push $\mathsf{Adv}^{(+)}_{\mathcal{G},\mathcal{A}}(t)$ up from the left. A rational portfolio budgets first for items with the largest $\mathbb{E}[I]\cdot \mathsf{Adv}^{(+)}_{\mathcal{G},\mathcal{A}}(t)$ per dollar (or per unit of operational friction), a workflow consistent with quantitative risk analysis programs that emphasise expected-loss triage and economically motivated control selection.

\paragraph{Time variation, monitoring, and model risk.}
Because $\beta_t^{\mathcal U}$ and $\bar\beta^{\mathcal U}$ evolve with attacker learning and defender tuning, $\mathsf{Adv}^{(+)}_{\mathcal{G},\mathcal{A}}(t)$ should be monitored like any key risk indicator. Operationally, one can estimate the branch miss–detection rates from telemetry and maintain posterior distributions for $\beta_t^{\mathcal U}$ and $\bar\beta^{\mathcal U}$ (e.g., Beta–Bernoulli updates or Bayesian-network models when dependencies exist). This yields credible intervals for $\mathsf{Adv}^{(+)}_{\mathcal{G},\mathcal{A}}(t)$, and hence for $R_t$ and $B_{\text{sec},t}$, supporting “decision with uncertainty” dashboards rather than point estimates only.

\section{Empirical Calibration and Validation}
\label{sec:empirical}
This section presents the empirical evaluation of the proposed risk analysis framework. The goal is twofold: first, to empirically calibrate the model across different stego methods and payloads to determine optimal hyperparameters and learning configurations; and second, to validate the robustness and generalisation capability of the detector under varying embedding conditions.

\subsection{Experimental Setup \& Implementation Details}
\label{sec:exp-setup}

\paragraph{Detector and codebase.}
To mirror common practice and hedge model bias, detection is evaluated with two small CNNs (GNCNN Deep learning for steganalysis and Xu-Net–style) trained in PyTorch with GPU support~\cite{xu2016structural, qian2015deep}. Our detector preserves XuNet’s input--output interface (high-pass prefiltering, convolutional blocks with batch normalisation, global average pooling, and a final log-softmax head), and we explore training controls for curriculum learning and lightweight hyperparameter sweeps, selecting the best learning rate.

To align with the adversarial defender framework and ensure robustness across multiple stego methods, modifications were introduced to the GNCNN. The implementation now supports extended metric logging, AUC and $\mathrm{TPR}@\mathrm{FPR}=0.10$ ( \(\mathrm{TPR}@\mathrm{FPR}{=}0.10\) is the true positive rate (TPR) measured at the operating point where false positive rate (FPR) is fixed at \(0.10\)), from which the defender’s measure $\beta_t^{\mathcal{U}} = 1 - \mathrm{TPR}$ is derived. The dataset loader was enhanced to automatically match cover–stego pairs, preventing runtime interruptions. Checkpoint loading was modernised for PyTorch~2.6 compatibility, ensuring consistent restoration of weights and optimisers. To avoid CUDA reinitialisation errors, the training pipeline enforces the \texttt{spawn} multiprocessing method with \texttt{num\_workers=0}. Random seeds, logging, and checkpointing were standardised to guarantee full reproducibility across runs and configurations.

\paragraph{Dataset and stego generation.}
All experiments use \textbf{BOSSbase 1.01}. For each embedding method: \textbf{WOW}, \textbf{S-UNIWARD}, \textbf{HILL}, and \textbf{MiPOD} stego images are generated at payloads \(\{0.400, 0.200, 0.100\}\) bits per pixel. We form disjoint training/validation splits (typical configurations use \(\sim\)4,000 pairs for training and \(\sim\)1,000 pairs for validation per method/payload \cite{ntivuguruzwa2023convolutional}), with a separate held-out test set for final reporting. To make the pipeline robust to partial or missing stego sets, the loader \emph{intersects} cover/stego filenames and skips missing pairs without aborting.

\paragraph{Training protocol.}
Training uses a batch size of 16 for up to 150 epochs with $(\mathrm{StepLR})$ learning rate decay (step size 30, $(\gamma=0.5))$. The Adam optimiser is supported \cite{kingma2014adam}, and weight decay is enabled by default. Early stopping monitors validation AUC, with a patience window set by default to 20 epochs. Additionally, a \emph{curriculum across payloads} is implemented: training is conducted sequentially from payload 0.400 to 0.200 to 0.100, with each stage initialised from the best checkpoint of the previous (higher) payload. This cross-payload transfer has been empirically shown to benefit all methods considered \cite{boroumand2018deep,qian2015deep,xu2016structural,Kodovsky2012Ensemble,fridrich2012rich,xie2023novel}.

\paragraph{Metrics and defender signal.}
Per epoch, training/validation loss, accuracy, ROC-AUC, and \(\mathrm{TPR}\) at a fixed FPR \(=0.10\) are recorded. Following the defender’s objective, we define
\[
\beta_t^{\mathcal{U}} \;=\; 1 - \mathrm{TPR}@\mathrm{FPR{=}0.10},
\]
and log it jointly with AUC. We propagate the bootstrap distribution of $\mathrm{TPR}_{\text{op}}$ through these algebraic maps to obtain Confidence intervals (CIs) for $\mathsf{Adv}^{(+)}_{\mathcal{G},\mathcal{A}}(t)$ and $R_t$ (Algorithm~\ref{alg:adv-risk}). The \emph{operating point} \(\tau_{\text{op}}\) ( \(\tau\) denotes a score threshold) is defined as any threshold satisfying \(\mathrm{FPR}(\tau_{\text{op}})=0.10\). We then write $\mathrm{TPR}_{\text{op}}\;:=\;\mathrm{TPR}(\tau_{\text{op}})$ and $\mathrm{TPR}@\mathrm{FPR{=}0.10\;:=\;\mathrm{TPR}_{\text{op}}},$ i.e., the detector’s sensitivity when the false alarm rate is held at \(10\%\) along the ROC curve. This mapping follows the conceptual interpretation used in Sections~\ref{MixNE} and~\ref{Risk_ana} where \(\beta^{\mathcal U}\) denotes the conditional probability that embedding prevents compromise. Also, snapshot system statistics (CPU\%, RAM usage, and GPU memory) are provided to contextualise throughput and resource needs during long sweeps.


\paragraph{Hyperparameters.}
The learning rate (\(\mathrm{lr}\)), weight decay (\(\mathrm{wd}\)), and optimiser settings are explored via command-line flags for easy configuration. For \textsc{S-UNIWARD} and \textsc{WOW}, the best settings consistently use the Adam optimiser with learning rates in the range \(\mathrm{lr}\in[1\times10^{-4},\,2\times10^{-4}]\) and weight decay \(\mathrm{wd}\in\{5\times10^{-4},\,10^{-3}\}\). Under the same search space, \textsc{HILL} and \textsc{MiPOD} exhibit similar optima and do not require method-specific adjustments. The process proceeds as follows: (i) Fix the optimiser to Adam. (ii) On the \(0.400\) bpp payload, perform a grid search over \(\mathrm{lr}\in\{5\times10^{-5},\,1\times10^{-4},\,2\times10^{-4}\}\) and \(\mathrm{wd}\in\{5\times10^{-4},\,10^{-3}\}\), selecting the configuration with the best validation AUC/\(\mathrm{TPR}@\mathrm{FPR}=0.10\). (iii) Initialise training for the \(0.200\) and \(0.100\) payloads from the best \(0.400\) checkpoint and fine-tune with the same optimiser and regularisation settings.


\subsection{Empirical Risk Behaviour and Interpretation}
\label{sec:risk_interpretation}

\begin{figure*}[ht]
  \centering
  \includegraphics[width=\linewidth]{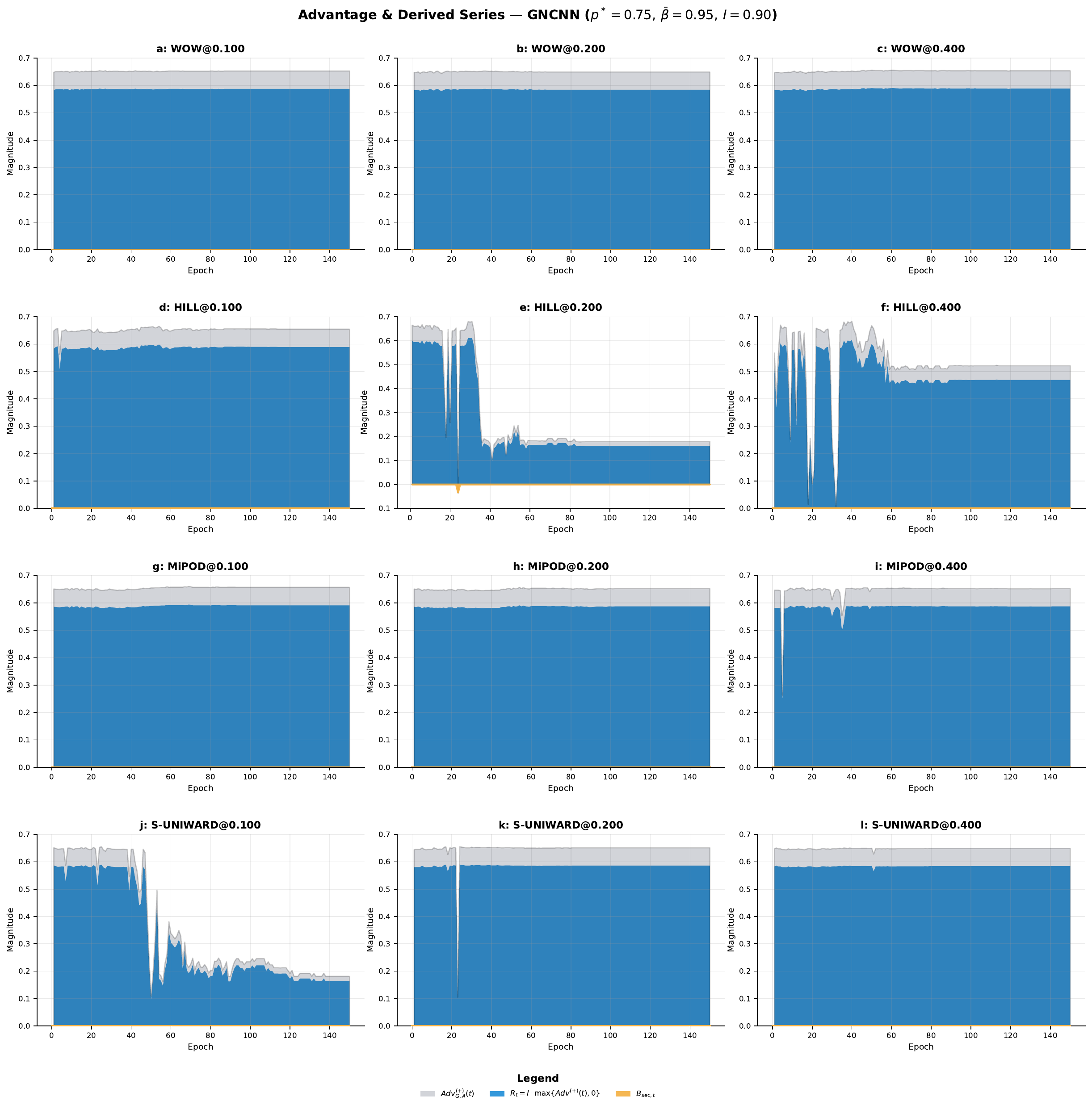}
  \caption{\small
  \textbf{Empirical evolution of adversarial advantage and risk for GNCNN.}
  Each subplot corresponds to a payload and embedding method combination, showing
  the empirical advantage $\mathsf{Adv}^{(+)}_{\mathcal G,\mathcal A}(t)$ (transparent light grey fill),
  the induced risk $R_t = I\!\cdot\!\max\{\mathsf{Adv}^{(+)}(t),0\}$ (blue shading above zero),
  and the security–benefit term $B_{\mathrm{sec},t} = I\!\cdot\!\max\{-\mathsf{Adv}^{(+)}(t),0\}$
  (orange shading below zero, or an orange baseline when $B_{\mathrm{sec},t}\equiv 0$).
  For the parameter setting $(p^\star=0.75,\ \bar\beta^{\mathcal U}=0.95,\ I=0.90)$, the empirical
  advantage remains non-negative across epochs, so $B_{\mathrm{sec},t}$ is identically zero and
  appears only as an orange baseline at $0$. Flat or slowly varying advantage profiles indicate
  stable but limited improvements in GNCNN's steganalytic capability.}

  \label{fig:panel-adv-risk-gncnn}
\end{figure*}

\begin{figure*}[ht]
  \centering
  \includegraphics[width=\linewidth]{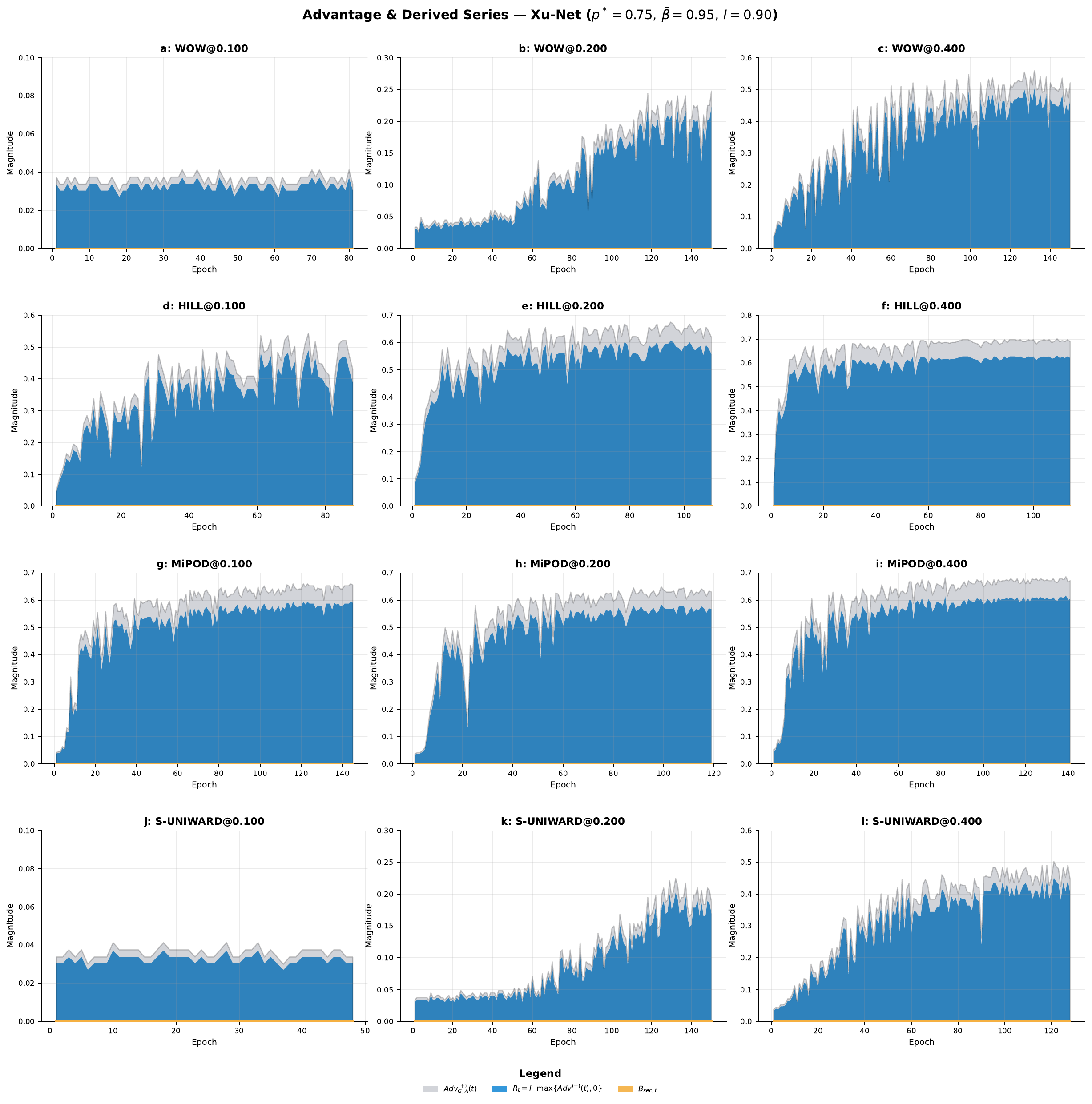}
  \caption{
  \textbf{\small Empirical evolution of adversarial advantage and risk for Xu--Net.} Each subplot corresponds to a payload and embedding method combination, with notation identical to Figure~\ref{fig:panel-adv-risk-gncnn}. Xu--Net displays higher temporal volatility and larger advantage amplitudes, particularly for WOW, HILL, and MiPOD at 0.200–0.400~bpp, demonstrating a stronger sensitivity of $\beta_t^{\mathcal U}$ to embedding rate and confirming the theoretical scaling of $R_t$ with $|\bar\beta_t^{\mathcal U}-\beta_t^{\mathcal U}|$.}
  \label{fig:panel-adv-risk-xunet}
\end{figure*}

\begin{figure*}[h!]
  \centering
  \includegraphics[width=0.85\linewidth]{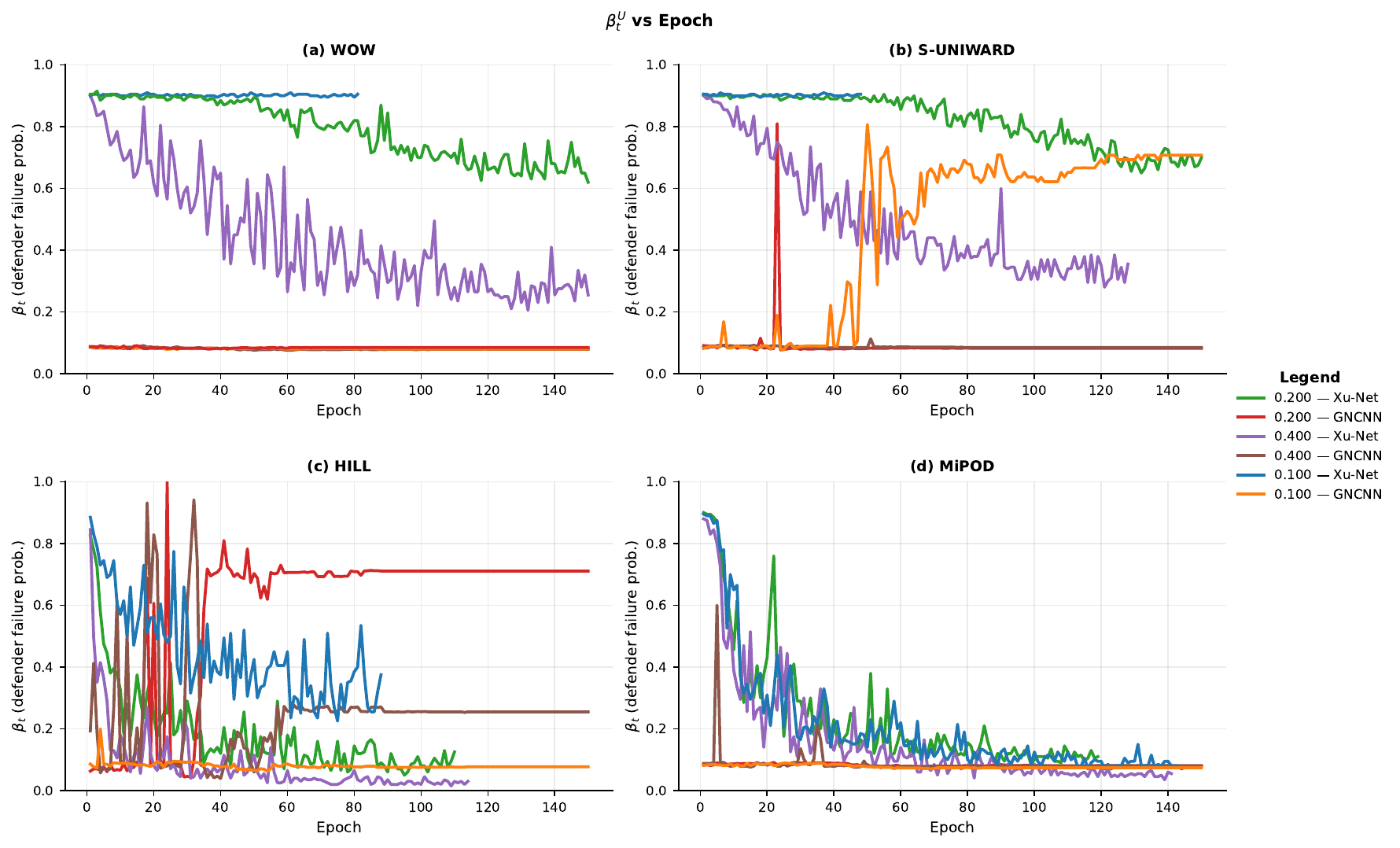}
    \caption{\textbf{\small
$\beta_t^{\mathcal U}$ vs Epoch (FPR target $=0.10$).}\small Temporal evolution of defender failure probability $\beta_t^{\mathcal U}$ 
  (FPR target $=0.10$) across embedding payloads and models for (a)~\textsc{WOW}, 
  (b)~\textsc{S-UNIWARD}, (c)~\textsc{HILL}, and (d)~\textsc{MiPOD}. 
  Each curve shows the epoch-wise $\beta_t^{\mathcal U}$ under identical training conditions for Xu-Net and GNCNN. 
  Xu-Net displays greater volatility—especially for \textsc{WOW} and \textsc{S-UNIWARD}—with convergence strongly payload-dependent. 
  GNCNN maintains lower and more stable $\beta_t^{\mathcal U}$ trajectories across all payloads, reflecting smoother convergence and reduced sensitivity to embedding variance. These results substantiate the equilibrium interpretation in Section~\ref{subsec:cond-success}, linking detector reliability and defender failure probability to the adversarial advantage term $\mathsf{Adv}^{(+)}_{\mathcal G,\mathcal A}(t)$.}
   \label{fig:beta-vs-epoch-all}
\end{figure*}

\begin{figure*}[t]
  \centering
  \includegraphics[width=0.85\linewidth]{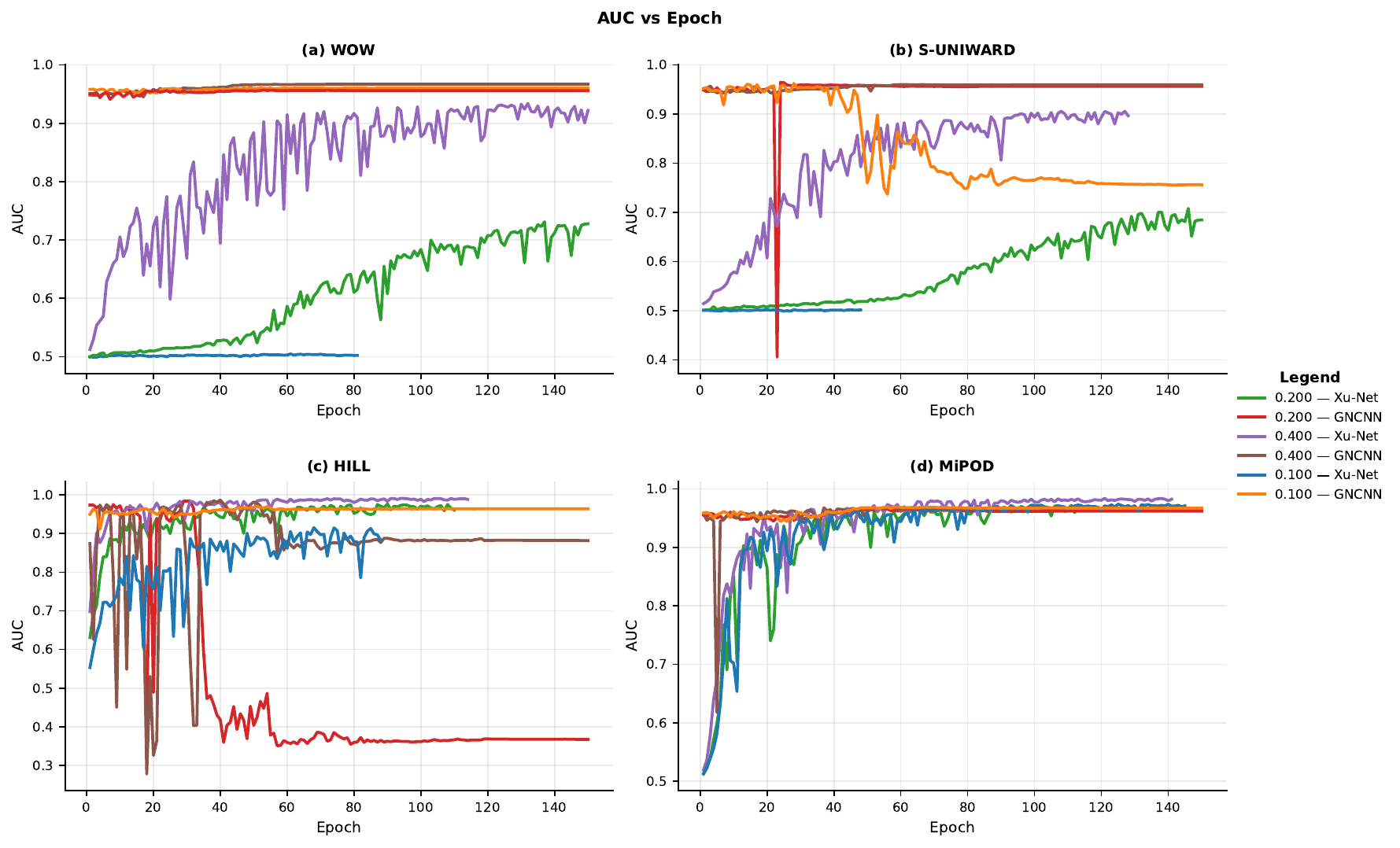}
     \caption{\small Evolution of AUC (Area Under the ROC Curve) over epochs for each embedding method and payload under both Xu-Net and GNCNN architectures. 
  Xu-Net generally exhibits faster early improvement but greater variance, particularly at low payloads for \textsc{WOW} and \textsc{S-UNIWARD}, reflecting higher sensitivity to texture complexity. 
  GNCNN, by contrast, converges smoothly with consistently high AUC values, indicating stable detection performance across payloads. 
  These AUC trajectories complement the $\beta_t^{\mathcal U}$ profiles in Figure~\ref{fig:beta-vs-epoch-all}, jointly illustrating the trade-off between discrimination stability and concealment effectiveness in mixed-strategy equilibria.}
   \label{fig:auc-vs-epoch-all}
\end{figure*}

\begin{figure*}[t]
  \centering
  \includegraphics[width=0.85\linewidth]{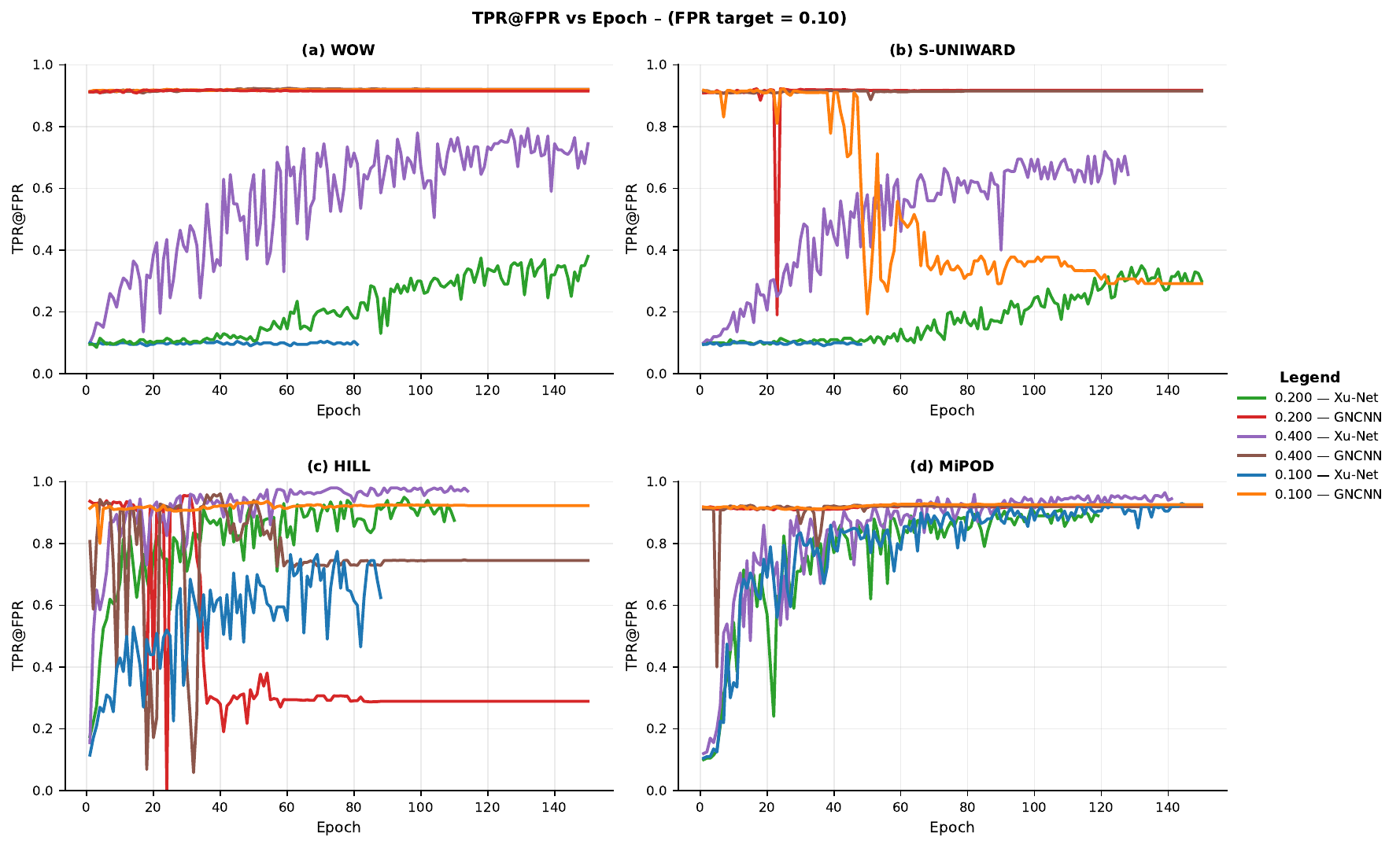}
     \caption{\small Evolution of $\mathrm{TPR@FPR}$ (True Positive Rate at fixed False Positive Rate target $=0.10$) across epochs for both Xu-Net and GNCNN detectors under different embedding methods and payloads. 
  GNCNN consistently maintains high $\mathrm{TPR@FPR}$ values with minimal variance, signifying stable detection sensitivity under constant false alarm constraints. 
  Xu-Net exhibits larger oscillations—especially for \textsc{WOW} and \textsc{S-UNIWARD}—indicating greater adaptability but also transient instability during early training. 
  These dynamics complement the AUC and $\beta_t^{\mathcal U}$ trends (Figs.~\ref{fig:auc-vs-epoch-all} and~\ref{fig:beta-vs-epoch-all}), providing a clearer view of detector robustness and convergence behaviour within the mixed-strategy learning process.}
   \label{fig:tpr-at-fpr-all}
\end{figure*}

Figures~\ref{fig:panel-adv-risk-gncnn} and~\ref{fig:panel-adv-risk-xunet} present the empirical evolution of the 
adversarial advantage $\mathsf{Adv}^{(+)}_{\mathcal G,\mathcal A}(t)$ and the corresponding derived quantities 
$R_t = I \!\cdot\! \max\{\mathsf{Adv}^{(+)}(t),0\}$ and 
$B_{\mathrm{sec},t} = I \!\cdot\! \max\{-\mathsf{Adv}^{(+)}(t),0\}$ over training epochs for all four embedding methods (i.e., WOW, S--UNIWARD, HILL, MiPOD) and payloads 
$\{0.100, 0.200, 0.400\}$~bpp. Both figures are produced under identical calibration parameters $(p^\star=0.75,\ \bar\beta^{\mathcal U}=0.95,\ I=0.90)$, ensuring comparability across models. The results show how the differential likelihood of compromise (captured by $\mathsf{Adv}^{(+)}(t)$) translates into the operational risk trajectory $R_t$ during the learning process of the steganalytic detectors.

\paragraph{Calibration rationale (high baseline effectiveness and impact).}
Mainstream risk guidance models expected loss as a function of \emph{likelihood} and \emph{impact}, and classifies \emph{high} impact as severe or catastrophic effects on confidentiality, integrity or availability; thus choosing $I$ near unity simply instantiates the “high” end of that scale.\footnote{NIST \emph{Guide for Conducting Risk Assessments} states that risk is a function of likelihood and impact; see SP~800-30, Rev.~1, §1.1 (p.~7). \url{https://csrc.nist.rip/publications/detail/sp/800-30/rev-1/final}. NIST FIPS~199 defines the \emph{high} impact level for security categorization; see \S2.2–2.3. \url{https://csrc.nist.gov/pubs/fips/199/final}.} Accordingly, modern CNN-based steganalysis on standard corpora (e.g., BOSSBase) routinely reports low error at moderate payloads (e.g., $0.4$\,bpp) for methods like WOW and S-UNIWARD; as training stabilises or models are fused, operating points with very high detection rates are attainable, supporting a baseline “no-hiding” effectiveness $\bar\beta^{\mathcal U}$ close to one.\footnote{See, e.g., Boroumand–Chen–Fridrich’s SRNet (deep residual steganalysis) demonstrating strong performance on spatial stego at 0.4\,bpp \cite[p.~1]{boroumand2018deep} (\url{https://ieeexplore.ieee.org/stamp/stamp.jsp?arnumber=8470101}); recent deep architectures and ensembles further reduce error on WOW/S-UNIWARD at 0.4\,bpp, \url{https://peerj.com/articles/cs-616/}, and subsequent CNN variants reporting $\leq\!0.15$ error at 0.4\,bpp on BOSSBase).}
Under this lens, $(\bar\beta^{\mathcal U},I)=(0.95,0.90)$ is a defensible “upper-end” stress-test for operational risk curves: it reflects a capable defender (high baseline detector effectiveness when \emph{not} hiding) and a system whose losses are categorised as high if compromise succeeds, while still allowing the empirical learning dynamics of $\beta_t^{\mathcal U}$ to modulate realised risk across epochs.

\paragraph{Global pattern.}
Across both models, the blue shaded area (risk) tracks closely with the dashed purple curve (advantage), 
confirming the theoretical relationship $R_t \propto \mathsf{Adv}^{(+)}(t)$ from \eqref{eq:Risk-nonneg}.  
Positive advantage values where $\bar\beta_t^{\mathcal U}>\beta_t^{\mathcal U}$ signify epochs in which 
the adversary gains relative to the defender, leading to elevated instantaneous risk.  
In contrast, negative advantages (where hiding outperforms baseline) would correspond to 
nonzero security benefits $B_{\mathrm{sec},t}$, which remain visually absent in most runs, 
indicating that steganographic concealment rarely reversed the adversarial edge 
under the given payload and model configurations.
\paragraph{Model comparison (GNCNN vs.\ Xu--Net).}
The first panel (Figure~\ref{fig:panel-adv-risk-gncnn}) shows that GNCNN exhibits predominantly flat or slowly rising 
risk surfaces across epochs, with $\mathsf{Adv}^{(+)}(t)$ stabilising near a constant positive offset.  
This behaviour reflects the limited representational capacity of GNCNN---its filters quickly converge to a 
steady detection regime, leaving the attacker’s effective advantage largely unchanged.  
Xu--Net, by contrast (Figure~\ref{fig:panel-adv-risk-xunet}), reveals pronounced temporal variation and sharper upward 
slopes, especially at higher payloads (0.200–0.400~bpp).  
These oscillations correspond to dynamic interactions between $\beta_t^{\mathcal U}$ (the defender’s 
steganographic effectiveness) and the training progress of deeper convolutional filters, 
illustrating that stronger models yield more expressive and volatile adversarial landscapes.

\paragraph{Effect of embedding method and payload.}
Within each model, the choice of embedding algorithm modulates both the amplitude and persistence of $\mathsf{Adv}^{(+)}(t)$.  
For GNCNN, WOW and MiPOD exhibit consistently high advantage values ($\approx0.25$–$0.30$), implying that non-steganographic baselines dominate even at modest payloads. S--UNIWARD and HILL, however, produce intermittent reductions in advantage are visible as transient dips or plateaus which align with epochs where the detector momentarily fails to generalise, reducing adversary success.  
For Xu--Net, the relationship reverses in some regimes: as payload increases, advantage and risk rise together, confirming the theoretical expectation that higher embedding rates (larger message sizes) amplify the detectable distortion 
and thus the adversary’s marginal benefit. The monotonic dependence of $R_t$ on $(p^\star, I, |\bar\beta_t^{\mathcal U}-\beta_t^{\mathcal U}|)$ predicted by \eqref{eq:Risk-nonneg} is empirically validated.

\paragraph{Interpretive implications.}
The alignment between $\mathsf{Adv}^{(+)}(t)$ and $R_t$ across epochs demonstrates that the proposed game-theoretic risk metric captures the operational vulnerability dynamics in a principled and interpretable manner. In equilibrium, when $p^\star$ reflects optimal adversary aggressiveness (\S\ref{MixNE}), a sustained positive advantage signifies persistent leakage risk despite steganographic use which is consistent with limited $\beta_t^{\mathcal U}$ evolution. Conversely, decreasing $\mathsf{Adv}^{(+)}(t)$ trajectories correspond to effective defender learning, 
as the embedding process recovers concealment performance over time. Thus, the experimental curves directly operationalise the theoretical notion that risk under adversarial equilibrium is not merely a product of event frequency, but of the 
\emph{strategic gap} between baseline and concealment effectiveness.

Table~\ref{tab:adv-risk-summary} summarises the epoch-aggregated statistics of adversarial advantage and derived risk for representative $(p^\star,\bar{\beta}^{\mathcal U},I)$ calibrations. The table condenses each experimental trajectory into its mean, median, and peak values of $\mathsf{Adv}^{(+)}_{\mathcal G,\mathcal A}(t)$ and $R_t$, together with counts of epochs classified as “risk” (positive advantage) or “benefit” (negative advantage).
Consistent with Figs.~\ref{fig:panel-adv-risk-gncnn}–\ref{fig:panel-adv-risk-xunet}, GNCNN exhibits relatively stable yet persistently positive advantages across payloads, indicating sustained adversarial leverage, whereas Xu–Net produces higher volatility and greater variation in $\max R_t$, reflecting its deeper network sensitivity to payload and embedding method.
The aggregated statistics therefore complement the temporal plots by quantifying the relative persistence and magnitude of adversarial risk within each steganographic regime.

\begin{table*}[h!]
\centering
\scriptsize
\setlength{\tabcolsep}{16pt}
\renewcommand{\arraystretch}{1.0}

\sisetup{
  round-mode          = places,
  round-precision     = 2,
  group-separator     = {,},
  group-minimum-digits= 2,
  table-number-alignment = center
}

\caption{Condensed summary of adversarial advantage and derived risk metrics across models, methods, and payloads under calibration $(p^\star,\bar{\beta}^{\mathcal U},I)=(\{0.25,0.50,0.75\},\,0.95,\,0.90)$. Values report mean adversary advantage, maximal risk $\max R_t$, the epoch at which it occurs, and counts of benefit/risk epochs (top entries per method by $\max R_t$).}
\label{tab:adv-risk-summary}
\vspace{0.3em}

\resizebox{0.98\linewidth}{!}{%
\begin{tabular}{
  l                                                 
  l                                                 
  S[table-format=1.0,table-number-alignment=center] 
  S[table-format=1.0,table-number-alignment=center] 
  S[table-format=1.0,table-number-alignment=center] 
  S[table-format=1.0,table-number-alignment=center] 
  S[table-format=1.1,table-number-alignment=center] 
  S[table-format=4.0,table-number-alignment=left] 
}
\toprule
\multicolumn{1}{c}{Model} &
\multicolumn{1}{c}{Method} &
\multicolumn{1}{c}{Payload} &
\multicolumn{1}{c}{$\langle\mathrm{Adv}^{(+)}\rangle$} &
\multicolumn{1}{c}{$\max R_t$} &
\multicolumn{1}{c}{$\mathrm{epoch}(\max R_t)$} &
\multicolumn{1}{c}{$\langle B_{\mathrm{sec},t}\rangle$} &
\multicolumn{1}{c}{Risk epochs} \\
\midrule
\multirow{8}{*}{GNCNN}
  & HILL      & 0.100 & 0.218 & 0.200 &  54 & 0.0 &    150 \\
  & HILL      & 0.200 & 0.092 & 0.204 &  29 & 0.0 &    149 \\
  & MiPOD     & 0.100 & 0.218 & 0.198 &  69 & 0.0 &    150 \\
  & MiPOD     & 0.200 & 0.217 & 0.197 &  54 & 0.0 &    150 \\
  & S-UNIWARD & 0.100 & 0.117 & 0.196 &  25 & 0.0 &    150 \\
  & S-UNIWARD & 0.200 & 0.216 & 0.196 &  17 & 0.0 &    150 \\
  & WOW       & 0.100 & 0.217 & 0.196 &  24 & 0.0 &    150 \\
  & WOW       & 0.200 & 0.216 & 0.196 &  36 & 0.0 &    150 \\
\cmidrule(lr){1-8}
\multirow{8}{*}{Xu-Net}
  & HILL      & 0.100 & 0.125 & 0.163 &  75 & 0.0 &     88 \\
  & HILL      & 0.200 & 0.192 & 0.202 &  95 & 0.0 &    110 \\
  & MiPOD     & 0.100 & 0.187 & 0.198 & 116 & 0.0 &    145 \\
  & MiPOD     & 0.200 & 0.177 & 0.195 &  99 & 0.0 &    119 \\
  & S-UNIWARD & 0.100 & 0.012 & 0.012 &  10 & 0.0 &     48 \\
  & S-UNIWARD & 0.200 & 0.033 & 0.067 & 132 & 0.0 &    150 \\
  & WOW       & 0.100 & 0.012 & 0.012 &  35 & 0.0 &     81 \\
  & WOW       & 0.200 & 0.040 & 0.074 & 150 & 0.0 &    150 \\
\bottomrule
\end{tabular}}
\end{table*}

\paragraph{Joint interpretation of detection and effectiveness metrics.}
Figures~\ref{fig:beta-vs-epoch-all}–\ref{fig:tpr-at-fpr-all} jointly characterise the interplay between detector discrimination, defensive concealment, and equilibrium dynamics. The AUC and $\mathrm{TPR@FPR}$ trajectories reveal that GNCNN achieves consistent and high discriminative power with limited epoch-wise variance, reflecting stable convergence and minimal susceptibility to local feature noise. Xu-Net, although exhibiting greater volatility, reaches sharper peaks in $\mathrm{TPR@FPR}$ and $\beta_t^{\mathcal U}$ at specific payloads, indicating transient bursts of high concealment effectiveness before stabilisation. These fluctuations correspond to short-term changes in the defender’s steganographic success probability under mixed-strategy equilibria, consistent with the conditional formulations in Section~\ref{subsec:cond-success}. Collectively, these empirical profiles validate the theoretical expectation that architectures optimising for steady discrimination (e.g., GNCNN) yield smoother advantage curves and reduced operational risk volatility, whereas more expressive models (e.g., Xu-Net) produce adaptive but higher-variance security responses.

\begin{algorithm}[ht]
\caption{Monte Carlo extension for uncertainty in $(p^\star,\beta_t^{\mathcal{U}},\bar{\beta}_t^{\mathcal{U}},I)$}
\label{alg:mc-risk}
\DontPrintSemicolon

\KwGoal{Propagate parameter uncertainty into risk summaries.}
\KwIn{Sampling model for $(p^\star,\beta_t^{\mathcal{U}},\bar{\beta}_t^{\mathcal{U}},I)$; number of draws $N$.}
\KwOut{$\{R_t^{(k)}\}_{k=1}^N$; mean $\bar R_t$; quantiles $q_\alpha$; tail risk $\mathrm{CVaR}_\alpha$.}

\For{$k \gets 1$ \KwTo $N$}{
  Sample $(p^{\star,(k)},\,\beta_t^{\mathcal{U},(k)},\,\bar{\beta}_t^{\mathcal{U},(k)},\,I^{(k)})$.\;
  $\mathsf{Adv}^{(+)}_{\mathcal{G},\mathcal{A}}(t)^{(k)} \gets p^{\star,(k)}\bigl(\bar\beta_t^{\mathcal{U},(k)}-{\beta}_t^{\mathcal{U},(k)}\bigr)$.\;
  $R_t^{(k)} \gets I^{(k)} \cdot \max\{\mathsf{Adv}^{(+)}_{\mathcal{G},\mathcal{A}}(t)^{(k)},\,0\}$.\;
}

$\bar R_t \gets \frac{1}{N}\sum_{k=1}^{N} R_t^{(k)}$.\;
$q_\alpha \gets \alpha\text{-quantile of }\{R_t^{(k)}\}$.\;
$\mathrm{CVaR}_\alpha \gets \frac{1}{|\mathcal{K}_\alpha|}\sum_{k\in\mathcal{K}_\alpha} R_t^{(k)}$, 
where $\mathcal{K}_\alpha=\{k:\, R_t^{(k)}\ge q_\alpha\}$.\;

\textbf{Remark:} If $\beta_t^{\mathcal{U}}$ is observed via a detection stream, update its posterior (e.g., Beta–Bernoulli) over time and re-run the loop at each epoch $t$ for evolving risk distributions.
\end{algorithm}

\subsection{Decision support and policy levers}\label{sec:decision-policy}

\begin{table*}[t]
\centering
\footnotesize
\setlength{\tabcolsep}{7.0pt}
\renewcommand{\arraystretch}{1.0}

\sisetup{
  round-mode=places,
  round-precision=3,
  group-separator={,},
  group-minimum-digits=2,
  table-number-alignment=center
}
\caption{Epoch-aggregated defender failure ($\beta_t$), adversarial advantage ($\mathrm{Adv}^{(+)}$), and monetary consequences (risk $R_t$, security benefit $B_{\mathrm{sec},t}$, impact = $\pounds 300,000$) across models, methods, payloads, and $p^\star$ scenarios. Currency columns show values in pounds sterling ($\pounds$).}
\label{tab:combined_adv_risk_summary_by_method}
\vspace{0.1em}
\resizebox{0.98\linewidth}{!}{%
\begin{tabular}{
l  
l  
S[table-format=1.3,table-number-alignment=center] 
S[table-format=1.2,table-number-alignment=center] 
S[table-format=1.6,table-number-alignment=center] 
S[table-format=1.6,table-number-alignment=center] 
S[table-format=1.6,table-number-alignment=center] 
S[table-format=1.6,table-number-alignment=center] 
S[table-format=1.6,table-number-alignment=center] 
S[table-format=1.6,table-number-alignment=center] 
l  
l  
l  
l  
}
\toprule
\multicolumn{1}{c}{Method} &
\multicolumn{1}{c}{Model} &
\multicolumn{1}{c}{Payload} &
\multicolumn{1}{c}{$p^\star$} &
\multicolumn{1}{c}{$\langle\beta_t\rangle$} &
\multicolumn{1}{c}{$\min\beta_t$} &
\multicolumn{1}{c}{$\max\beta_t$} &
\multicolumn{1}{c}{$\langle \mathrm{Adv}^{(+)} \rangle$} &
\multicolumn{1}{c}{$\min \mathrm{Adv}^{(+)}$} &
\multicolumn{1}{c}{$\max \mathrm{Adv}^{(+)}$} &
\multicolumn{1}{c}{$\langle R_t\rangle$} &
\multicolumn{1}{c}{$\max R_t$} &
\multicolumn{1}{c}{$\langle B_{\mathrm{sec},t}\rangle$} &
\multicolumn{1}{c}{$\max B_{\mathrm{sec},t}$} \\
\midrule
  WOW & GNCNN & 0.100 & 0.75 & 0.080907 & 0.079000 & 0.086000 & 0.651820 & 0.648000 & 0.653250 & \pounds 195,546 & \pounds 195,975 & \pounds 0 & \pounds 0 \\
   &  & 0.100 & 0.50 & 0.080907 & 0.079000 & 0.086000 & 0.434547 & 0.432000 & 0.435500 & \pounds 130,364 & \pounds 130,650 & \pounds 0 & \pounds 0 \\
   &  & 0.100 & 0.25 & 0.080907 & 0.079000 & 0.086000 & 0.217273 & 0.216000 & 0.217750 & \pounds 65,182 & \pounds 65,325 & \pounds 0 & \pounds 0 \\
   &  & 0.200 & 0.75 & 0.084433 & 0.080000 & 0.091000 & 0.649175 & 0.644250 & 0.652500 & \pounds 194,752 & \pounds 195,750 & \pounds 0 & \pounds 0 \\
   &  & 0.200 & 0.50 & 0.084433 & 0.080000 & 0.091000 & 0.432783 & 0.429500 & 0.435000 & \pounds 129,835 & \pounds 130,500 & \pounds 0 & \pounds 0 \\
   &  & 0.200 & 0.25 & 0.084433 & 0.080000 & 0.091000 & 0.216392 & 0.214750 & 0.217500 & \pounds 64,918 & \pounds 65,250 & \pounds 0 & \pounds 0 \\
   &  & 0.400 & 0.75 & 0.080360 & 0.076000 & 0.091000 & 0.652230 & 0.644250 & 0.655500 & \pounds 195,669 & \pounds 196,650 & \pounds 0 & \pounds 0 \\
   &  & 0.400 & 0.50 & 0.080360 & 0.076000 & 0.091000 & 0.434820 & 0.429500 & 0.437000 & \pounds 130,446 & \pounds 131,100 & \pounds 0 & \pounds 0 \\
   &  & 0.400 & 0.25 & 0.080360 & 0.076000 & 0.091000 & 0.217410 & 0.214750 & 0.218500 & \pounds 65,223 & \pounds 65,550 & \pounds 0 & \pounds 0 \\
\cmidrule(lr){2-14}
   & Xu-Net & 0.100 & 0.75 & 0.902037 & 0.895000 & 0.910000 & 0.035972 & 0.030000 & 0.041250 & \pounds 10,792 & \pounds 12,375 & \pounds 0 & \pounds 0 \\
   &  & 0.100 & 0.50 & 0.902037 & 0.895000 & 0.910000 & 0.023981 & 0.020000 & 0.027500 & \pounds 7,194 & \pounds 8,250 & \pounds 0 & \pounds 0 \\
   &  & 0.100 & 0.25 & 0.902037 & 0.895000 & 0.910000 & 0.011991 & 0.010000 & 0.013750 & \pounds 3,597 & \pounds 4,125 & \pounds 0 & \pounds 0 \\
   &  & 0.200 & 0.75 & 0.789500 & 0.620000 & 0.915000 & 0.120375 & 0.026250 & 0.247500 & \pounds 36,112 & \pounds 74,250 & \pounds 0 & \pounds 0 \\
   &  & 0.200 & 0.50 & 0.789500 & 0.620000 & 0.915000 & 0.080250 & 0.017500 & 0.165000 & \pounds 24,075 & \pounds 49,500 & \pounds 0 & \pounds 0 \\
   &  & 0.200 & 0.25 & 0.789500 & 0.620000 & 0.915000 & 0.040125 & 0.008750 & 0.082500 & \pounds 12,037 & \pounds 24,750 & \pounds 0 & \pounds 0 \\
   &  & 0.400 & 0.75 & 0.435167 & 0.205000 & 0.900000 & 0.386125 & 0.037500 & 0.558750 & \pounds 115,838 & \pounds 167,625 & \pounds 0 & \pounds 0 \\
   &  & 0.400 & 0.50 & 0.435167 & 0.205000 & 0.900000 & 0.257417 & 0.025000 & 0.372500 & \pounds 77,225 & \pounds 111,750 & \pounds 0 & \pounds 0 \\
   &  & 0.400 & 0.25 & 0.435167 & 0.205000 & 0.900000 & 0.128708 & 0.012500 & 0.186250 & \pounds 38,612 & \pounds 55,875 & \pounds 0 & \pounds 0 \\
\cmidrule(lr){2-14}
\addlinespace[0.1em]
\cmidrule(lr){1-14}
  S-UNIWARD & GNCNN & 0.100 & 0.75 & 0.480040 & 0.077000 & 0.807000 & 0.352470 & 0.107250 & 0.654750 & \pounds 105,741 & \pounds 196,425 & \pounds 0 & \pounds 0 \\
   &  & 0.100 & 0.50 & 0.480040 & 0.077000 & 0.807000 & 0.234980 & 0.071500 & 0.436500 & \pounds 70,494 & \pounds 130,950 & \pounds 0 & \pounds 0 \\
   &  & 0.100 & 0.25 & 0.480040 & 0.077000 & 0.807000 & 0.117490 & 0.035750 & 0.218250 & \pounds 35,247 & \pounds 65,475 & \pounds 0 & \pounds 0 \\
   &  & 0.200 & 0.75 & 0.087620 & 0.077000 & 0.810000 & 0.646785 & 0.105000 & 0.654750 & \pounds 194,036 & \pounds 196,425 & \pounds 0 & \pounds 0 \\
   &  & 0.200 & 0.50 & 0.087620 & 0.077000 & 0.810000 & 0.431190 & 0.070000 & 0.436500 & \pounds 129,357 & \pounds 130,950 & \pounds 0 & \pounds 0 \\
   &  & 0.200 & 0.25 & 0.087620 & 0.077000 & 0.810000 & 0.215595 & 0.035000 & 0.218250 & \pounds 64,678 & \pounds 65,475 & \pounds 0 & \pounds 0 \\
   &  & 0.400 & 0.75 & 0.086047 & 0.083000 & 0.113000 & 0.647965 & 0.627750 & 0.650250 & \pounds 194,390 & \pounds 195,075 & \pounds 0 & \pounds 0 \\
   &  & 0.400 & 0.50 & 0.086047 & 0.083000 & 0.113000 & 0.431977 & 0.418500 & 0.433500 & \pounds 129,593 & \pounds 130,050 & \pounds 0 & \pounds 0 \\
   &  & 0.400 & 0.25 & 0.086047 & 0.083000 & 0.113000 & 0.215988 & 0.209250 & 0.216750 & \pounds 64,796 & \pounds 65,025 & \pounds 0 & \pounds 0 \\
\cmidrule(lr){2-14}
   & Xu-Net & 0.100 & 0.75 & 0.902083 & 0.895000 & 0.910000 & 0.035937 & 0.030000 & 0.041250 & \pounds 10,781 & \pounds 12,375 & \pounds 0 & \pounds 0 \\
   &  & 0.100 & 0.50 & 0.902083 & 0.895000 & 0.910000 & 0.023958 & 0.020000 & 0.027500 & \pounds 7,187 & \pounds 8,250 & \pounds 0 & \pounds 0 \\
   &  & 0.100 & 0.25 & 0.902083 & 0.895000 & 0.910000 & 0.011979 & 0.010000 & 0.013750 & \pounds 3,594 & \pounds 4,125 & \pounds 0 & \pounds 0 \\
   &  & 0.200 & 0.75 & 0.818600 & 0.650000 & 0.905000 & 0.098550 & 0.033750 & 0.225000 & \pounds 29,565 & \pounds 67,500 & \pounds 0 & \pounds 0 \\
   &  & 0.200 & 0.50 & 0.818600 & 0.650000 & 0.905000 & 0.065700 & 0.022500 & 0.150000 & \pounds 19,710 & \pounds 45,000 & \pounds 0 & \pounds 0 \\
   &  & 0.200 & 0.25 & 0.818600 & 0.650000 & 0.905000 & 0.032850 & 0.011250 & 0.075000 & \pounds 9,855 & \pounds 22,500 & \pounds 0 & \pounds 0 \\
   &  & 0.400 & 0.75 & 0.497344 & 0.280000 & 0.900000 & 0.339492 & 0.037500 & 0.502500 & \pounds 101,848 & \pounds 150,750 & \pounds 0 & \pounds 0 \\
   &  & 0.400 & 0.50 & 0.497344 & 0.280000 & 0.900000 & 0.226328 & 0.025000 & 0.335000 & \pounds 67,898 & \pounds 100,500 & \pounds 0 & \pounds 0 \\
   &  & 0.400 & 0.25 & 0.497344 & 0.280000 & 0.900000 & 0.113164 & 0.012500 & 0.167500 & \pounds 33,949 & \pounds 50,250 & \pounds 0 & \pounds 0 \\
\cmidrule(lr){2-14}
\addlinespace[0.1em]
\cmidrule(lr){1-14}
  MiPOD & GNCNN & 0.100 & 0.75 & 0.077200 & 0.070000 & 0.089000 & 0.654600 & 0.645750 & 0.660000 & \pounds 196,380 & \pounds 198,000 & \pounds 0 & \pounds 0 \\
   &  & 0.100 & 0.50 & 0.077200 & 0.070000 & 0.089000 & 0.436400 & 0.430500 & 0.440000 & \pounds 130,920 & \pounds 132,000 & \pounds 0 & \pounds 0 \\
   &  & 0.100 & 0.25 & 0.077200 & 0.070000 & 0.089000 & 0.218200 & 0.215250 & 0.220000 & \pounds 65,460 & \pounds 66,000 & \pounds 0 & \pounds 0 \\
   &  & 0.200 & 0.75 & 0.082160 & 0.074000 & 0.091000 & 0.650880 & 0.644250 & 0.657000 & \pounds 195,264 & \pounds 197,100 & \pounds 0 & \pounds 0 \\
   &  & 0.200 & 0.50 & 0.082160 & 0.074000 & 0.091000 & 0.433920 & 0.429500 & 0.438000 & \pounds 130,176 & \pounds 131,400 & \pounds 0 & \pounds 0 \\
   &  & 0.200 & 0.25 & 0.082160 & 0.074000 & 0.091000 & 0.216960 & 0.214750 & 0.219000 & \pounds 65,088 & \pounds 65,700 & \pounds 0 & \pounds 0 \\
   &  & 0.400 & 0.75 & 0.086300 & 0.076000 & 0.600000 & 0.647775 & 0.262500 & 0.655500 & \pounds 194,332 & \pounds 196,650 & \pounds 0 & \pounds 0 \\
   &  & 0.400 & 0.50 & 0.086300 & 0.076000 & 0.600000 & 0.431850 & 0.175000 & 0.437000 & \pounds 129,555 & \pounds 131,100 & \pounds 0 & \pounds 0 \\
   &  & 0.400 & 0.25 & 0.086300 & 0.076000 & 0.600000 & 0.215925 & 0.087500 & 0.218500 & \pounds 64,778 & \pounds 65,550 & \pounds 0 & \pounds 0 \\
\cmidrule(lr){2-14}
   & Xu-Net & 0.100 & 0.75 & 0.203414 & 0.070000 & 0.895000 & 0.559940 & 0.041250 & 0.660000 & \pounds 167,982 & \pounds 198,000 & \pounds 0 & \pounds 0 \\
   &  & 0.100 & 0.50 & 0.203414 & 0.070000 & 0.895000 & 0.373293 & 0.027500 & 0.440000 & \pounds 111,988 & \pounds 132,000 & \pounds 0 & \pounds 0 \\
   &  & 0.100 & 0.25 & 0.203414 & 0.070000 & 0.895000 & 0.186647 & 0.013750 & 0.220000 & \pounds 55,994 & \pounds 66,000 & \pounds 0 & \pounds 0 \\
   &  & 0.200 & 0.75 & 0.241597 & 0.085000 & 0.900000 & 0.531303 & 0.037500 & 0.648750 & \pounds 159,391 & \pounds 194,625 & \pounds 0 & \pounds 0 \\
   &  & 0.200 & 0.50 & 0.241597 & 0.085000 & 0.900000 & 0.354202 & 0.025000 & 0.432500 & \pounds 106,261 & \pounds 129,750 & \pounds 0 & \pounds 0 \\
   &  & 0.200 & 0.25 & 0.241597 & 0.085000 & 0.900000 & 0.177101 & 0.012500 & 0.216250 & \pounds 53,130 & \pounds 64,875 & \pounds 0 & \pounds 0 \\
   &  & 0.400 & 0.75 & 0.161064 & 0.035000 & 0.880000 & 0.591702 & 0.052500 & 0.686250 & \pounds 177,511 & \pounds 205,875 & \pounds 0 & \pounds 0 \\
   &  & 0.400 & 0.50 & 0.161064 & 0.035000 & 0.880000 & 0.394468 & 0.035000 & 0.457500 & \pounds 118,340 & \pounds 137,250 & \pounds 0 & \pounds 0 \\
   &  & 0.400 & 0.25 & 0.161064 & 0.035000 & 0.880000 & 0.197234 & 0.017500 & 0.228750 & \pounds 59,170 & \pounds 68,625 & \pounds 0 & \pounds 0 \\
\cmidrule(lr){2-14}
\addlinespace[0.1em]
\cmidrule(lr){1-14}
  HILL & GNCNN & 0.100 & 0.75 & 0.079760 & 0.063000 & 0.200000 & 0.652680 & 0.562500 & 0.665250 & \pounds 195,804 & \pounds 199,575 & \pounds 0 & \pounds 0 \\
   &  & 0.100 & 0.50 & 0.079760 & 0.063000 & 0.200000 & 0.435120 & 0.375000 & 0.443500 & \pounds 130,536 & \pounds 133,050 & \pounds 0 & \pounds 0 \\
   &  & 0.100 & 0.25 & 0.079760 & 0.063000 & 0.200000 & 0.217560 & 0.187500 & 0.221750 & \pounds 65,268 & \pounds 66,525 & \pounds 0 & \pounds 0 \\
   &  & 0.200 & 0.75 & 0.581887 & 0.044000 & 1.000000 & 0.276085 & -0.037500 & 0.679500 & \pounds 82,900 & \pounds 203,850 & \pounds 75 & \pounds 11,250 \\
   &  & 0.200 & 0.50 & 0.581887 & 0.044000 & 1.000000 & 0.184057 & -0.025000 & 0.453000 & \pounds 55,267 & \pounds 135,900 & \pounds 50 & \pounds 7,500 \\
   &  & 0.200 & 0.25 & 0.581887 & 0.044000 & 1.000000 & 0.092028 & -0.012500 & 0.226500 & \pounds 27,634 & \pounds 67,950 & \pounds 25 & \pounds 3,750 \\
   &  & 0.400 & 0.75 & 0.247053 & 0.039000 & 0.942000 & 0.527210 & 0.006000 & 0.683250 & \pounds 158,163 & \pounds 204,975 & \pounds 0 & \pounds 0 \\
   &  & 0.400 & 0.50 & 0.247053 & 0.039000 & 0.942000 & 0.351473 & 0.004000 & 0.455500 & \pounds 105,442 & \pounds 136,650 & \pounds 0 & \pounds 0 \\
   &  & 0.400 & 0.25 & 0.247053 & 0.039000 & 0.942000 & 0.175737 & 0.002000 & 0.227750 & \pounds 52,721 & \pounds 68,325 & \pounds 0 & \pounds 0 \\
\cmidrule(lr){2-14}
   & Xu-Net & 0.100 & 0.75 & 0.448523 & 0.225000 & 0.885000 & 0.376108 & 0.048750 & 0.543750 & \pounds 112,832 & \pounds 163,125 & \pounds 0 & \pounds 0 \\
   &  & 0.100 & 0.50 & 0.448523 & 0.225000 & 0.885000 & 0.250739 & 0.032500 & 0.362500 & \pounds 75,222 & \pounds 108,750 & \pounds 0 & \pounds 0 \\
   &  & 0.100 & 0.25 & 0.448523 & 0.225000 & 0.885000 & 0.125369 & 0.016250 & 0.181250 & \pounds 37,611 & \pounds 54,375 & \pounds 0 & \pounds 0 \\
   &  & 0.200 & 0.75 & 0.181818 & 0.050000 & 0.825000 & 0.576136 & 0.093750 & 0.675000 & \pounds 172,841 & \pounds 202,500 & \pounds 0 & \pounds 0 \\
   &  & 0.200 & 0.50 & 0.181818 & 0.050000 & 0.825000 & 0.384091 & 0.062500 & 0.450000 & \pounds 115,227 & \pounds 135,000 & \pounds 0 & \pounds 0 \\
   &  & 0.200 & 0.25 & 0.181818 & 0.050000 & 0.825000 & 0.192045 & 0.031250 & 0.225000 & \pounds 57,614 & \pounds 67,500 & \pounds 0 & \pounds 0 \\
   &  & 0.400 & 0.75 & 0.080921 & 0.015000 & 0.845000 & 0.651809 & 0.078750 & 0.701250 & \pounds 195,543 & \pounds 210,375 & \pounds 0 & \pounds 0 \\
   &  & 0.400 & 0.50 & 0.080921 & 0.015000 & 0.845000 & 0.434539 & 0.052500 & 0.467500 & \pounds 130,362 & \pounds 140,250 & \pounds 0 & \pounds 0 \\
   &  & 0.400 & 0.25 & 0.080921 & 0.015000 & 0.845000 & 0.217270 & 0.026250 & 0.233750 & \pounds 65,181 & \pounds 70,125 & \pounds 0 & \pounds 0 \\
\bottomrule
\end{tabular}}
\end{table*}
\noindent
Table~\ref{tab:combined_adv_risk_summary_by_method} summarises the aggregated defender failure probabilities $\beta_t^{\mathcal U}$, the adversarial advantage
$\mathrm{Adv}^{(+)}_{\mathcal G,\mathcal A}(t)=p^{*}(\bar{\beta}^{\mathcal U}-\beta_t^{\mathcal U})$, and the derived monetary effects $R_t=I\!\cdot\!\max\{\mathrm{Adv}^{(+)},0\}$ and $B_{\mathrm{sec},t}=I\!\cdot\!\max\{-\mathrm{Adv}^{(+)},0\}$. A low $\beta_{\text{mean}}$ reflects high steganographic effectiveness by the defender, while a large $\mathrm{Adv}^{(+)}_{\text{mean}}$ indicates growing adversarial success under stronger detectors or higher equilibrium attack probabilities~$p^{*}$. As expected, $\mathrm{Adv}^{(+)}_{\text{mean}}$ and $R_{\text{mean}}$ scale approximately linearly with~$p^{*}$, confirming that risk intensifies with more aggressive adversarial search. 

Consistent with the calibrated instantiation in \S\ref{scenario setup}, we set the defender’s leak impact to $C^{\mathcal U}_{\text{leak}}=\pounds 300{,}000$. Under the modelling assumption $B^{\mathcal A}_{\text{leak}}=\rho\,C^{\mathcal U}_{\text{leak}}$ this implies $B^{\mathcal A}_{\text{leak}}\approx 0.83\times\pounds 300{,}000=\pounds 249{,}000\approx\pounds 250{,}000$, which is the value used in the comparative analyses.

For example, the \textsc{WOW} and \textsc{S-UNIWARD} results at $p^{*}{=}0.75$ and $I{=}\pounds300{,}000$ exhibit the largest $R_{\max}$ values ($\approx\!\pounds1.9\times10^{5}$), whereas \textsc{MiPOD} and \textsc{HILL} maintain lower $\beta_{\text{mean}}$ and hence smaller $\mathrm{Adv}^{(+)}$ gaps. These observations substantiate the theoretical claim that diminishing defender effectiveness ($\beta_t^{\mathcal U}\!\downarrow$) suppresses $\mathrm{Adv}^{(+)}$, while high $p^{*}$ regimes expose the system to elevated economic loss. In operational terms, the results show that organisations must monitor not only the average detector failure but also the worst‐case risk envelope across attacker mixes and payloads. It also suggests that embedding in regimes with high attacker search probability may require additional controls or policy levers before steganography is beneficial. Overall, the results quantify how equilibrium adversarial intensity and payload magnitude jointly determine expected loss in operational currency units, bridging empirical steganalytic performance with game-theoretic cost modelling.

The results further underscore the central security goal of confidentiality: each leak event represents a breach of the protected communication, and the model’s outcomes map directly to the monetary consequences of that breach. By quantifying expected losses $(R_t)$ in currency units, the analysis links strategic embedding decisions, defender detection performance, and adversary search intensity to the actual exposure of confidential information. In turn, the security benefit metric $(B_{\mathrm{sec},t})$ captures the value of preserved confidentiality when the adversary’s advantage is negative. The results emphasise that embedding decisions should be guided by realistic attacker behaviour, detector effectiveness over time, and payload size, rather than static assumptions, if confidentiality is to be maintained cost-effectively.

\subsection{Model Discrimination and Steganographic Effectiveness}

\begin{table}[htbp]
  \centering
  \caption{Embedding statistics grouped by method and payload (distortion and embedding cost).}
  \label{tab:embedding-stats}
  \footnotesize
  \setlength{\tabcolsep}{4pt}%
  \renewcommand{\arraystretch}{1.12}%
  \sisetup{
    round-mode=places,
    round-precision=3,
    table-number-alignment=center
  }

  \begin{adjustbox}{max width=\linewidth}
  \begin{tabular}{
    l
    S[table-format=1.3]
    S[table-format=2.2]
    S[table-format=1.4]
    S[table-format=1.3]
    S[table-format=3.1]
  }
    \toprule
    \multicolumn{1}{c}{Method} &
    \multicolumn{1}{c}{Payload} &
    \multicolumn{1}{c}{Mean PSNR (dB)} &
    \multicolumn{1}{c}{Mean MSE} &
    \multicolumn{1}{c}{Median $t_{\mathrm{embed}}$ (s)} &
    \multicolumn{1}{c}{Success (\%)} \\
    \midrule

    \multirow{3}{*}{WOW}
      & 0.100 & 65.03 & 0.0204 & 0.226 & 100.0 \\
      & 0.200 & 61.68 & 0.0442 & 0.229 & 100.0 \\
      & 0.400 & 58.29 & 0.0965 & 0.235 & 100.0 \\
    \addlinespace[0.35em]
    \cmidrule(lr){1-6}

    \multirow{3}{*}{S-UNIWARD}
      & 0.100 & 67.31 & 0.0136 & 0.028 & 100.0 \\
      & 0.200 & 67.19 & 0.0144 & 0.029 & 100.0 \\
      & 0.400 & 67.19 & 0.0144 & 0.028 & 100.0 \\
    \addlinespace[0.35em]
    \cmidrule(lr){1-6}

    \multirow{3}{*}{MiPOD}
      & 0.100 & 56.14 & 0.2012 & 0.088 & 100.0 \\
      & 0.200 & 55.91 & 0.2030 & 0.180 & 100.0 \\
      & 0.400 & 55.41 & 0.2114 & 0.087 & 100.0 \\
    \addlinespace[0.35em]
    \cmidrule(lr){1-6}

    \multirow{3}{*}{HILL}
      & 0.100 & 58.14 & 0.0998 & 0.013 & 100.0 \\
      & 0.200 & 55.15 & 0.1989 & 0.013 & 100.0 \\
      & 0.400 & 52.13 & 0.3986 & 0.013 & 100.0 \\

    \bottomrule
  \end{tabular}
  \end{adjustbox}
\end{table}

Table~\ref{tab:embedding-stats} benchmarks distortion and embedding cost across payloads and four spatial-domain schemes. As expected, higher payloads degrade fidelity monotonically (PSNR$\downarrow$, MSE$\uparrow$). \textsc{S-UNIWARD} consistently yields the best image quality (e.g., $67.3$\,dB at $0.1$\,bpp with MSE $1.36\times10^{-2}$), followed by \textsc{WOW} and \textsc{HILL}; \textsc{MiPOD} trades slightly lower PSNR for stable behaviour across payloads. Embedding runtimes are near-constant for \textsc{HILL} ($\approx\!13$\,ms) and \textsc{S-UNIWARD} ($\approx\!28$–$29$\,ms), while \textsc{WOW} is comparatively heavier (0.226–0.235\,s) and \textsc{MiPOD} exhibits higher variability (0.087–0.180\,s), reflecting its per-pixel optimisation. All methods achieved a 100\% embedding success rate.

\begin{figure*}
    \centering
    \includegraphics[width=0.9\linewidth]{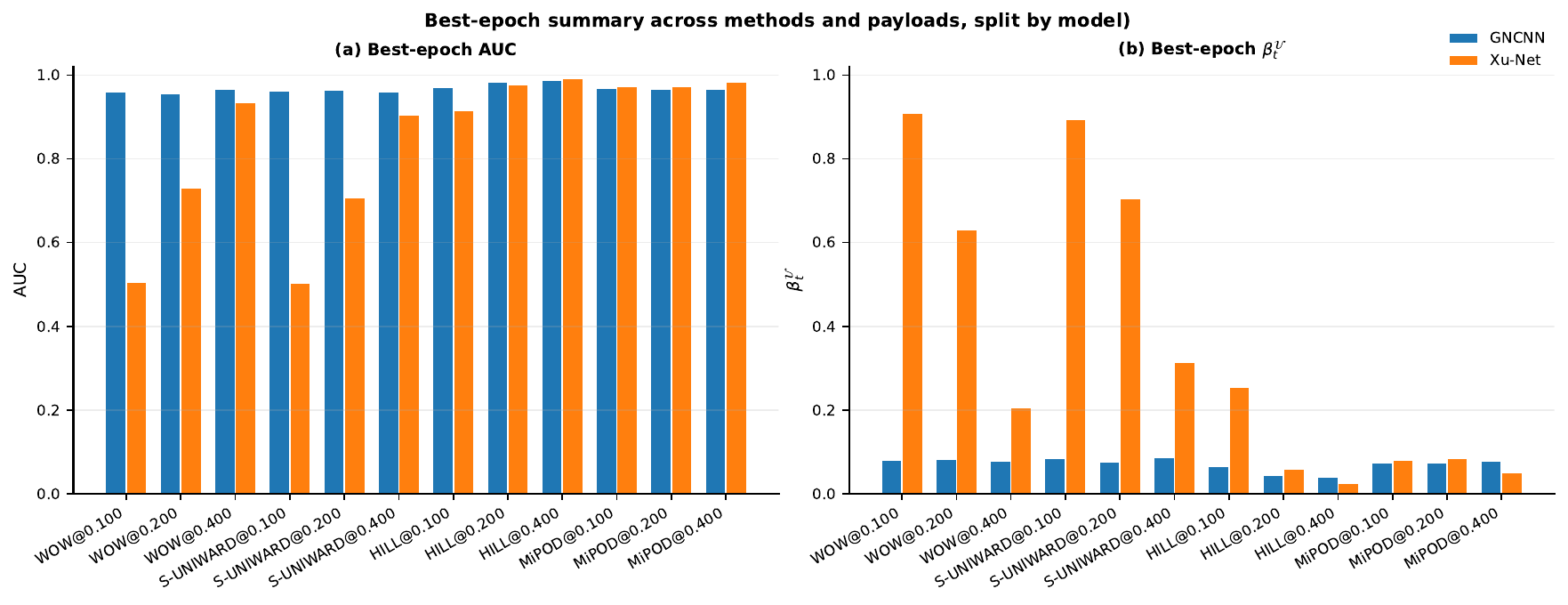}
    \caption{Comparison of (a) discriminative power (AUC) and (b) steganographic effectiveness gap proxy \(\beta_t^{\mathcal U}\) at the best epoch for each method–payload configuration and model (GNCNN, Xu-Net). Panel~(a) summarises detector operating characteristics; Panel~(b) summarises the corresponding defender failure probability at the selected best epoch. Together, they contextualise the trade-off between \emph{recognition quality} and \emph{residual failure} that drives the adversarial advantage metric in Sec.~\ref{sec:risk_interpretation}.}
  \label{fig:best-epoch-auc-beta-combined}
\end{figure*}

Subsequently, figure~\ref{fig:best-epoch-auc-beta-combined} consolidates the best-epoch discrimination (AUC) and defender effectiveness $\beta_t^{\mathcal U}$ for Xu-Net and GNCNN across payloads and embedding methods. 
GNCNN achieves uniformly high AUC values ($\ge0.95$) with minimal variance, indicating strong and consistent detection of stego artefacts. 
Xu-Net, while occasionally surpassing GNCNN at low payloads, shows higher volatility—especially for \textsc{wow} and \textsc{s-uniward}—reflecting sensitivity to embedding perturbations and texture diversity. 

The $\beta_t^{\mathcal U}$ profiles reveal complementary trends: Xu-Net attains higher peak concealment at small payloads but fluctuates sharply, whereas GNCNN maintains a stable, moderate $\beta_t^{\mathcal U}$ across payloads. 
This stability translates into smoother operational advantage trajectories and fewer abrupt risk spikes in the mixed-strategy dynamics (cf.\ Section~\ref{subsec:cond-success}). 
Overall, Xu-Net exhibits burst-like defensive gains under favourable conditions, while GNCNN delivers steadier, architecture-driven robustness. For \textit{Detector Operator points} and  \textit{Training and compute efficiency summary with system metrics } see the Tables \ref{tab:efficiency} and \ref{tab:operating-points} in Appendix \ref{Det_Training}.

\section{Comparison to Adaptive Steganography Schemes}
\label{sec:comparison-adaptive}

\begin{table*}[htbp]
  \centering
  \scriptsize
  \caption{Comparison of classical content–adaptive steganography to our curvature–based game model with extensions for dynamic risk governance. 
  Each distortion “principle” captures the original scheme’s embedding logic; 
  “Our framework view” interprets the induced curvature and how it evolves over time through equilibrium mixes $(p^\star,q^\star)$ under varying detection effectiveness and impact valuation.}
  \label{tab:comparison-adaptive}
  \setlength{\tabcolsep}{6pt}%
  \renewcommand{\arraystretch}{0.8}%
  \begin{tabularx}{\linewidth}{l l X X X}
    \toprule
    \textbf{Scheme} & \textbf{Domain} & \textbf{Distortion principle (original goal)} & \textbf{Strengths / Typical use} & \textbf{Our framework view (effective curvature $\rightarrow$ effect on $p^\star,q^\star$)} \\
    \midrule
    WOW~\cite{Holub2012WOW} 
      & Spatial 
      & Wavelet-obtained weights penalise modifications in smooth structures; concentrate changes in textured/noisy areas to reduce detectability. 
      & Strong baseline across BOSSBase/BOWS2; simple, fast cost; reference method in many studies. 
      & Heavier penalties for “fragile” pixels $\Rightarrow$ higher effective leak-loss convexity $\gamma_\ell$, which \emph{depresses} attacker mix $p^\star$ in our model (Figure~\ref{fig:F3_sensitivity}). \\
    \addlinespace[3pt]
    S-UNIWARD~\cite{holub2014universal}
      & Spatial (wavelet-domain cost)
      & Universal relative distortion in wavelet subbands; content adaptivity via multi-scale responses; objective is low detectability at target payload.
      & Robust, generalises across cover sources; widely used as “hard” baseline in steganalysis papers.
      & More conservative cost in predictable regions $\Rightarrow$ increases effective $\gamma_\ell$; pushes equilibrium toward lower $p^\star$ (attacker searches less often). \\
    \addlinespace[3pt]
    HILL~\cite{li2014new}
      & Spatial
      & High-pass/low-pass residual guided cost; favours embedding in edges/textures, suppresses smooth areas.
      & Competitive performance with modest complexity; often a comparator for CNN steganalysis.
      & Residual-driven penalties behave like steeper loss curvature on “smooth” pixels; net effect similar: higher effective $\gamma_\ell \Rightarrow$ smaller $p^\star$. \\
    \addlinespace[3pt]
    MiPOD~\cite{Sedighi2016MiPOD}
      & Spatial
      & Probabilistic embedding that \emph{minimises detector power} (near-optimal/locally optimal tests); change probabilities tied to statistical detectability.
      & Strong asymptotic rationale; remains a high-bar baseline vs. modern CNN steganalysis.
      & MiPOD focuses on embedding in areas with lower detector sensitivity. In our curvature-based game model, this is represented by a higher effective leakage cost after utility transformation, shifting the mixed equilibrium toward lower attacker search incentives ($p^\star$). Implementation and tuning overheads may also increase the defender's hiding-cost curvature ($\gamma_h$), reflecting diminishing returns as system complexity grows.\\
    \addlinespace[3pt]
    Schöttle–Böhme~\cite{Schottle2013,Schottle2016a}
      & Theory (content-adaptive)
      & Equilibria under detector priors/independence assumptions; mixes driven by cover–detector statistics (game of signals).
      & Clear equilibrium logic; explains when/why mixed strategies arise under statistical models.
      & Our model turns the \emph{driver} from divergences to \emph{economically curved utilities}: concave benefits and convex losses/overheads. 
      Increased $\gamma_\ell$ (breach loss convexity) $\downarrow p^\star$; increased $\gamma_h$ (ops convexity) $\uparrow p^\star$; increased $\gamma_a$ (search convexity) $\downarrow q^\star$. \\
    \addlinespace[3pt]
    \textbf{This work} 
      & Dynamic Game–Risk Framework
      & Introduces detection dynamics via time–varying defender effectiveness $\beta^{\mathcal U}_t$ with baseline $\bar\beta^{\mathcal U}$; defines per–epoch adversarial advantage (See Equation \ref{eq:AdvPlus-def}); integrates impact as currency-valued risk (See equation \ref{eq:Risk-nonneg}); enables simulation-based distributions of outcomes rather than static snapshots.
      & Extends beyond static distortion maps—models dynamic detection, impact translation, and temporal adaptation. Enables Monte Carlo risk evaluation and supports decision–ready governance metrics.
      & Time-varying $\beta^{\mathcal U}_t$ alters curvature flow over epochs; higher variance in $\beta^{\mathcal U}_t$ produces cyclical $p^\star,q^\star$ adjustments. Monetary scaling ($I$) and stochastic updates propagate through equilibrium to yield evolving distributions of expected loss, risk premium, and adversary advantage over time. \\
    \bottomrule
  \end{tabularx}
  \vspace{0.25em}
  \footnotesize\emph{Notes.} 
  The curvature-based formulation subsumes classical content–adaptive cost maps into an \emph{economic game of incentives}. Extensions introduced in this work incorporate (i) dynamic detection processes through $\beta^{\mathcal U}_t$, (ii) monetary translation of adversarial advantage $R_t = I\cdot \max\{\mathsf{Adv}^{(+)}(t),0\}$, and (iii) Monte Carlo simulation of evolving equilibria. Together, these upgrades transform static embedding equilibria into dynamic, time-aware distributions of risk, loss, and adversarial success—providing a governance-oriented bridge between steganographic adaptivity and strategic risk management.
\end{table*}

Classical content–adaptive steganography calibrates embedding changes to local image complexity via a distortion functional, trading payload for detectability. Canonical spatial families include \emph{WOW}, \emph{S-UNIWARD}, \emph{HILL}, and \emph{MiPOD} \cite{Holub2012WOW,holub2014universal,fridrich2011steganalysis,Sedighi2016MiPOD,li2014new}. While their cost maps differ, they share the same design principle: concentrate modifications in textured/noisy regions to minimise steganalytic evidence. In our economic/game–of–incentives view (See Sections \S \ref{A_Steganographic_Game-Theoretic_Model} and \ref{Game_Analysis}), these design choices manifest as \emph{effective curvatures} of benefits and losses: heavier penalties for changing “fragile” pixels translate into higher leak–loss convexity on the defender side, which depresses the attacker’s equilibrium look mix $p^\star$, whereas costlier defensive pipelines (e.g., aggressive pre/post–processing) steepen the operational–overhead convexity and push $p^\star$ upward. Hence, the curvature sweeps reported in Figure~\ref{fig:F3_sensitivity} provide a unifying lens across these content–adaptive families by isolating how curvature governs strategic responses.

Against this backdrop, game–theoretic analyses by Schöttle and Böhme form the closest theoretical comparators  \cite{Schottle2016a, Schottle2013, SchottlePascalandLaszkaAronandJohnsonBenjaminandGrossklagsJensandBohme2013}. Their models develop equilibrium reasoning for content–adaptive embedding under detector priors and (in parts) independent embedding constraints, yielding mixed strategies driven by \emph{statistical} properties of covers and steganalysers. Our framework shares the equilibrium focus but shifts the driver from statistical divergences to \emph{economically calibrated, explicitly nonlinear} utilities: concave protection benefits and convex breach/overhead losses anchor best responses in observables (See \ref{tab:combined_adv_risk_summary_by_method}) and feed directly into risk. Where Schöttle–Böhme models are fundamentally games of signals, our formulation is a game of incentives that preserves the content–adaptive intuition while making curvature an explicit policy knob.

This difference enables three extensions that matter for risk governance. First, we make detection dynamics an input through a time–varying defender effectiveness $\beta^{\mathcal U}_t$ with a baseline $\bar\beta^{\mathcal U}$, and define a per–epoch adversarial advantage, which is expressed in Equation \ref{eq:AdvPlus-def}, quantifying when the attacker’s search incentives temporarily outpace the defender's effective detection capability. Second, we integrate impact and replace the likelihood with adversary advantage (See Equation \ref{eq:Risk-nonneg}), improving on the classical risk framework and turning equilibrium mixes into decision–ready risk measures. Third, we structure the pipeline for time–series simulation and Monte Carlo scenario analysis (See Algorithm \ref{alg:mc-risk} for Monte Carlo extension for uncertainty in $(p^\star,\beta_t^{\mathcal{U}},\bar{\beta}_t^{\mathcal{U}},I)$), so that payload shifts, detector learning, and calibrated mixing translate into distributions of success rates and losses rather than a static snapshot.

Finally, we note that newer “robust/generative’’ paradigms (e.g., diffusion/flow–based generative steganography) pursue semantic control and channel robustness rather than classical pixel–wise distortion minimisation. Even there, the central trade–offs remain payload, fidelity, and detectability \cite{qi2024provably, yu2023cross}; in our terms, increased robustness effectively raises the transformed leak cost, which, all else equal, lowers $p^\star$ and shifts risk in the defender’s favour. Thus, the curvature–based interpretation extends naturally from traditional content–adaptive schemes to modern generative pipelines, providing a consistent bridge from equilibrium behaviour to quantitative risk. Table \ref{tab:comparison-adaptive} summarises the comparison of classical content–adaptive steganography to our model.

\section{Conclusion}
\label{Conclusion}
This work bridges the domains of strategic steganography and quantitative risk management by developing a novel game-theoretic framework that links mixed-strategy equilibrium behaviour to monetised risk outcomes. The classical risk template is refined into a decision-conditioned form, where the adversarial advantage metric captures how the attacker’s incentives can temporarily outstrip a defender’s time-varying detection effectiveness. By introducing this new metric and translating it into currency-unit risk, it enables decision-makers to move from theoretical equilibrium mixes to actionable assessments of exposure and defence investment.

The model presented in this paper integrates empirical calibration—from breach cost statistics and regulatory fine bounds to expert-elicited hardening costs—and embeds it into a Monte Carlo simulation pipeline that captures uncertainty, detector learning, payload shifts and temporal dynamics. The result is a distributional view of success probabilities and expected losses rather than a single static snapshot. In addition, this work shows how different dimensions—such as operational “harmony” benefits, defender search/overhead cost curvature, and adversary search cost—interact strategically to delineate when steganographic embedding is beneficial, marginal or counter-productive.

Finally, analysis identifies transparent mappings between policy levers (such as raising attacker search cost or reducing adversarial benefit via data minimisation) and game-model parameters, enabling comparative statics and governance prescriptions tailored to realistic organisational settings. This fusion of strategic embedding modelling and quantitative risk management opens new pathways for steganography research and organisational defence planning. Future work can extend our model to broader operational domains beyond spatial-domain image embedding.

\appendices
\section{Parameter Elicitation for Nonlinear Utilities}
\label{app:param-elicitation}

This appendix gives a reproducible recipe to set two key curvature parameters in the nonlinear utility maps of Section~\ref{A_Steganographic_Game-Theoretic_Model}:
(i) the superlinearity exponent \(\gamma_\ell>1\) in \(\tilde C^U_{\text{leak}}=\alpha_\ell\,(C^U_{\text{leak}})^{\gamma_\ell}\),
anchored to GDPR fine exposure; and
(ii) the benefit scale \(\beta_U\) in \(\tilde B^U_{\text{hide}}=\beta_U\log(1+B^U_{\text{hide}}/b_U)\),
anchored to a Gordon–Loeb–style budget fraction for optimal security investment.

\subsection{Eliciting \texorpdfstring{\(\gamma_\ell\)}{gamma\_ℓ} from GDPR Exposure Bands}
The EDPB Guidelines on calculating administrative fines (final, 24~May~2023) prescribe a five-step methodology that sets a \emph{turnover-aware starting amount} using infringement gravity and enterprise size, then adjusts for aggravating/mitigating factors before checking the Article~83 caps (up to the higher of fixed euro amounts or a turnover percentage).%
\footnote{See EDPB Guidelines~04/2022 (final), esp. Chapter~2; Article~83(4)--(6) GDPR caps: up to EUR~10/20~million or 2\%/4\% of worldwide turnover.}
Let \(F(s,T)\) denote the \emph{starting amount} as a function of seriousness band \(s\) (e.g., low/medium/high/very high) and annual global turnover \(T\).
To encode the empirically convex exposure of large controllers and of serious infringements, choose \(\gamma_\ell>1\) so that the transformed loss \(\tilde C^U_{\text{leak}}\) \emph{matches} representative fine levels across two or more calibration points.

\paragraph{Inputs.}
Pick \(k\ge2\) scenarios \(\{(C_i, F_i)\}_{i=1}^k\), where
\(C_i\) is a baseline monetary loss proxy for scenario \(i\) (e.g., direct breach losses plus internal response costs),
and \(F_i\) is a GDPR \emph{starting amount} consistent with the EDPB table for seriousness band \(s_i\) at turnover \(T_i\).
(If available, incorporate expected civil damages and sector penalties into \(F_i\).)

\paragraph{Two-point closed form.}
For \(k=2\), solve
\[
\alpha_\ell\,C_1^{\gamma_\ell}=F_1,\qquad \alpha_\ell\,C_2^{\gamma_\ell}=F_2
\]
to obtain
\[
\gamma_\ell \;=\; \frac{\ln(F_2/F_1)}{\ln(C_2/C_1)} \;>\; 1,\qquad
\alpha_\ell \;=\; \frac{F_1}{C_1^{\gamma_\ell}} \;=\; \frac{F_2}{C_2^{\gamma_\ell}}.
\]
Use \(C_2>C_1\) and \(F_2>F_1\) reflecting a higher seriousness band and/or larger turnover.

\paragraph{Multi-point fit.}
For \(k>2\), estimate \((\alpha_\ell,\gamma_\ell)\) by log–log regression:
\[
\min_{\alpha_\ell>0,\,\gamma_\ell>1}\;
\sum_{i=1}^k \Bigl(\ln F_i - \ln \alpha_\ell - \gamma_\ell \ln C_i\Bigr)^2,
\]
optionally weighting each calibration point by regulator-assessed infringement gravity or by the estimated number of affected data subjects.

\paragraph{Practical bands.}
As a starting rule of thumb when detailed \(F_i\) are not available:
\[
\gamma_\ell \in 
\begin{cases}
[1.10,1.30] & \text{low/medium gravity, SME turnover},\\
[1.30,1.70] & \text{high gravity or mid/large turnover},\\
[1.70,2.20] & \text{very high gravity and large/global turnover}.
\end{cases}
\]
Once actual EDPB starting amounts for your sector/size are inserted, recompute via the two-point or multi-point fit to replace these priors.

\subsection{Calibrating \texorpdfstring{\(\beta_U\)}{beta\_U} from a Gordon–Loeb Budget Fraction}
Gordon and Loeb show that, for a given information set with expected loss \(E[L]\), the optimal security investment typically does not exceed \(1/e\approx0.368\) of \(E[L]\).%
\footnote{See \emph{The Economics of Information Security Investment}, ACM TISSEC 2002.}
We translate this into a target optimal spend \(x^\star = f\,E[L]\) with \(f\in(0,1/e]\), then choose \(\beta_U\) so that \(x^\star\) solves the first-order condition for the defender’s utility of spend.

\paragraph{Spend–benefit mapping.}
Let security spend be \(x\ge0\).
Model the (monetary) benefit from spending as a risk-reduction response \(B^U_{\text{hide}}(x)=\kappa\,x\) for small-to-moderate \(x\)
(\(\kappa>0\) converts spend into avoided loss units),%
\footnote{Replace \(\kappa x\) with any differentiable increasing response \(R(x)\); the formula below adapts by substituting \(R'(x)\).}
and keep the utility transform \(\tilde B^U_{\text{hide}}(x)=\beta_U\log\bigl(1+\kappa x\bigr)\).
For the cost side, take \(\tilde C^U_{\text{hide}}(x)=\eta_1 x^{\gamma_h}\) with \(\eta_1>0\), \(\gamma_h\ge1\).

\paragraph{FOC at the target optimum.}
Maximize \(U(x)=\beta_U\log(1+\kappa x)-\eta_1 x^{\gamma_h}\).
The first-order condition \(\beta_U \kappa/(1+\kappa x) - \eta_1\gamma_h x^{\gamma_h-1}=0\) evaluated at \(x^\star=f\,E[L]\) gives a direct calibration:
\[
\boxed{\;
\beta_U \;=\; \frac{(1+\kappa x^\star)\,\eta_1\,\gamma_h\,\bigl(x^\star\bigr)^{\gamma_h-1}}{\kappa}
\;},\quad x^\star=f\,E[L],\;\; f\le 1/e.
\]
Choose \(\kappa\) from an empirical control–impact curve, e.g., observed 
or set \(\kappa=1/b_U\) if you normalize the log scale by \(b_U=E[L]\) (dimensionless argument).
Then compute \(\beta_U\) to make \(x^\star\) the optimiser.

\paragraph{Checklist.}
(i) Estimate \(E[L]\) from recent incidents/sector studies; set \(f\in[0.2,\,0.37]\) to encode budget policy.
(ii) Choose \(\gamma_h\) from engineering cost curvature (e.g., linear \(\gamma_h{=}1\) if dominated by license/CapEx; convex \(\gamma_h{>}1\) if dominated by analyst-hours/latency).
(iii) Set \(\eta_1\) to match current operating costs at status quo spend \(x_0\).
(iv) Pick \(b_U=E[L]\) so \(\log(1+\kappa x)\) acts on a natural loss scale.
(v) Compute \(\beta_U\) from the boxed formula; verify that the induced \(x^\star\) matches the Gordon–Loeb fraction and that marginal utility at \(x^\star\) balances marginal cost.

\section{Absolute Risk Variant}\label{app:risk-variants}

For completeness, we record an \emph{absolute} formulation of the adversarial risk that can be useful when the objective is to rank scenarios strictly by the \emph{magnitude} of exposure, irrespective of whether steganography improves or degrades outcomes relative to baseline. Expanding on the adversarial advantage defined in equation \ref{eq:AdvPlus-def}, the \emph{absolute adversarial advantage} takes the magnitude of this differential:
\begin{equation}
\label{eq:abs-adv-def}
\mathsf{Adv}^{\mathsf{attack}}_{\mathcal{G},\mathcal{A}}(t)
\;:=\;
\bigl|\mathsf{Adv}^{(+)}_{\mathcal{G},\mathcal{A}}(t)\bigr|
\;=\; p^\star \,\bigl|\bar\beta_t^{\mathcal U}-\beta^{\mathcal U}\bigr|.
\end{equation}
If \(I_t\ge 0\) denotes the (possibly time–varying) impact random variable measured in currency units, the \emph{absolute} risk at epoch \(t\) is defined as
\begin{equation}
\label{eq:R-abs}
R^{\mathrm{abs}}_t
\;:=\;
\mathbb{E}[I_t]\cdot \mathsf{Adv}^{\mathsf{attack}}_{\mathcal{G},\mathcal{A}}(t)
\;=\;
\mathbb{E}[I_t]\cdot p^\star\,\bigl|\bar\beta_t^{\mathcal U}-\beta^{\mathcal U}\bigr|.
\end{equation}
This construction obeys standard risk–management doctrine in which risk is evaluated as a combination (or function) of \emph{likelihood} and \emph{consequence/impact}: here the likelihood proxy is \(\mathsf{Adv}^{\mathsf{attack}}_{\mathcal{G},\mathcal{A}}(t)\in[0,1]\) and the consequence is \(\mathbb{E}[I_t]\ge 0\).\cite{NIST80030,ENISA_RM_Standards}

\paragraph{Basic properties.}
From \eqref{eq:abs-adv-def}–\eqref{eq:R-abs} we immediately have: (i) \textbf{Nonnegativity and bounds:} \(0\le \mathsf{Adv}^{\mathsf{attack}}_{\mathcal{G},\mathcal{A}}(t)\le p^\star\le 1\) since \(|\bar\beta_t^{\mathcal U}-\beta^{\mathcal U}|\le 1\). Consequently \(0\le R^{\mathrm{abs}}_t\le \mathbb{E}[I_t]\). (ii) \textbf{Monotonicity in aggressiveness:} for fixed \(\beta_t^{\mathcal U},\bar\beta^{\mathcal U}\), the advantage is linear and increasing in \(p^\star\). (iii) \textbf{Symmetry in effectiveness gap:} \(\mathsf{Adv}^{\mathsf{attack}}_{\mathcal{G},\mathcal{A}}(t)\) depends only on the magnitude \(|\bar\beta_t^{\mathcal U}-\beta^{\mathcal U}|\), so it is indifferent to whether hiding is \emph{better} or \emph{worse} than baseline; it measures the size of the gap, not its sign. (iv) \textbf{Zero conditions:} \(R^{\mathrm{abs}}_t=0\) if either \(p^\star=0\) (adversary abstains), or \(\beta_t^{\mathcal U}=\bar\beta^{\mathcal U}\) (no incremental effect from hiding), or \(\mathbb{E}[I_t]=0\) (no loss).

\paragraph{Relation to the signed formulation.}
Let the signed risk/benefit pair from the main text be
\(
R_t=\mathbb{E}[I_t]\cdot \max\{\mathsf{Adv}^{(+)}_{\mathcal{G},\mathcal{A}}(t),0\}
\)
and
\(
B_{\text{sec},t}=\mathbb{E}[I_t]\cdot \max\{-\mathsf{Adv}^{(+)}_{\mathcal{G},\mathcal{A}}(t),0\}
\),
which keep \emph{risk} nonnegative and separately account for any \emph{security benefit} when hiding outperforms baseline (i.e., \(\mathsf{Adv}^{(+)}_{\mathcal{G},\mathcal{A}}(t)<0\)). Then
\begin{multline}
\label{eq:decomp}
 R^{\mathrm{abs}}_t
\;=\; \mathbb{E}[I_t]\cdot |\mathsf{Adv}^{(+)}_{\mathcal{G},\mathcal{A}}(t)|
\;=\; R_t + B_{\text{sec},t}, 
\\ 
R_t \;=\; \mathbb{E}[I_t]\cdot \max\{\mathsf{Adv}^{(+)}_{\mathcal{G},\mathcal{A}}(t),0\}
\;\le\; R^{\mathrm{abs}}_t.
\end{multline}
Thus, the absolute variant aggregates the two signed quantities into a single nonnegative magnitude; it is informative for \emph{size} comparisons but discards the directional information (whether hiding helps or hurts).

\paragraph{Aggregation across time or scenarios.}
Over a horizon \(t=1,\dots,T\) with discount factor \(\delta\in(0,1]\), one may summarize absolute risk by
\(
\mathsf{ARisk}=\sum_{t=1}^T \delta^{t-1} R^{\mathrm{abs}}_t
\).
If \((I_t,\bar\beta_t^{\mathcal U},\beta^{\mathcal U})\) are modeled as random (e.g., via Monte Carlo or Bayesian posteriors), unbiased estimation follows by averaging \(\hat R^{\mathrm{abs}}_t = \hat{\mathbb{E}}[I_t]\cdot p^\star\,\bigl|\hat{\bar\beta}_t^{\mathcal U}-\hat{\beta}_t^{\mathcal U}\bigr|\) over draws. When independence between \(I_t\) and the effectiveness gap is assumed at this layer, \(\mathbb{E}[R^{\mathrm{abs}}_t]=\mathbb{E}[I_t]\cdot p^\star\,\mathbb{E}\bigl[|\bar\beta_t^{\mathcal U}-\beta^{\mathcal U}|\bigr]\).

\paragraph{When to use the absolute variant.}
The absolute measure \eqref{eq:R-abs} is appropriate when stakeholders require a single nonnegative risk score to rank/control portfolios or to comply with reporting regimes that mandate magnitudes only. However, because it suppresses the sign of \(\mathsf{Adv}^{(+)}_{\mathcal{G},\mathcal{A}}(t)\), it cannot differentiate \emph{security benefit} from \emph{residual risk}. For strategic decision-making (e.g., control selection), we therefore recommend the signed formulation in the main text and report the pair \((R_t,B_{\text{sec},t})\); the absolute variant should then be interpreted as their sum, cf.\ \eqref{eq:decomp}. This keeps “risk” nonnegative, aligns with standards that define risk as a combination of likelihood and impact, and avoids the conceptual pitfall of reporting negative risk values. \cite{NIST80030,ENISA_RM_Standards}

\paragraph{Normalised risk--benefit illustration}
\begin{figure*}[t]
  \centering
  \includegraphics[width=0.75\linewidth]{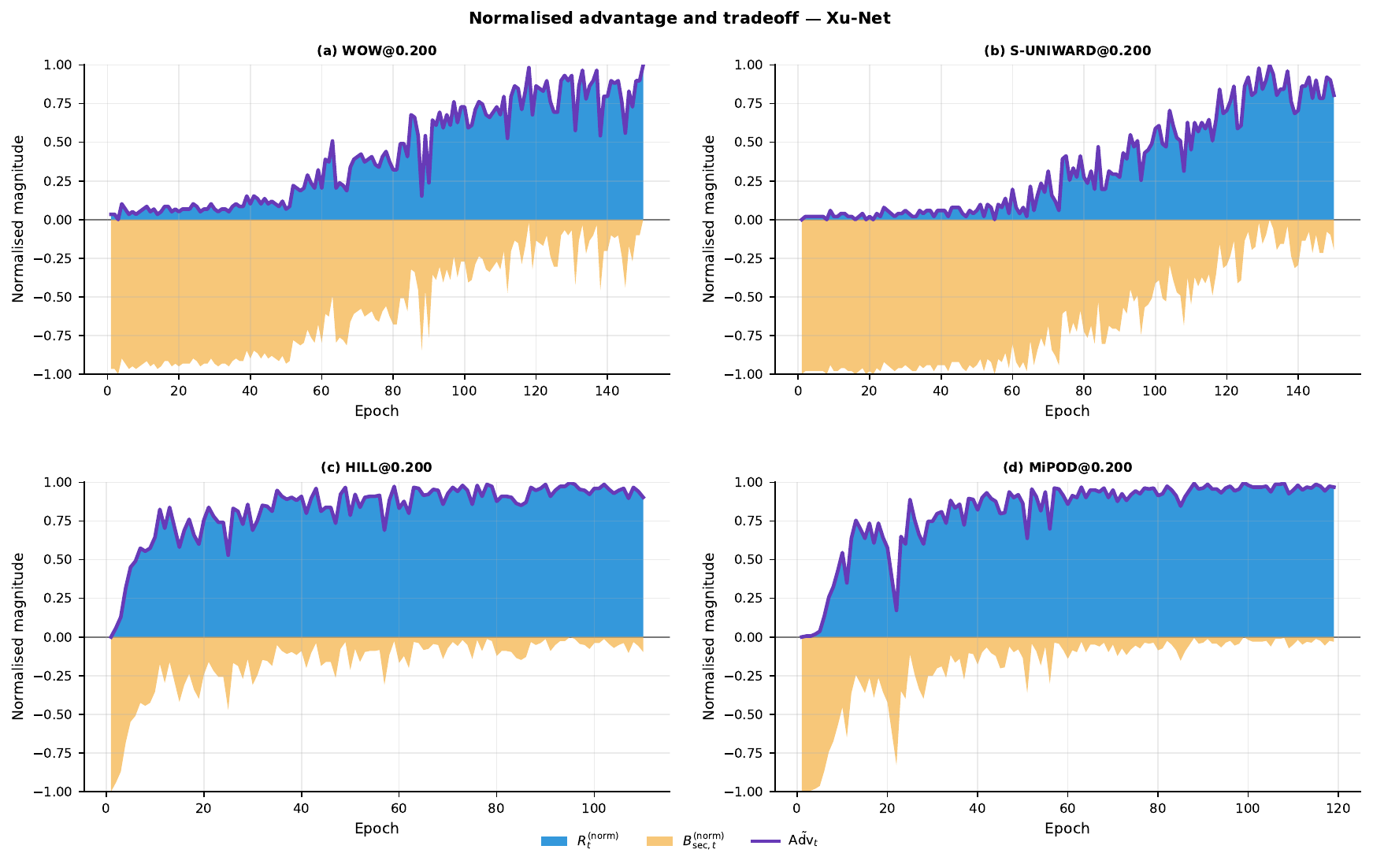}
  \caption{\small Normalised advantage and risk--benefit tradeoff over training epochs for Xu-Net at payload $0.200$ across WOW, S-UNIWARD, HILL, and MiPOD.}
  \label{fig:appendix-norm-xunet}
\end{figure*}

The main text defines the impact series $R_t = I\cdot\max\{\mathsf{Adv}^{(+)}_t,0\}$ and $B_{\mathrm{sec},t} = I\cdot\max\{-\mathsf{Adv}^{(+)}_t,0\}$ to separate attacker-favourable and defender-favourable regimes relative to a fixed baseline $\bar\beta^{\mathcal U}$.
For intuition, Fig.\ref{fig:appendix-norm-xunet} in this appendix shows an alternative normalised
parameterisation based on $\tilde{\mathsf{Adv}}_t \in [0,1]$, with
$R^{(\mathrm{norm})}_t = I\cdot\,\tilde{\mathsf{Adv}}_t$ and
$B^{(\mathrm{norm})}_t = I\cdot,(1 - \tilde{\mathsf{Adv}}_t)$, which
produces a smooth tradeoff where risk increases as residual security
benefit decreases. These normalised curves are used only for illustration and are not part
of the operational risk metric reported in the main evaluation.

\section{Quantitative Evaluation Setup}
\label{Quantitative Evaluation setup}

\paragraph{Objective.}
To evaluate robustness under uncertainty, we perform a Monte-Carlo ensemble of \(n{=}10{,}000\) equilibrium computations for the game \(\mathcal G\) with nonlinear utilities (Section~\ref{sec:assumptions}). Each iteration samples monetary primitives and behavioural parameters, applies the transformations \eqref{eq:def-concave-benefits}–\eqref{eq:def-adversary}, and solves for mixed-strategy probabilities \((p^\star,q^\star)\) via \eqref{equation_p}–\eqref{equation_q}. The resulting distributions capture the time-varying interaction between detection capability, curvature of utilities, and economic impact within our dynamic risk-governance framework \cite{ross2022simulation}.

\paragraph{Monetary primitives.}
All quantities are denominated in pounds sterling and drawn from uniform bands consistent with the scenario in Section~\ref{Motivation_for_Security_Game Model}:
\[
\begin{aligned}
B^{U}_{\text{harmony}} &\sim \mathsf{Unif}(120{,}000,350{,}000), &
C^{U}_{\text{hide}} &\sim \mathsf{Unif}(30{,}000,120{,}000),\\
B^{U}_{\text{hide}} &\sim \mathsf{Unif}(80{,}000,250{,}000), &
C^{U}_{\text{leak}} &\sim \mathsf{Unif}(150{,}000,900{,}000),\\
B^{A}_{\text{leak}} &\sim \mathsf{Unif}(120{,}000,600{,}000), &
C^{A}_{\text{look}} &\sim \mathsf{Unif}(5{,}000,60{,}000).
\end{aligned}
\]
These ranges satisfy the ordering assumptions (A2)–(A3): \(C^U_{\text{hide}}<C^U_{\text{leak}}\) and \(C^A_{\text{look}}<B^A_{\text{leak}}\).

\paragraph{Behavioural parameters.}
Curvature and scale parameters are drawn as:
\begin{multline*}
\gamma_\ell \sim \mathsf{Unif}(1.2,2.0), 
\gamma_h \sim \mathsf{Unif}(1.0,1.8),
\gamma_a \sim \mathsf{Unif}(1.0,1.8), \\
\alpha_\ell \sim \mathsf{Unif}(5{\times}10^{-7},2{\times}10^{-6}), 
\alpha_a \sim \mathsf{Unif}(5{\times}10^{-6},3{\times}10^{-5}), \\
\eta_1 \sim \mathsf{Unif}(1{\times}10^{-4},5{\times}10^{-4}),
\end{multline*}

with fixed \(b_U{=}h_U{=}b_A{=}100{,}000\) and \(\beta_U{=}\delta_U{=}\beta_A{=}1\), \(\eta_0{=}0\).  
These encode concave benefits and convex losses (A1) while keeping transformed magnitudes comparable.

\paragraph{Utility transformation and equilibrium.}
Each draw produces transformed payoffs:
\[
\begin{aligned}
\tilde B^U_{\text{hide}} &= \beta_U\log(1{+}B^U_{\text{hide}}/b_U), &
\tilde B^U_{\text{harmony}} &= \delta_U\log(1{+}B^U_{\text{harmony}}/h_U),\\
\tilde C^U_{\text{leak}} &= \alpha_{\ell}(C^U_{\text{leak}})^{\gamma_\ell}, &
\tilde C^U_{\text{hide}} &= \eta_1 (C^U_{\text{hide}})^{\gamma_h},\\
\tilde B^A_{\text{leak}} &= \beta_A\log(1{+}B^A_{\text{leak}}/b_A), &
\tilde C^A_{\text{look}} &= \alpha_a (C^A_{\text{look}})^{\gamma_a}.
\end{aligned}
\]
Defining \(\Delta=\tilde B^U_{\text{harmony}}-\tilde B^U_{\text{hide}}+\tilde C^U_{\text{hide}}\),
\[
p^\star=\frac{\Delta}{\tilde B^U_{\text{harmony}}+\tilde C^U_{\text{leak}}},\qquad
q^\star=1-\frac{\tilde C^A_{\text{look}}}{\tilde B^A_{\text{leak}}},
\]
with denominators lower-bounded at \(10^{-12}\) and results clipped to \([0,1]\) for numerical stability (Theorem~\ref{theorem_p}).  
Draws where \(p^\star,q^\star\in(0,1)\) are classified as interior equilibria; others collapse to pure strategies.

\paragraph{Sampling design and reproducibility.}
Independent uniform sampling maintains transparent comparative statics; variance-reduced designs (e.g., Latin Hypercube Sampling) can reduce noise at the same computational budget \cite{stein1987lhs, ross2022simulation}.  
All draws use NumPy’s \texttt{Generator} with seed 42; figure generation is deterministic under the same environment, ensuring reproducibility.

\paragraph{Assumption checks.}
Each draw logs (A2)–(A4) compliance. Occasional violations of monetary orderings are discarded; affine transformations (A4) have no equilibrium effect.  
The ensemble distributions in Figure~\ref{fig:F5_distribution} quantify how parameter uncertainty and detection variability translate into evolving adversarial advantage and currency-valued risk $R_t$ within the simulated horizon.

\section{Detector Operating Point / Training and compute efficiency}\label{Det_Training}
The detector operating point are presented in Tables \ref{tab:efficiency} and \ref{tab:operating-points}.

\begin{table*}[t]
\centering
\footnotesize
\setlength{\tabcolsep}{27pt}
\renewcommand{\arraystretch}{1.0}
\sisetup{
  table-number-alignment = center,
  round-mode = places,
  round-precision = 3
}

\caption{Detector operating points at the selected best epoch (max AUC or min $\beta^{\mathcal U}_t$). \emph{Condensed for space: highest AUC rows retained.}}
\label{tab:operating-points}
\vspace{0.01em}

\resizebox{0.98\linewidth}{!}{%
\begin{tabular}{
  l l
  S[table-format=1.3]
  S[table-format=1.2]
  S[table-format=1.3]
  S[table-format=1.3]
  S[table-format=1.3]
  S[table-format=1.3]
}
\toprule
\multicolumn{1}{c}{\textbf{Model}} &
\multicolumn{1}{c}{\textbf{Method}} &
\multicolumn{1}{c}{\textbf{Payload}} &
\multicolumn{1}{c}{\textbf{FPR} $(\alpha)$} &
\multicolumn{1}{c}{\textbf{TPR@FPR}} &
\multicolumn{1}{c}{$\boldsymbol{\beta^{\mathcal U}_t}$} &
\multicolumn{1}{c}{\textbf{AUC}} &
\multicolumn{1}{c}{\textbf{thr@FPR}} \\
\midrule

\multirow{12}{*}{\textbf{GNCNN}}
  & \multirow{3}{*}{HILL}
    & 0.100 & 0.10 & 0.935 & 0.065 & 0.970 & 0.177 \\
  & & 0.200 & 0.10 & 0.956 & 0.044 & 0.984 & 0.274 \\
  & & 0.400 & 0.10 & 0.961 & 0.039 & 0.987 & 0.298 \\
\cmidrule(lr){2-8}
  & \multirow{3}{*}{MiPOD}
    & 0.100 & 0.10 & 0.926 & 0.074 & 0.969 & 0.243 \\
  & & 0.200 & 0.10 & 0.926 & 0.074 & 0.966 & 0.201 \\
  & & 0.400 & 0.10 & 0.921 & 0.079 & 0.968 & 0.228 \\
\cmidrule(lr){2-8}
  & \multirow{3}{*}{S-UNIWARD}
    & 0.100 & 0.10 & 0.916 & 0.084 & 0.962 & 0.266 \\
  & & 0.200 & 0.10 & 0.923 & 0.077 & 0.965 & 0.268 \\
  & & 0.400 & 0.10 & 0.913 & 0.087 & 0.960 & 0.284 \\
\cmidrule(lr){2-8}
  & \multirow{3}{*}{WOW}
    & 0.100 & 0.10 & 0.919 & 0.081 & 0.961 & 0.326 \\
  & & 0.200 & 0.10 & 0.917 & 0.083 & 0.957 & 0.310 \\
  & & 0.400 & 0.10 & 0.921 & 0.079 & 0.967 & 0.178 \\

\cmidrule(lr){1-8}

\multirow{12}{*}{\textbf{Xu-Net}}
  & \multirow{3}{*}{HILL}
    & 0.100 & 0.10 & 0.745 & 0.255 & 0.915 & 0.574 \\
  & & 0.200 & 0.10 & 0.940 & 0.060 & 0.977 & 0.693 \\
  & & 0.400 & 0.10 & 0.975 & 0.025 & 0.991 & 0.318 \\
\cmidrule(lr){2-8}
  & \multirow{3}{*}{MiPOD}
    & 0.100 & 0.10 & 0.920 & 0.080 & 0.974 & 0.581 \\
  & & 0.200 & 0.10 & 0.915 & 0.085 & 0.972 & 0.523 \\
  & & 0.400 & 0.10 & 0.950 & 0.050 & 0.984 & 0.510 \\
\cmidrule(lr){2-8}
  & \multirow{3}{*}{S-UNIWARD}
    & 0.100 & 0.10 & 0.105 & 0.895 & 0.502 & 0.515 \\
  & & 0.200 & 0.10 & 0.295 & 0.705 & 0.708 & 0.638 \\
  & & 0.400 & 0.10 & 0.685 & 0.315 & 0.906 & 0.633 \\
\cmidrule(lr){2-8}
  & \multirow{3}{*}{WOW}
    & 0.100 & 0.10 & 0.090 & 0.910 & 0.505 & 0.501 \\
  & & 0.200 & 0.10 & 0.370 & 0.630 & 0.731 & 0.595 \\
  & & 0.400 & 0.10 & 0.795 & 0.205 & 0.934 & 0.477 \\

\bottomrule
\end{tabular}%
}
\end{table*}

\begin{table*}[t]
\centering
\footnotesize
\setlength{\tabcolsep}{17pt}
\renewcommand{\arraystretch}{1.0}

\sisetup{
  table-number-alignment = center,
  round-mode = places,
  round-precision = 3,
  group-separator = {,},
  group-minimum-digits = 3
}

\caption{Training and compute efficiency summary with system metrics (mean values across epochs).}
\label{tab:efficiency}
\vspace{0.25em}

\begin{threeparttable}
\resizebox{0.98\linewidth}{!}{%
\begin{tabular}{
  l l
  S[table-format=1.3]
  S[table-format=1.3]
  S[table-format=2.3, table-text-alignment=center] 
  S[table-format=1.3]
  S[table-format=2.3]
  S[table-format=1.3]
}
\toprule
\multicolumn{1}{c}{\textbf{Model}} &
\multicolumn{1}{c}{\textbf{Method}} &
\multicolumn{1}{c}{\textbf{Payload}} &
\multicolumn{1}{c}{\textbf{Mean AUC}} &
\multicolumn{1}{c}{\textbf{Mean $t_{\mathrm{epoch}}$ (s)}} &
\multicolumn{1}{c}{\textbf{Mean CPU\%}} &
\multicolumn{1}{c}{\textbf{Mean RAM\%}} &
\multicolumn{1}{c}{\textbf{Mean GPU mem (GB)}} \\
\midrule

\multirow{12}{*}{\textbf{GNCNN}}
  & \multirow{3}{*}{HILL}
    & 0.100 & 0.961184 & 46.231300 & 0.798333 & 48.357667 & 7.331805 \\
  & & 0.200 & 0.500026 & 47.538589 & 1.332667 & 42.712667 & 7.028487 \\
  & & 0.400 & 0.870760 & 47.302302 & 1.406667 & 42.965333 & 7.031497 \\
\cmidrule(lr){2-8}

  & \multirow{3}{*}{MiPOD}
    & 0.100 & 0.963751 & 47.516723 & 2.769333 & 49.897000 & 7.255461 \\
  & & 0.200 & 0.959412 & 47.178661 & 2.164000 & 44.348333 & 7.372381 \\
  & & 0.400 & 0.961260 & 45.717345 & 0.774333 & 48.695667 & 7.078159 \\
\cmidrule(lr){2-8}

  & \multirow{3}{*}{S-UNIWARD}
    & 0.100 & 0.829641 & 46.804040 & 1.980333 & 46.639000 & 7.339918 \\
  & & 0.200 & 0.952133 & 45.788732 & 0.398667 & 40.666333 & 7.160612 \\
  & & 0.400 & 0.956153 & 45.587356 & 0.908333 & 42.826000 & 7.184938 \\
\cmidrule(lr){2-8}

  & \multirow{3}{*}{WOW}
    & 0.100 & 0.959093 & 46.662203 & 3.531000 & 43.715333 & 7.427169 \\
  & & 0.200 & 0.954774 & 46.260254 & 0.774333 & 41.155667 & 7.097771 \\
  & & 0.400 & 0.963602 & 47.033806 & 1.140000 & 41.385333 & 7.280163 \\

\cmidrule(lr){1-8}

\multirow{12}{*}{\textbf{Xu-Net}}
  & \multirow{3}{*}{HILL}
    & 0.100 & 0.829071 & NaN & 1.772727 & 39.217045 & 0.054000 \\
  & & 0.200 & 0.936848 & NaN & 5.044545 & 43.669091 & 0.050000 \\
  & & 0.400 & 0.971926 & NaN & 3.171930 & 49.324561 & 0.035000 \\
\cmidrule(lr){2-8}

  & \multirow{3}{*}{MiPOD}
    & 0.100 & 0.927580 & NaN & 0.427586 & 42.630345 & 0.054000 \\
  & & 0.200 & 0.915864 & NaN & 0.449580 & 42.565546 & 0.049000 \\
  & & 0.400 & 0.944602 & NaN & 0.358156 & 42.548936 & 0.035000 \\
\cmidrule(lr){2-8}

  & \multirow{3}{*}{S-UNIWARD}
    & 0.100 & 0.501030 & NaN & 5.887500 & 50.452083 & 0.054000 \\
  & & 0.200 & 0.581380 & NaN & 2.642000 & 47.860667 & 0.049000 \\
  & & 0.400 & 0.806093 & NaN & 1.949219 & 45.697656 & 0.035000 \\
\cmidrule(lr){2-8}

  & \multirow{3}{*}{WOW}
    & 0.100 & 0.502146 & NaN & 0.440741 & 41.829630 & 0.063000 \\
  & & 0.200 & 0.609838 & NaN & 0.440000 & 41.729333 & 0.049000 \\
  & & 0.400 & 0.836124 & NaN & 0.332667 & 41.642000 & 0.035000 \\

\bottomrule
\end{tabular}%
}
\begin{tablenotes}[flushleft]
\footnotesize
\item \textit{Note.} Xu\textendash Net epoch times ($t_{\mathrm{epoch}}$) were not recorded by the logging backend and are marked “NaN”. Cross-model wall-clock comparisons are excluded \\ where telemetry is incomplete. All other metrics represent epoch-wise means across training.
\end{tablenotes}
\end{threeparttable}
\end{table*}

Table~\ref{tab:efficiency} summarises detector performance and system telemetry. With GNCNN, AUCs are uniformly high at $0.1$\,bpp and $0.4$\,bpp (0.83–0.96), indicating clear discriminability at low and moderate payloads; the $0.2$\,bpp dip for \textsc{HILL} (AUC $\approx\!0.50$) is a notable outlier and suggests payload/method interactions that merit separate analysis or cross-checking. Mean $t_{\mathrm{epoch}}$ for GNCNN is stable at $\sim\!46$–$48$\,s across methods/payloads on the reported hardware, with modest CPU and RAM utilisation and $\sim\!7$\,GB GPU memory use. Xu--Net shows strong AUC at high payloads (e.g., \textsc{HILL} $0.4$\,bpp: $0.972$), but epoch times were not captured by the logging backend and are therefore marked $NaN$. To avoid bias, we refrain from cross-model wall-clock comparisons where telemetry is incomplete (see the table note). The remaining Xu--Net telemetry (CPU/RAM/GPU mem) was recorded and is reported for completeness.

\bibliographystyle{IEEEtran} 
\bibliography{bare_arXiv}

@article{xie2023novel,
  title={A novel gradient-guided post-processing method for adaptive image steganography},
  author={Xie, Guoliang and Ren, Jinchang and Marshall, Stephen and Zhao, Huimin and Li, Rui},
  journal={Signal Processing},
  volume={203},
  pages={108813},
  year={2023},
  publisher={Elsevier}
}

@article{ntivuguruzwa2023convolutional,
  title={A convolutional neural network to detect possible hidden data in spatial domain images},
  author={Ntivuguruzwa, Jean De La Croix and Ahmad, Tohari},
  journal={Cybersecurity},
  volume={6},
  number={1},
  pages={23},
  year={2023},
  publisher={Springer}
}

@article{boroumand2018deep,
  title={Deep residual network for steganalysis of digital images},
  author={Boroumand, Mehdi and Chen, Mo and Fridrich, Jessica},
  journal={IEEE Transactions on Information Forensics and Security},
  volume={14},
  number={5},
  pages={1181--1193},
  year={2018},
  publisher={IEEE}
}

@article{kingma2014adam,
  title={Adam: A method for stochastic optimization},
  author={Kingma, Diederik P},
  journal={arXiv preprint arXiv:1412.6980},
  year={2014}
}

@article{kaur2022systematic,
  title={A Systematic Review of Computational Image Steganography Approaches.},
  author={Kaur, Sharanpreet and Singh, Surender and Kaur, Manjit and Lee, Heung-No},
  journal={Archives of Computational Methods in Engineering},
  volume={29},
  number={7},
  year={2022}
}

@article{subramanian2021image,
  title={Image steganography: A review of the recent advances},
  author={Subramanian, Nandhini and Elharrouss, Omar and Al-Maadeed, Somaya and Bouridane, Ahmed},
  journal={IEEE access},
  volume={9},
  pages={23409--23423},
  year={2021},
  publisher={IEEE}
}

@article{qi2024provably,
  title={Provably Secure Robust Image Steganography via Cross-Modal Error Correction},
  author={Qi, Yuang and Chen, Kejiang and Zhao, Na and Yang, Zijin and Zhang, Weiming},
  journal={arXiv preprint arXiv:2412.12206},
  year={2024}
}

@article{yu2023cross,
  title={Cross: Diffusion model makes controllable, robust and secure image steganography},
  author={Yu, Jiwen and Zhang, Xuanyu and Xu, Youmin and Zhang, Jian},
  journal={Advances in Neural Information Processing Systems},
  volume={36},
  pages={80730--80743},
  year={2023}
}

@techreport{ENISA_RM_Standards,
  title        = {Implementation Guidance: On Commission Implementing Regulation (EU) 2024/2690 of 17.10.2024 laying down rules for the application of Directive (EU) 2022/2555 as regards technical and methodological requirements of cybersecurity risk-management measures (Draft for Public Consultation)},
  author  = {Konstantinos Moulinos, Marianthi Theocharidou, ENISA},
  institution  = {ENISA},
  year         = {2025},
  url          = {https://www.enisa.europa.eu/sites/default/files/2024-11/Implementation%20guidance%20on%20security%20measures_FOR%20PUBLIC%20CONSULTATION.pdf}
}

@misc{gdpr_article83_legislation,
author       = {{European Parliament and Council of the European Union}},
title        = {{Regulation (EU) 2016/679 (General Data Protection Regulation), Article 83: General conditions for imposing administrative fines}},
year         = {2016},
howpublished = {\emph{legislation.gov.uk}},
url          = {https://www.legislation.gov.uk/eur/2016/679/article/83},
note         = {Sets maximum administrative fines, including an upper tier of up to EUR~20 million or 4% of total worldwide annual turnover}
}

@article{computerweekly_2025_uk_costs,
  title   = {AI-enabled security pushes down breach costs for UK organisations},
  author  = {Scroxton, Alex},
  journal = {Computer Weekly},
  year    = {2025},
  month   = {July},
  note    = {UK-specific coverage of IBM’s Cost of a Data Breach findings; also references the global average trend}
}

@book{Ablon2014HackersBazaar,
  author    = {Ablon, Lillian and Libicki, Martin C. and Andrea M. Abler.},
  title     = {Markets for Cybercrime Tools and Stolen Data: Hackers' Bazaar},
  publisher = {RAND Corporation},
  address   = {Santa Monica, CA},
  year      = {2014},
  doi       = {10.7249/RR610},
  howpublished = {\url{https://www.rand.org/pubs/research_reports/RR610.html}}
}

@misc{IBM2025CDBR,
  title        = {Cost of a Data Breach Report 2025},
  author       = {{IBM Security}},
  year         = {2025},
  howpublished = {\url{https://www.ibm.com/reports/data-breach}},
  note         = {Accessed 2025-10-26}
}

@techreport{Verizon2024,
  title        = {2024 Data Breach Investigations Report},
  author       = {{Verizon Business}},
  year         = {2024},
  institution  = {Verizon},
  url          = {https://www.verizon.com/business/resources/reports/2024-dbir-data-breach-investigations-report.pdf}
}

@article{ColsonCooke2018,
  author  = {A. R. Colson and R. M. Cooke},
  title   = {Expert Elicitation: Using the Classical Model to Validate Experts' Judgments},
  journal = {Review of Environmental Economics and Policy},
  year    = {2018},
  volume  = {12},
  number  = {1},
  pages   = {113--132},
  doi     = {10.1093/reep/rex022}
}

@article{Kardes2011Robust,
  title={Discounted robust stochastic games and an application to queueing control},
  author={Karde{\c{s}}, Erim and Ord{\'o}{\~n}ez, Fernando and Hall, Randolph W},
  journal={Operations research},
  volume={59},
  number={2},
  pages={365--382},
  year={2011},
  publisher={INFORMS}
}

@article{McKelveyPalfrey1995,
  author    = {McKelvey, Richard D and Palfrey, Thomas R},
  title     = {Quantal Response Equilibria for Normal Form Games},
  journal   = {Games and economic behavior},
  year      = {1995},
  volume    = {10},
  number    = {1},
  pages     = {6--38},
  publisher = {Elsevier}
}

@book{Camerer2011BGT,
  author    = {Colin F. Camerer},
  title     = {Behavioral Game Theory: Experiments in Strategic Interaction},
  publisher = {Princeton University Press},
  year      = {2011}
}

@article{Wang2022ChannelErrors,
  title   = {Improving robust adaptive steganography via minimizing channel errors},
  author  = {Wang, Yaofei and Ding, Shangyi and Liu, Yao and Liu, Yang and Huang, Jiwu},
  journal = {Signal Processing},
  volume  = {195},
  pages   = {108498},
  year    = {2022},
  doi     = {10.1016/j.sigpro.2022.108498}
}

@article{AISM2024,
  title={AISM: An Adaptable Image Steganography Model with user customization},
  author={Guo, Bobiao and Ping, Ping and Zhou, Siqi and Bloh, Olano Teah and Xu, Feng and Zhou, Xiaofeng},
  journal={IEEE Transactions on Services Computing},
  volume={17},
  number={5},
  pages={1955--1968},
  year={2024},
  publisher={IEEE}
}

@article{ImmuneCover2024,
  title={Constructing immune-cover for improving holistic security of spatial adaptive steganography},
  author={Chen, Yijing and Wang, Hongxia and Li, Wanjie},
  journal={IEEE Transactions on Dependable and Secure Computing},
  volume={21},
  number={6},
  pages={5403--5419},
  year={2024},
  publisher={IEEE}
}

@article{FuzzyStego2025,
  title={FuzzyStego: An Adaptive Steganographic Scheme Using Fuzzy Logic for Optimizing Embeddable Areas in Spatial Domain Images.},
  author={Ananti, Mardhatillah Shevy and D'Layla, Adifa Widyadhani Chanda and Croix, Ntivuguruzwa Jean De La and Ahmad, Tohari},
  journal={Computers, Materials \& Continua},
  volume={84},
  number={1},
  year={2025}
}

@article{MRAS2025,
  title={MRAS: A Matching Robust Adaptive Steganography Scheme for JPEG Images over Social Networking Platforms},
  author={Kumar, Rakesh and Bansal, Savina and Bansal, RK},
  journal={Signal Processing},
  pages={110172},
  year={2025},
  publisher={Elsevier}
}

@inproceedings{cogranne2020alaska,
  title={ALASKA\# 2: Challenging academic research on steganalysis with realistic images},
  author={Cogranne, R{\'e}mi and Giboulot, {\'E}va and Bas, Patrick},
  booktitle={2020 IEEE International Workshop on Information Forensics and Security (WIFS)},
  pages={1--5},
  year={2020},
  organization={IEEE}
}

@article{luo2024comprehensive,
  title={A Comprehensive Survey of Digital Image Steganography and Steganalysis},
  author={Luo, Weiqi and Wei, Kangkang and Li, Qiushi and Ye, Miaoxin and Tan, Shunquan and Tang, Weixuan and Huang, Jiwu and others},
  journal={APSIPA Transactions on Signal and Information Processing},
  volume={13},
  number={1},
  year={2024},
  publisher={Now Publishers, Inc.}
}

@inproceedings{li2014new,
  title={A new cost function for spatial image steganography},
  author={Li, Bin and Wang, Ming and Huang, Jiwu and Li, Xiaolong},
  booktitle={2014 IEEE International conference on image processing (ICIP)},
  pages={4206--4210},
  year={2014},
  organization={IEEE}
}

@inproceedings{qian2015deep,
  title={Deep learning for steganalysis via convolutional neural networks},
  author={Qian, Yinlong and Dong, Jing and Wang, Wei and Tan, Tieniu},
  booktitle={Media Watermarking, Security, and Forensics 2015},
  volume={9409},
  pages={171--180},
  year={2015},
  organization={SPIE}
}

@article{shi2020cnn,
  title={CNN-based steganalysis and parametric adversarial embedding: a game-theoretic framework},
  author={Shi, Xiaoyu and Tondi, Benedetta and Li, Bin and Barni, Mauro},
  journal={Signal Processing: Image Communication},
  volume={89},
  pages={115992},
  year={2020},
  publisher={Elsevier}
}

@article{xu2016structural,
  title={Structural design of convolutional neural networks for steganalysis},
  author={Xu, Guanshuo and Wu, Han-Zhou and Shi, Yun-Qing},
  journal={IEEE Signal Processing Letters},
  volume={23},
  number={5},
  pages={708--712},
  year={2016},
  publisher={IEEE}
}

@article{Kodovsky2012Ensemble,
  author    = {Jan Kodovsk{\'y} and Jessica Fridrich and Vojt{\v e}ch Holub},
  title     = {Ensemble Classifiers for Steganalysis of Digital Media},
  journal   = {{IEEE} Transactions on Information Forensics and Security},
  year      = {2012},
  volume    = {7},
  number    = {2},
  pages     = {432--444},
  doi       = {10.1109/TIFS.2011.2175919}
}

@article{fridrich2012rich,
  title={Rich models for steganalysis of digital images},
  author={Fridrich, Jessica and Kodovsky, Jan},
  journal={IEEE Transactions on information Forensics and Security},
  volume={7},
  number={3},
  pages={868--882},
  year={2012},
  publisher={IEEE}
}

@article{Sedighi2016MiPOD,
  title={Content-adaptive steganography by minimizing statistical detectability},
  author={Sedighi, Vahid and Cogranne, R{\'e}mi and Fridrich, Jessica},
  journal={IEEE transactions on information forensics and security},
  volume={11},
  number={2},
  pages={221--234},
  year={2015},
  publisher={IEEE}
}

@inproceedings{fridrich2011steganalysis,
  title={Steganalysis of content-adaptive steganography in spatial domain},
  author={Fridrich, Jessica and Kodovsk{\`y}, Jan and Holub, Vojt{\v{e}}ch and Goljan, Miroslav},
  booktitle={International Workshop on Information Hiding},
  pages={102--117},
  year={2011},
  organization={Springer}
}

@article{holub2014universal,
  title={Universal distortion function for steganography in an arbitrary domain},
  author={Holub, Vojt{\v{e}}ch and Fridrich, Jessica and Denemark, Tom{\'a}{\v{s}}},
  journal={EURASIP Journal on Information Security},
  volume={2014},
  number={1},
  pages={1},
  year={2014},
  publisher={Springer}
}

@inproceedings{Holub2012WOW,
  title={Designing steganographic distortion using directional filters},
  author={Holub, Vojt{\v{e}}ch and Fridrich, Jessica},
  booktitle={2012 IEEE International workshop on information forensics and security (WIFS)},
  pages={234--239},
  year={2012},
  organization={IEEE}
}

@misc{AmnestyPegasus2021,
  title        = {Forensic Methodology Report: How to Catch {NSO} Group’s Pegasus},
  author       = {{Amnesty International Security Lab}},
  year         = {2021},
  month        = {July},
  url          = {https://www.amnesty.org/en/latest/research/2021/07/forensic-methodology-report-how-to-catch-nso-groups-pegasus/}
}

@online{OpenFAIR,
  author       = {Pankaj, Goyal and Nick, Sanna and Todd, Tucker},
  title        = {A FAIR Framework for Effective Cyber
Risk Management},
  year         = {2021},
  url          = {https://1616664.fs1.hubspotusercontent-na1.net/hubfs/1616664/FAIR%20Institute%20--%20Integrating%20FAIR%20Models%20for%20Cyber%20Risk%20Management%20(December%202024).pdf},
  note         = {FAIR defines risk via probable frequency and magnitude of future loss}
}

@standard{ISO27000,
    author    = {ISO},
  title        = {ISO/IEC 27000:2018(en)
Information technology — Security techniques — Information security management systems — Overview and vocabulary},
  number       = {ISO/IEC 27000:2018},
  organization = {International Organization for Standardization (ISO) and International Electrotechnical Commission (IEC)},
  year         = {2018},
  url          = {https://www.iso.org/obp/ui/#iso:std:iso-iec:27000:ed-5:v1:en}
}

@techreport{NIST80030,
  title        = {Guide for Conducting Risk Assessments},
  author={Joint Task Force Transformation Initiative},
  number       = {NIST SP 800-30 Rev. 1},
  institution  = {National Institute of Standards and Technology},
  address      = {Gaithersburg, MD 20899},
  year         = {2012},
  url          = {https://csrc.nist.gov/pubs/sp/800/30/r1/final}
}

@book{ross2020first,
  title={A first course in probability},
  author={Ross, Sheldon M},
  year={2020},
  publisher={Pearson Harlow, UK}
}

@book{KayDetTheory,
  author    = {Kay, Steven M.},
  title     = {Fundamentals of Statistical Signal Processing, Vol.~2: Detection Theory},
  publisher = {Prentice Hall},
  year      = {2009}
}

@article{stein1987lhs,
  title={Large sample properties of simulations using Latin hypercube sampling},
  author={Stein, Michael},
  journal={Technometrics},
  volume={29},
  number={2},
  pages={143--151},
  year={1987},
  publisher={Taylor \& Francis}
}

@techreport{RFProbLoss,
  author    = {Gundert, Levi and Ladd, Bill},
  title     = {The Probability of Loss},
  institution = {Recorded Future},
  year      = {2017},
  url       = {https://go.recordedfuture.com/hubfs/ebooks/probability-of-loss.pdf}
}

@book{FentonNeilBNs,
  title={Risk assessment and decision analysis with Bayesian networks},
  author={Fenton, Norman and Neil, Martin},
  year={2018},
  publisher={Crc Press}
}

@article{Becker1968,
  author       = {Becker, Gary S.},
  title        = {Crime and Punishment: An Economic Approach},
  journal      = {Journal of Political Economy},
  year         = {1968},
  volume       = {76},
  number       = {2},
  pages        = {169--217},
  doi          = {10.1086/259394}
}

@book{AndersonSE3,
  author       = {Anderson, Ross},
  title        = {Security Engineering: A Guide to Building Dependable Distributed Systems},
  edition      = {3},
  publisher    = {John Wiley \& Sons},
  address      = {Indianapolis, IN},
  year         = {2020},
  isbn         = {978-1-119-64278-7}
}

@incollection{IoannidisPymWilliams2009,
  author    = {Ioannidis, Christos and Pym, David and Williams, Julian},
  title     = {Investments and Trade-offs in the Economics of Information Security},
  booktitle = {Financial Cryptography and Data Security: 13th International Conference, FC 2009, Revised Selected Papers},
  editor    = {Dingledine, Roger and Golle, Philippe},
  series    = {Lecture Notes in Computer Science},
  volume    = {5628},
  pages     = {148--166},
  publisher = {Springer},
  address   = {Berlin, Heidelberg},
  year      = {2009},
  doi       = {10.1007/978-3-642-03549-4_9}
}

@article{levy2002arrow,
  title={Arrow-Pratt risk aversion, risk premium and decision weights},
  author={Levy, Haim and Levy, Moshe},
  journal={Journal of Risk and Uncertainty},
  volume={25},
  number={3},
  pages={265--290},
  year={2002},
  publisher={Springer}
}

@article{deck2024simple,
  title={A Simple Approach for Measuring Higher-Order Arrow-Pratt Coefficients of Risk Aversion},
  author={Deck, Cary and Huang, Rachel J and Tzeng, Larry Y and Zhao, Lin},
  journal={Management Science},
  year={2024},
  publisher={INFORMS}
}

@article{Romanosky2016,
  title={Examining the costs and causes of cyber incidents},
  author={Romanosky, Sasha},
  journal={Journal of Cybersecurity},
  volume={2},
  number={2},
  pages={121--135},
  year={2016},
  publisher={Oxford University Press}
}

@article{GordonLoeb2002,
  title={The economics of information security investment},
  author={Gordon, Lawrence A and Loeb, Martin P},
  journal={ACM Transactions on Information and System Security (TISSEC)},
  volume={5},
  number={4},
  pages={438--457},
  year={2002},
  publisher={ACM New York, NY, USA}
}

@incollection{Pratt1964,
  title={Risk aversion in the small and in the large},
  author={Pratt, John W},
  booktitle={Uncertainty in economics},
  pages={59--79},
  year={1978},
  publisher={Elsevier}
}

@techreport{EDPBGuidelines2023Fines,
  author      = {{European Data Protection Board}},
  title       = {Guidelines 04/2022 on the calculation of administrative fines under the GDPR},
  year        = {2023},
  institution = {EDPB},
  url         = {https://www.edpb.europa.eu/system/files/2023-06/edpb_guidelines_042022_calculationofadministrativefines_en.pdf}
}

@INPROCEEDINGS{Erokhin2023Analysis0,
  author={Erokhin, S. D. and Borisenko, B. B. and Fadeev, A. S.},
  booktitle={2023 Systems of Signal Synchronization, Generating and Processing in Telecommunications (SYNCHROINFO}, 
  title={Analysis and Development of Game-Theoretic Models for Combating Threats to Information Security in Critical Information Infrastructure}, 
  year={2023},
  volume={},
  number={},
  pages={1-7},
  doi={10.1109/SYNCHROINFO57872.2023.10178500}}

@article{basit2023dynamic,
  title={Dynamic event-triggered approach for distributed state and parameter estimation over networks subjected to deception attacks},
  author={Basit, Abdul and Tufail, Muhammad and Rehan, Muhammad and Ahn, Choon Ki},
  journal={IEEE Transactions on Signal and Information Processing over Networks},
  volume={9},
  pages={373--385},
  year={2023},
  publisher={IEEE}
}

@article{behbehani2023cloud,
  title={Cloud Enterprise Dynamic Risk Assessment (CEDRA): a dynamic risk assessment using dynamic Bayesian networks for cloud environment},
  author={Behbehani, Dawood and Komninos, Nikos and Al-Begain, Khalid and Rajarajan, Muttukrishnan},
  journal={Journal of Cloud Computing},
  volume={12},
  number={1},
  pages={79},
  year={2023},
  publisher={Springer}
}

@article{d2023including,
  title={Including insider threats into risk management through Bayesian threat graph networks},
  author={d'Ambrosio, Nicola and Perrone, Gaetano and Romano, Simon Pietro},
  journal={Computers \& Security},
  volume={133},
  pages={103410},
  year={2023},
  publisher={Elsevier}
}

@article{uflaz2024quantifying,
  title={Quantifying potential cyber-attack risks in maritime transportation under Dempster--Shafer theory FMECA and rule-based Bayesian network modelling},
  author={Uflaz, Esma and Sezer, Sukru Ilke and Tun{\c{c}}el, Ahmet Lutfi and Aydin, Muhammet and Akyuz, Emre and Arslan, Ozcan},
  journal={Reliability Engineering \& System Safety},
  volume={243},
  pages={109825},
  year={2024},
  publisher={Elsevier}
}

@book{fenton2024risk,
  title={Risk Assessment and Decision Analysis with Bayesian Networks},
  author={Fenton, N. and Neil, M.},
  isbn={9781032917917},
  year={2024},
  publisher={CRC Press}
}

@book{hubbard2023measure,
  title={How to Measure Anything in Cybersecurity Risk},
  author={Hubbard, D.W. and Seiersen, R.},
  isbn={9781119892304},
  lccn={2022061027},
  year={2023},
  publisher={Wiley}
}

@book{ross2022simulation,
  title={Simulation},
  author={Ross, Sheldon M},
  year={2022},
  publisher={academic press}
}

@article{Wei2022Application,title={Application of Bayesian Algorithm in Risk Quantification for Network Security},author={Lei Wei},journal={Computational Intelligence and Neuroscience},year={2022},volume={2022},doi={10.1155/2022/7512289}}

@article{boyd2022assumptions,
  title={Assumptions, uncertainty, and catastrophic/existential risk: National risk assessments need improved methods and stakeholder engagement},
  author={Boyd, Matt and Wilson, Nick},
  journal={Risk analysis},
  volume={43},
  number={12},
  pages={2486--2502},
  year={2023},
  publisher={Wiley Online Library}
}

@book{Tadelis2013,
author = {Tadelis, Steven},
booktitle = {Mathematics in Cyber Research},
doi = {10.1201/9780429354649-11},
isbn = {9781000542691},
pages = {416},
publisher = {Princeton University Press; Illustrated edition (6 Jan. 2013)},
title = {{Game Theory: An Introduction}},
year = {2013}
}

@inproceedings{maghrabimaeva,
  title={MAEVA: A Framework for Attack Incentive Analysis with Application to Game Theoretic Security Assessment},
  author={Maghrabi, Louai and Pfluegel, Eckhard},
  booktitle={The Sixteenth International Conference on Internet Monitoring and Protection},
  pages={31--36},
  year={2021},
  organization={ICIMP}
}

@book{pfleeger2015security,
  author    = {Charles P. Pfleeger and Shari Lawrence Pfleeger and Jonathan Margulies},
  title     = {Security in Computing},
  edition   = {5},
  year      = {2015},
  publisher = {Pearson},
  address   = {Boston, MA},
  isbn      = {978-0134085043}
}

@inproceedings{Chlosta2021,
  author = {Chlosta, Merlin and Rupprecht, David and P{\"{o}}pper, Christina and Holz, Thorsten},
  booktitle = {WiSec 2021 - Proceedings of the 14th ACM Conference on Security and Privacy in Wireless and Mobile Networks},
  doi = {10.1145/3448300.3467826},
  isbn = {9781450383493},
  pages = {359--364},
  publisher = {ACM (Association for Computing Machinery)},
  title = {5G SUCI-catchers: Still catching them all?},
  year = {2021}
}

@inproceedings{ker2013moving,
  title={Moving steganography and steganalysis from the laboratory into the real world},
  author={Ker, Andrew D and Bas, Patrick and B{\"o}hme, Rainer and Cogranne, R{\'e}mi and Craver, Scott and Filler, Tom{\'a}{\v{s}} and Fridrich, Jessica and Pevn{\`y}, Tom{\'a}{\v{s}}},
  booktitle={Proceedings of the first ACM workshop on Information hiding and multimedia security},
  pages={45--58},
  year={2013}
}

@inproceedings{SchottlePascalandLaszkaAronandJohnsonBenjaminandGrossklagsJensandBohme2013,
author = {{Sch{\"{o}}ttle, Pascal and Laszka, Aron and Johnson, Benjamin and Grossklags, Jens and Bohme}, Rainer},
booktitle = {21st European Signal Processing Conference (EUSIPCO 2013)},
pages = {1--5},
title = {{A game-theoretic analysis of content-adaptive steganography with independent embedding}},
year = {2013}
}

@inproceedings{schottle2013,
  title={A game-theoretic approach to content-adaptive steganography},
  author={Sch{\"o}ttle, Pascal and B{\"o}hme, Rainer},
  booktitle={International Workshop on Information Hiding},
  pages={125--141},
  year={2012},
  organization={Springer}
}

@article{Schottle2016a,
author = {Sch{\"{o}}ttle, Pascal and B{\"{o}}hme, Rainer},
doi = {10.1109/TIFS.2015.2509941},
issn = {15566013},
journal = {IEEE Transactions on Information Forensics and Security},
number = {4},
pages = {760--773},
title = {{Game theory and adaptive steganography}},
volume = {11},
year = {2016}
}

@article{Cox2009,
author = {{Louis Anthony and Cox}, Jr},
doi = {10.1111/j.1539-6924.2009.01247.x},
issn = {02724332},
journal = {Risk Analysis},
number = {8},
pages = {1062--1068},
publisher = {Wiley},
title = {{Game Theory and Risk Analysis}},
volume = {29},
year = {2009}
}

@article{@Manshaei2013,
author = {{Manshaei, Mohammad Hossein and Zhu, Quanyan and Alpcan, Tansu and Bac{\c{s}}ar, Tamer and Hubaux}, Jean-Pierre},
doi = {10.1145/2480741.2480742},
isbn = {8415683111},
issn = {03600300},
journal = {ACM Comput. Surv.},
number = {3},
pages = {1--45},
publisher = {ACM},
title = {{Game Theory Meets Network Security and Privacy}},
volume = {45},
year = {2013}
}

@book{LehmannRomano2022,
  author    = {Erich L. Lehmann and Joseph P. Romano},
  title     = {Testing Statistical Hypotheses},
  edition   = {4},
  year      = {2022},
  publisher = {Springer},
  address   = {Cham},
  series    = {Springer Texts in Statistics}
}

@book{@RainerBohme2010,
author = {{Rainer B{\"{o}}hme}},
isbn = {9783642022944},
pages = {288},
publisher = {Springer Science {\&} Business Media},
title = {{Advanced Statistical Steganalysis}},
year = {2010}
}

@book{@Fridrich2009,
address = {Cambridge},
author = {Fridrich, Jessica},
doi = {10.1017/CBO9781139192903},
edition = {1},
isbn = {9781139192903},
pages = {476},
publisher = {Cambridge University Press},
title = {{Steganography in Digital Media: : Principles, Algorithms, and Applications}},
year = {2009}
}

@article{osborne2004introduction,
  title={An Introduction to Game Theory},
  author={Osborne, Martin J},
  journal={Oxford University Press google schola},
  volume={2},
  pages={554 pages},
  year={2012}
}

\end{document}